\documentclass[a4paper, accepted=2024-12-12, onecolumn]{quantumarticle}
\pdfoutput=1

\usepackage[english]{babel}
\usepackage[utf8]{inputenc}
\usepackage[T1]{fontenc}

\usepackage{amsmath,physics,amsfonts}
\usepackage{graphicx,subcaption}
\usepackage{hyperref}
\usepackage{float}
\usepackage{hhline}
\usepackage{mathtools}

\mathtoolsset{showonlyrefs}

\DeclareMathOperator{\diam}{{\rm diam}}
\DeclareMathOperator{\e}{{\rm e}}
\DeclareMathOperator{\id}{\mathbf{1}}

\usepackage{amsthm}
\usepackage{authblk}

\newtheorem{theorem}{Theorem}[section]
\newtheorem{corollary}{Corollary}[theorem]
\newtheorem{lemma}[theorem]{Lemma}
\newtheorem{prop}[theorem]{Proposition}

\title{Efficient Learning of Long-Range and Equivariant Quantum Systems}
\author[1]{Štěpán Šmíd}
\author[1]{Roberto Bondesan}
\affil[1]{Department of Computing, Imperial College London, London SW7 2AZ, United Kingdom}

\begin{document}

\maketitle

\begin{abstract}
    In this work, we consider a fundamental task in quantum many-body physics -- finding and learning ground states of quantum Hamiltonians and their properties. 
    Recent works have studied the task of predicting the ground state expectation value of sums of geometrically local observables by learning from data. For short-range gapped Hamiltonians, a sample complexity that is logarithmic in the number of qubits and quasipolynomial in the error was obtained.
    Here we extend these results beyond the local requirements on both Hamiltonians and observables, motivated by the relevance of long-range interactions in molecular and atomic systems. For interactions decaying as a power law with exponent greater than twice the dimension of the system, we recover the same efficient logarithmic scaling with respect to the number of qubits, but the dependence on the error worsens to exponential. 
    Further, we show that learning algorithms equivariant under the automorphism group of the interaction hypergraph achieve a sample complexity reduction, leading in particular to a constant number of samples for learning sums of local observables in systems with periodic boundary conditions. We demonstrate the efficient scaling in practice by learning from DMRG simulations of $1$D long-range and disordered systems with up to $128$ qubits. 
    Finally, we provide an analysis of the concentration of expectation values of global observables stemming from the central limit theorem, resulting in increased prediction accuracy.
\end{abstract}

\section{Introduction}

The simulation of quantum systems underpins our understanding of nature as well as the discovery of novel critical technologies.
However, this task is notoriously hard. 
While a plethora of clever classical algorithms has been developed, such as density functional theory \cite{parr1994density}, quantum Monte Carlo \cite{RevModPhys.73.33}, and tensor networks \cite{RevModPhys.93.045003}, these methods are still too slow for many practical tasks, such as the discovery of novel catalyst for renewable energy storage \cite{zitnick2020introduction}.
The simulation of quantum systems is also one of the main use cases for quantum computers, and several quantum algorithms with provable advantages have been developed \cite{dalzell2023quantum}. 
However, despite tremendous technological advances in recent years, the current quantum computers are too small and noisy to provide a clear advantage \cite{tindall2023efficient}.

Machine learning (ML) algorithms, such as large language models like Chat-GPT, have become ubiquitous in everyday life, and recent years have also witnessed a wide adoption of ML in science, with successes such as AlphaFold \cite{jumper2021highly}.
Many researchers have also applied ML to accelerate the simulation of quantum systems, e.g.~\cite{gilmer2017neural,Carleo_2017,torlai2018neural,Carrasquilla_2020,Pfau_2020}.
One important difference between applying ML to image or text data and to quantum data is the scale at which data is available. In fact, in contrast to the abundance of images and text on the internet that has fuelled the wide scale application of deep learning to solve tasks in image and natural language processing, data from quantum experiments or simulation is hard to obtain. Experiments can be very expensive and difficult to realise, and the challenges of simulation is precisely the reason to turn to ML in a first place.
Thus, when we try to apply ML to predict  properties of quantum systems we cannot simply rely on the abundance of data, and we instead need to focus on a theoretical understanding of what tasks can be learned efficiently and what models we should use.

One recent strand of research at the intersection of quantum information and statistical learning theory aims at studying the learning complexity of quantum tasks as a new fundamental concept in quantum information \cite{anshu2023survey}.
A particularly important task is learning to predict the expectation value of observables at a new point in parameters space given a dataset of measurements of observables at other values of the parameters. 
This is in fact the setting of predicting the force field at a new position of the nuclei given the energy at a set of previous  coordinates. It is also a relevant setting in the context of variational quantum algorithms, where the learning problem is to predict the expectation value of an observable at new value of the variational parameters, which can save potentially costly experiments, as for example is used in Bayesian optimisation approaches \cite{nicoli2023physicsinformed}.

The recent works \cite{lewis2023improved,onorati2023efficient,onorati2023provably} have studied the problem of predicting the expectation value of observables within a phase of  gapped short range Hamiltonians. They derived efficient ML algorithms that provably achieve a sample complexity logarithmic in the number of qubits for observables that are given by the sum of geometrically local terms.
This is remarkable since computing properties of gapped phases include problems that are NP-hard, proving an advantage of algorithms that learn from data compared to algorithms that do not learn from data \cite{Huang_2022}. Further, in \cite{che2024exponentially} they have considered a related task of predicting arbitrary quantum states smoothly parameterised by a small constant number of parameters. Here they focused primarily on the scaling with respect to the prediction error, and achieved a polynomial sample complexity in the inverse error and the system size. Although the dependence on the prediction error is an important aspect of these learning algorithms, their methods are not applicable to systems with an extensive number of parameters like those ones considered in quantum chemistry.

However, short range interactions are often  an approximation to the long range interactions that most systems in nature experience \cite{defenu2021longrange}. Important examples of long range interacting systems are: electronic systems with Coulomb forces of relevance to quantum chemistry, dipolar molecules \cite{yan2013observation}, Rydberg atom arrays 
used to implement quantum gates
\cite{Saffman_2010}, quantum simulators based on trapped ion crystals \cite{britton2012engineered}, and 
 spin glasses \cite{RevModPhys.58.801}.

Motivated by the recent advances in understanding ML methods for short range interactions and by the importance of long range interactions in applications, in this work we study provably efficient ML algorithms for predicting properties of long range quantum Hamiltonians. 
More specifically, the relevance of this questions for the broad quantum computing field stems from the fact that quantum chemistry, which is described by long-range interacting quantum systems, is likely to be one of the earliest and most impactful applications of quantum computers.
Downstream tasks in drug design however require performing an error-corrected quantum computation 
a billion times \cite{Santagati_2024}.
ML models that predict properties of long-range interacting systems are thus critical to enable impactful applications of quantum computers to drug design.
In this work, we make progress along this line by studying extensions of the generalisation bounds previously derived in \cite{lewis2023improved,onorati2023efficient} to long-range interacting systems.
We note that \cite{Huang_2022} already sketched the extension of their results to long range interacting systems with linear light cone -- a technical condition which will be discussed at length below -- suggesting ML algorithms that use $N=m^{{\cal O}(1/\epsilon)}$ data samples, with $m$ being the number of parameters of the Hamiltonian and $\epsilon$ the accuracy.
In this work, we improve on this scaling and also introduce novel rigorous generalisation bounds for equivariant quantum systems. 
Our definition of equivariance encompasses and extends the notion of invariance, and is particularly powerful for disordered lattice systems with periodic boundary conditions.
This setup is ubiquitous in quantum and classical condensed matter physics, where disorder models impurities, defects or complex environment interactions \cite{Vojta_2019}.
In the companion paper \cite{vsmid2024accurate} we build on the results  for equivariant models we derive here to predict observables within a whole topological phase from data for a single value of the parameters.

We note that the concept of equivariance has been developed extensively in the machine learning literature to deal with symmetries in the data.
However, all the works that apply equivariant ML models to quantum physics problems are heuristic in nature.
A prominent example of those is the use of equivariant graph neural networks to predict force fields for molecular dynamics, where the symmetry is the invariance of the electronic energy under roto-translations of the nuclei \cite{satorras2022en,Batzner_2022}. 
The heuristic nature of these methods means that in practice they can suffer from uncontrollable failure modes, such as failure to generalise across different regions of the conformational space
\cite{fu2023forces}.

\pagebreak
In summary, we make the following novel contributions:
\enlargethispage{5pt}
\begin{itemize}
    \item We prove that a \textit{number of samples logarithmic in the system size} suffices to predict properties of gapped quantum systems with interactions that decay exponentially or as a power law with exponent $\alpha$ greater than twice the system's dimension. Our analysis also allows for predicting $p$-local observables which are not necessarily geometrically local and can depend on the system's parameters -- for example the ground state energy.
    \item We show that ML models that are equivariant under the automorphism group of the interaction hypergraph (i.e.~the hypergraph with vertices the qubits and hyperedges connecting qubits in the support of Hamiltonian terms) achieve a \textit{sample complexity reduction}. This implies a \textit{constant sample complexity} w.r.t. the system size for learning expectations of sums of local observables in systems with periodic boundaries. This was the first such known learning algorithm with constant sample complexity, though later \cite{wanner2024predictinggroundstateproperties} has also achieved this for short-range Hamiltonians even for open boundary conditions, but they suggest that their algorithm is not extendable to the long-range case.
    \item We lift the requirements on the observable having a norm bounded by a constant using a simple modification of the ML algorithm, which allows predictions up until the standard deviation of the expectation values is bounded by a constant. This is substantiated by an analysis of the expectation values of global observables using the central limit theorem.
    \item We provide extensive numerical simulations of one-dimensional systems with up to $128$ qubits, demonstrating the efficient scaling in practice. 
    The corresponding code for these simulations is available at \cite{Smid_Efficient_Learning_of_2024}.

\end{itemize}
In the next section we give a technical summary of our main results.

\subsection{Summary of main results}

\begin{figure}[ht]
    \centering
    \includegraphics[width=\textwidth]{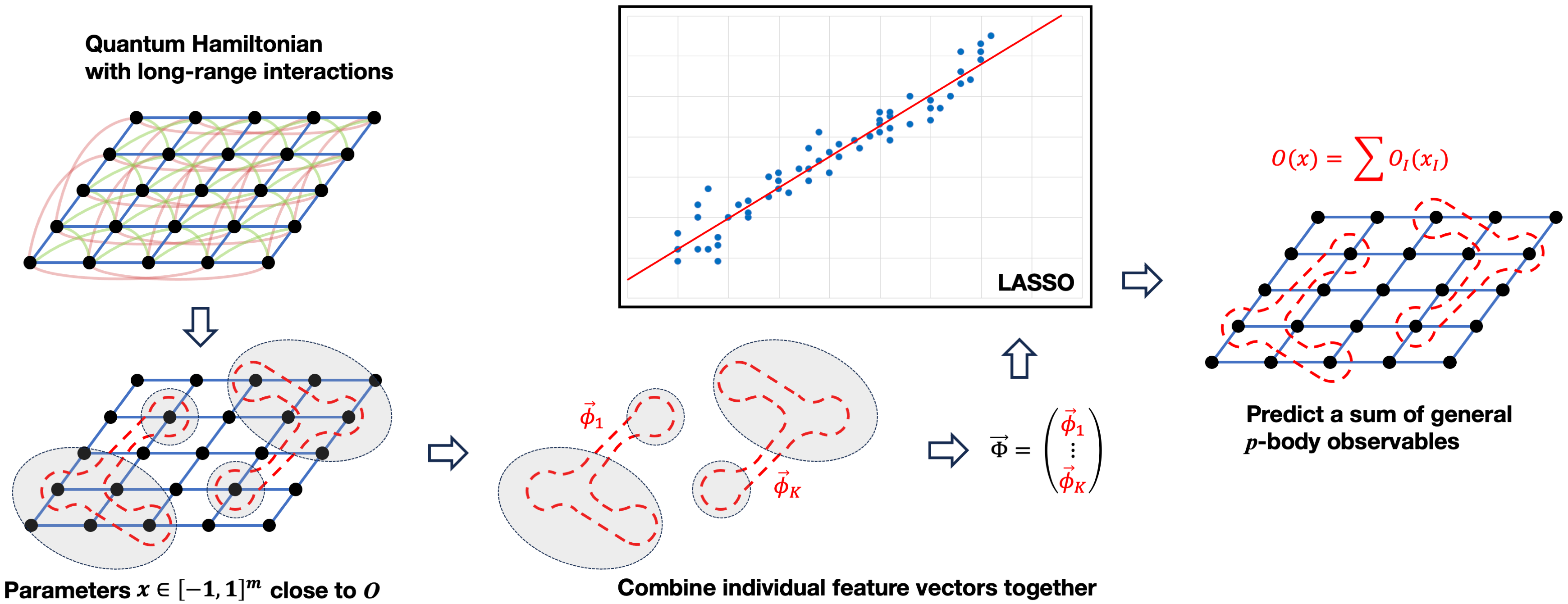}
    \caption{\textbf{Overview of the efficient machine learning algorithm.} Given a vector $x \in [-1,1]^m$ parameterising a quantum Hamiltonian with long-range interactions, it gets separated into neighbourhoods of the summands in the observable $O$, which are then mapped to  features $\vec{\phi}_i$, and concatenated into one full feature vector. This vector is then used as an input of the LASSO model, which is trained to predict $O$ for a new value of $x$.}
    \label{fig:Overview}
\end{figure}

\subsubsection{Machine learning algorithm}
The ML algorithm we study is shown in Figure \ref{fig:Overview}. It is an extension of the algorithm introduced in \cite{lewis2023improved} to the setting of long range interactions.

The input to the ML algorithm is a set of data pairs $\{ (x^{(i)},y^{(i)}) \}_{i=1}^N$ where $x\in [-1,1]^m$ is the set of parameters of the Hamiltonian and $y$ is the expectation of an observable $O$ in the ground state of the system with Hamiltonian described by $x$.
We assume that the Hamiltonian has interactions $h_{I}(x_{I})$ which are supported on a set of qubits $I=(i_1,\dots,i_\ell)$, $\ell\le k$, with $k={\cal O}(1)$, and that interactions decay exponentially with the distance between two qubits or as a power law with exponent $\alpha>2D$. 
We also consider observables $O=\sum_{I\in {\cal S}}O_I$ with ${\cal S}$ containing sets with at most $p={\cal O}(1)$ qubits, not necessarily geometrically local.
We shall normalise $O$ such that the $\ell_1$ norm of the coefficients of the expansion in the Pauli basis is ${\cal O}(1)$, which implies that $\|O\|={\cal O}(1)$.
The data $y^{(i)}$ can be either obtained by direct measurement of $O$ in a quantum experiment or a classical simulation, or by the classical shadow formalism \cite{Huang_2020} which relies on randomised measurements of $N$ ground states and postprocessing.

The ML model is built out of features defined for each operator $O_I$ as follows. 
We select the components $x_J$ of $x$ such that $J$ is within distance $\delta$ to $I$, and the diameter of $J$ is also not larger than $\delta$.
We call $Z_I(x)\in [-1,1]^{m_{I,\delta}}$ this subset of $x$.
Next we construct a feature map $\vec{\phi}_I(x)$ out of $Z_I(x)$.
As in \cite{lewis2023improved}, the theoretical analysis uses a feature map based on a discretisation, while the experiments use random Fourier features with input $Z_I(x)$.
The feature map based on discretisation is a one-hot vector of dimension equal to the number of points in a grid of mesh size $\delta_2$ in $[-1,1]^{m_{I,\delta}}$.
We set $\phi_{I,x'}(x)=1$ for the point $x'$ in this grid that is closest to $Z_I(x)$.

These features are then concatenated over $I$ to form the feature vector $\vec{\Phi}$ (bottom center of Figure \ref{fig:Overview}). A linear model with those features and a $\ell_1$ norm penalisation (LASSO) is learned from the data, and can be used to predict the ground state expectation value of this observable at new values of $x$ within the same gapped phase of the training data.

\subsubsection{Theoretical guarantees}

We prove the following main theorems about the ML model. The setting and notation is as in the previous paragraph. Specifically, we shall consider $k$-local Hamiltonians with interaction terms $h_I$ supported on sets of qubits $I$ decaying either exponentially or as a power law with exponent $\alpha$ with respect to the diameters of $I$. We then want to predict a $p$-local observable $O$ constituting of $|\mathcal{S}|$ few-body terms.
\begin{theorem}[Theorem \ref{thm:sample_complexity_single_obs} informal]\label{thm:main_thm1}
Choose
\begin{align}
    \delta 
    \ge 
    \begin{cases}
    {\cal O}(\log^2(1/\epsilon))
    & \textup{ exp decay, }\\     {\cal O}
    (\epsilon^{-1/(\nu-D)})
    & \textup{ power law, }
    \end{cases}
    \quad 
    \delta_2 = 
    {\cal O}
(\epsilon/\sqrt{m_{I,\delta}})\,,
\end{align}
with $\nu$ the function of $\alpha$ and $D$ given in Proposition \ref{prop:approx_alpha_gt_2d}, where the exponent $1/(\nu-D)$ diverges as $\alpha \to 2D$ from above.
Then the ML model with truncation parameter $\delta$ and a mesh size $\delta_2$ achieves generalisation error at most $\epsilon$ with probability at least $1-\gamma$ if we choose the number of samples as
\begin{align}
    N 
    =
    \mathcal{O}(\log(|{\cal S}|/\gamma))    
    \times 
    \begin{cases}
    2^{\mathcal{O}(\operatorname{polylog}(1/\epsilon))} & \textup{ exp decay},\\
    2^{\mathcal{O}(\epsilon^{-\omega} \log(1/\epsilon))} & \textup{ power law, }
    \end{cases}
    \quad \omega = kD/(\nu-D)\,.
\end{align}
\end{theorem}
We remark that if $|{\cal S}|={\cal O}(1)$, the number of samples does not depend on $n$, while if $|{\cal S}|=n^p$, with $p={\cal O}(1)$, the number of samples grows only logarithmically with $n$.
The dependence on the error $\epsilon$ is quasipolynomial in the case of exponentially decaying interactions, which as expected coincides with the behavior for short range interactions \cite{lewis2023improved,onorati2023efficient}.
This dependence becomes instead exponential for power law interactions with $\alpha>2D$.
The exponent $1/(\nu-D)\to \infty$ as $\alpha\to 2D$ and our results break down for $\alpha\le 2D$.

\begin{proof}[Proof strategy]
    There are two main parts of proving this theorem. 
    
    The first part requires to show that ground state expectation values of few-body non-geometrically-local observables can be accurately approximated only with the parameters in the neighbourhood of their support, where the size of this neighbourhood does not depend on the system size. As we will show, this property will hold true for Hamiltonians with interactions decaying at least as a power law with exponent greater than twice the system's dimension. This property will further hold true even for observables that depend themselves on the parameters of the system. To prove this, one needs to consider the quasi-adiabatic continuation operator \cite{hastings2010locality}, which quantifies the change of the expectation values with respect to the parameters within a topological phase. The action of this operator can then be bounded using Lieb-Robinson bounds \cite{Tran_2021}. Hence one can bound the difference between the expectation values of the full set of parameters and the truncated set of parameters on the neighbourhood to be arbitrarily small for appropriately large neighbourhood, whose size does not depend on the number of qubits.

    The second part then consists of showing that an expectation value of a $p$-local observable constituting of few-body terms obeying the previously proven property can be accurately expressed across the parameter space as a simple linear function, where the unknown coefficients are to be learned, while the features are based on a discretisation grid of the parameter space. The in-built local geometric bias in these features combined with the previously proven property is what allows this model to be accurate. Finally, we would use well-understood sample complexity theory for linear models \cite{mohri2018foundations} to determine how many training samples we would need to learn the unknown coefficients and obtain a constant generalisation error.
\end{proof}

In Theorem \ref{thm:sample_complexity_classical_shadows} we also show that we can learn all observables 
that are sums of at most $p$-body terms and have bounded $\ell_1$ norm in the Pauli basis with a sample complexity as in Theorem \ref{thm:main_thm1}, where  $|{\cal S}|$ is replaced by the cardinality of the set of Pauli operators with weight at most $p$.
In this case we need to prepare $N$ ground states, one per parameter $x^{(i)}$, and then perform  ${\cal O}(\log(n/\delta)/\epsilon^2)$ randomised measurements to compute the classical shadow of the density matrix \cite{Huang_2020}, which allows one to produce with accuracy $\epsilon$ a dataset $\{x^{(i)}, y^{(i), P} \}_{i=1}^N$ for each Pauli $P$. This data is then used to learn a ML model for each observable using Theorem \ref{thm:main_thm1}.

Next we present our results on equivariance.
We define the interaction hypergraph as the hypergraph with vertex set the vertices of the lattice and one hyperedge per interaction.
Its automorphism group $G$ is a subgroup of the permutation group of the qubits that fixes the hyperedge set. It relates the predictions for the observable $O_I$ to that of $O_{gI}$, with $g\in G$, and allows one to reduce the sample complexity by using an equivariant ML model.
The precise definition of equivariance and equivariant weights are in Propositions \ref{prop:equi_ml_model} and \ref{prop:equi_rff}.
\begin{corollary}[Corollary \ref{cor:sample_complexity_equiv}]
    A $G$-equivariant ML model achieves a sample complexity reduction that amounts to replace $|{\cal S}|$ with $|{\cal S}/G|$ in Theorem \ref{thm:main_thm1}.
\end{corollary}
This result is particularly powerful if $|{\cal S}|\approx |G|$, in which case the complexity becomes independent of $n$. This is the case for example for the frequently encountered case of models defined on a cubic lattice with periodic boundary conditions and observables with $|{\cal S}|={\cal O}(n)$. 
Note that equivariance reduction applies independently of the details of the Hamiltonian and in particular does not require the Hamiltonian to be translation invariant in the case of periodic boundary conditions.
A similar sample complexity reduction from equivariance is obtained in Corollary \ref{cor:equiv_classical_shadows} for predicting expectation values of all observables that are sums of $p$-body terms.

Physically, we thus expect our results on equivariant models to be most relevant for quantum disordered systems on a lattice, such as
the paradigmatic disordered Heisenberg and Ising models \cite{Luitz_2015}, spin glasses \cite{nishimori2001statistical}, and atomic chains in the presence of impurities, such as Rydberg atoms \cite{Marcuzzi2017}; some of which are discussed further in the section on experiments.

While we assume here exact equivariance of the Hamiltonian, we expect our results to be relevant in the case of approximate equivariance as well.
That is, if the Hamiltonian under consideration is close to that of an equivariant system, in the limit of large number of qubits the equivariant ML models should provide a good description of the observables.
This applies for example to the case where the lattice has open boundary conditions, as we expect that in the thermodynamic limit the system has the same bulk observables as the same system with periodic boundaries.

\subsubsection{Experiments} Another contribution of this paper is to verify numerically the analytical results. For the first time, we provide numerical evidence for the logarithmic scaling of the sample complexity with the number of qubits.
The data is produced by DMRG simulations of disordered short range Heisenberg chains, disordered Ising chains 
with spin-spin interactions decaying as a power law with exponent $\alpha$, and Rydberg atom chains with position displacements with interactions decaying as the $6^\text{th}$ power.
We predict the expectation value of the Hamiltonian in the ground state across a range of couplings.
We point out that due to the short range correlations in the systems, the central limit theorem applies. A naïve normalisation of the observable has an expectation value that concentrates around the average over disorder in the large $n$ limit, leading to an $x$-independent target for the machine learning model and a trivial sample complexity -- a constant model with a single parameter can achieve zero test error in the $n\to\infty$ limit.
We show that scaling and zero-centering the observable so that the Gaussian fluctuations in the large $n$ limit are not suppressed, leads to a non-trivial ML problem also as $n\to\infty$, and this allows us to verify the non-trivial scaling of the sample complexity, see Figures \ref{fig:HeisenbergOpen}, \ref{fig:IsingOpen}, and \ref{fig:Plot-Rydberg}.
We also verify that using an equivariant ML model allows us to reduce the sample complexity from ${\cal O}(\log(n))$ to ${\cal O}(1)$ in Figure \ref{fig:HeisenbergPeriodicRMSE} for the disordered Heisenberg model on a closed chain.
Finally, in Figure \ref{fig:IsingBadAlpha} we show that we can apply successfully the ML derived for $\alpha>2D$ to a case with $\alpha\le 2D$. 
While the numerical data suggest worse, seemingly linear, scaling, we can not definitely conclude the sample complexity to not be logarithmic in the number of qubits. We leave it as an important open problem to derive theoretical guarantees for $\alpha\le 2D$.

\subsection{Outline of the paper}

The paper is structured as follows:
Section \ref{sec:main - Quantum info} presents the results and comments on the proofs that the expectation values of observables $O_I$ depend only on Hamiltonian parameters that are geometrically close to $I$; while Appendix \ref{sec:Theoretical guarantees for approximating observables} presents the detailed proof. This is shown for gapped Hamiltonians with exponentially decaying interactions and power law decaying interactions with exponent $\alpha>2D$.
These sections constitute the bulk of the novel technical contributions of this paper.
Section \ref{sec:main - Machine learning model} then builds on these results and the literature on generalisation bounds to derive rigorous guarantees for the sample complexity of the ML models; together with details in Appendix \ref{sec:Machine learning model and sample complexity bounds} about predicting many observables via classical shadows.
In Section \ref{sec:main - Machine learning model} we also introduce the notion of equivariance and show how it allows one to reduce the sample complexity.
Finally, Section \ref{sec:Experiments}
shows numerical experiments that validate the theoretical findings.
We discuss the normalisation of observables and the details of the experimental setup. We also apply the ML model to Hamiltonians with $\alpha\le 2D$.
Appendix \ref{sec:Technical Lemmas} contains combinatorial identities and bounds used to derive results in Section \ref{sec:main - Quantum info}, Appendix \ref{sec:Lieb Robinson bounds for power law interactions} discusses Lieb-Robinson bounds for power law interactions, and Appendix \ref{sec: Implementation MPO} the DMRG implementations of the considered models.

\section{Quantum information bounds}
\label{sec:main - Quantum info}

\subsection{Setup and notation}
\label{sec:Setup and notation}

We consider a $D$-dimensional lattice with vertex set $\Lambda$ of cardinality $n$. We denote by ${\cal P}(\Lambda)$ the set of subsets of $\Lambda$, and by ${\cal P}_k(\Lambda)$ those $I\in {\cal P}(\Lambda)$ with cardinality $|I|\le k$. 
We denote by $d(i,j)$ the distance between $i,j\in \Lambda$ and
for $I,J \in {\cal P}(\Lambda)$ we define their distance and diameter as
\begin{align}
    d(I,J)=\min_{i\in I, j\in J}d(i,j)
    \,,\quad
    \diam(I) = 1+\max_{i,j\in I}d(i,j)
    \,.
\end{align}

Now we associate a system of qubits to the vertices in $\Lambda$ and a Hamiltonian with $k$-body long range interactions which depend on a set of parameters $x$ as follows:
\begin{align}
    H(x)
    =
    \sum_{I\in {\cal P}_k(\Lambda)}
    h_I(x_I)
    \,,
\end{align}
where $h_I(x_I)$ is supported only on the qubits in $I$ and $x_I\in [-1,1]^{q_I}$.
We denote the dimensionality of $x$ by
\begin{align}
    m
    =
    \sum_{I\in {\cal P}_k(\Lambda)}
    q_I
    \le q_* |{\cal P}_k(\Lambda)|
    =
    q_*
    \sum_{\ell=1}^k \binom{n}{\ell}
    \,,\quad 
    q_*=
    \max_{I\in {\cal P}_k(\Lambda)}
    q_I
    \,.
\end{align}
For example, when $q_I=1$ and $k=2$, we have the family of Hamiltonians
\begin{align}
    H(x)
    =
    \sum_{ij} h_{ij}(x_{ij})
    +
    \sum_{i} h_{i}(x_{i})
    \,.
\end{align}
The case of interactions with short range $r$
can be seen as a special case where $h_{I}=0$ if $\diam(I)>r$.

We are going to study the following machine learning problem. 
Let us denote the ground state of $H(x)$ by $\rho(x)$, which is defined as
\begin{align}
    \rho(x) = \lim_{\beta\to\infty}
    \frac{\e^{-\beta H(x)}}{\Tr(\e^{-\beta H(x)})}\,.
\end{align}
We are given a dataset 
\begin{align}
    S
    =
    \{x^{(i)}, 
    y^{(i)}\}_{i=1}^N\,,\quad
    y^{(i)}=
    \Tr(\rho(x^{(i)}) O)
    +\epsilon
    \,,
\end{align}
where $\epsilon\ge 0$ accounts for errors in measuring the observable.
We then aim at constructing a predictor for $\Tr(\rho(x_*) O)$ at a test point $x_*$, and we ask how many samples do we need to solve this problem accurately.

Unless explicitly stated, we will always consider the operator norm $\|\cdot \|$ over Hermitian operators, defined as the absolute value of the largest eigenvalue. If $A=(A_1,\dots,A_p)$ is a vector of operators, then we denote $\| A \| = \sqrt{\sum_{i=1}^p \|A_i\|^2}$.

We are going to solve this ML problem assuming that:
\begin{itemize}
    \item $k, q_I$ do not depend on $n$.
    \item There is a gap $\gamma$ above the ground state that is uniform over all $x\in [-1,1]^{m}$.
    \item $\| h_I(x_I) \|$ and $\| \partial_{x_i}h_I(x_I) \|$ 
    decay exponentially or as a power law with $\diam(I)$ for all $x_I$. More precise bounds will be discussed below.
    \item The observable $O$ is such that, with $p=\mathcal{O}(1)$:
    \begin{align}
        O = \sum_{I\in {\cal P}_p(\Lambda)}O_I\,.
    \end{align}
\end{itemize}
This is similar to the setup of \cite{lewis2023improved,onorati2023efficient,onorati2023provably}; however, differently from those works, we do not assume that neither $h_I$ nor $O_I$ are geometrically-local.
In the following we will denote by $\id(\cdot)$ the indicator function and by $\delta(\cdot ,\cdot )$ the Kronecker delta.

\subsection{Dependency of observables on Hamiltonian parameters}
\label{sec:Dependency of observables on Hamiltonian parameters}

The goal of this section is to prove that under the assumptions of Section \ref{sec:Setup and notation} the function
\begin{align}
    f(O_I, x)
    =
    \Tr(O_I \rho(x))\,,\quad 
    I\in {\cal P}_p(\Lambda)\,,
\end{align}
depends only on the $x_J$'s in a neighborhood of $I$.
This neighborhood is defined as the $x_J$'s that are within a distance $\delta>0$ from $I$ and are such that $\diam(J)\le \delta$:
\begin{align}
    \label{eq:S_I_delta}
    S_{I,\delta}
    =
    \{
    J\in {\cal P}_k(\Lambda)
    \,|\,
    d(I,J)\le \delta
    \wedge
    \diam(J)\le \delta
    \}\,.
\end{align}
Figure \ref{fig:Sdelta} visualizes $S_{I,\delta}$ for a case in $D=2$.
Note that the complementary set is
\begin{align}
    (S_{I,\delta})^c
    =
    \{
    J\in {\cal P}_k(\Lambda)
    \,|\,
    \left(
    d(I,J)\le \delta
    \wedge
    \diam(J)> \delta
    \right)
    \vee
    d(I,J)>\delta
    \}\,.
\end{align}
\begin{figure}[ht!]
    \centering
    \includegraphics[width=.75\textwidth]{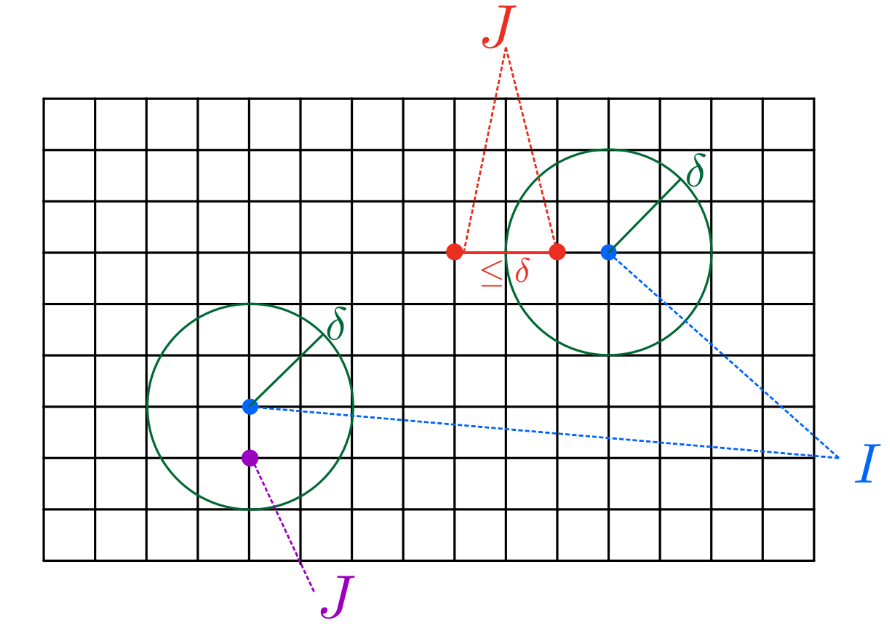}
    \caption{$S_{I,\delta}$ for two choices of $J$, once with cardinality $1$ (purple) and once with $2$ (red).
    Here $D=|I|=2$.
    }
    \label{fig:Sdelta}
\end{figure}

We now define $x'$ such that
\begin{align}
    x'_J
    =
    \begin{cases}
        x_J & \text{if } J \in S_{I,\delta}\,,\\
        0   & \text{if } J \not\in S_{I,\delta}\,.\\
    \end{cases}
\end{align}
There is nothing special about setting $x'$ to $0$ outside $S_{I,\delta}$, and any other fixed value could have been chosen for the results below to follow. We shall also denote by $\chi_S(x)$ the vector $x$ that is set to zero for $J\in S$.
Now we will study under what conditions $f(O_I, x)$ can be approximated by $f(O_I, x')$. First, following \cite{lewis2023improved, onorati2023efficient} we can rewrite this problem as that of bounding the gradient of $f$. Let $x(s) = x's+x(1-s)$, then:
\begin{align}
    &|f(O_I, x')-f(O_I, x)|
    =
    \left|
    \int_0^1\dd s
    \partial_s f(O_I, x(s))
    \right|
    \le
    \int_0^1\dd s
    |\partial_s f(O_I, x(s))|
    \\
    &\le
    \int_0^1\dd s
    \sum_{J\in {\cal P}_k(\Lambda)}
    |(x'_{J}-x_{J})\cdot
    \nabla_{J} f(O_I, x(s))|
    \le 
    q_*
    \int_0^1\dd s
    \sum_{J\not \in S_{I,\delta}}
    \| \nabla_{J} f(O_I, x(s)) \|\\
    &\le 
    q_*
    \int_0^1\dd s
    \sum_{J\in {\cal P}_k(\Lambda)}
    \left( 
    \id\left(d(I,J)\le \delta
    \wedge
    \diam(J)> \delta
    \right)
    +
    \id\left(
    d(I,J)>\delta\right)
    \right)
    \| \nabla_{J} f(O_I, x(s)) \|\,,
\end{align}
where we used the Cauchy–Schwarz inequality and the penultimate inequality follows from $\|x_J\|\le q_J\le q_*$.

Then recall the quasiadiabatic operator or spectral flow.
\begin{lemma}[Corollary 2.8 of \cite{Bachmann_2011}]\label{lemma:quasiadiabatic}
Let $H(s)$ be a Hamiltonian such that 1) $\|H'(s)\|$ is uniformly bounded in $s\in[0,1]$ and 2) the spectrum of $H(s)$ can be decomposed in two parts separated by a spectral gap $\gamma$. Then if $P(s)$ is the projector onto the lower part, it satisfies
\begin{align}
    \frac{\dd}{\dd s}P(s) = i[D(s),P(s)]\,,\quad
    D(s)=\int_{-\infty}^{+\infty}\dd t W_\gamma(t)\e^{itH(s)}H'(s)
    \e^{-itH(s)}
    \,,
\end{align}
where the function $W_\gamma(t)$ satisfies 
$W_\gamma(t) = -W_\gamma(-t)$
and $\sup_t W_\gamma(t) = \frac{1}{2}$.
\end{lemma}

Assuming $O_I$ does not depend on $x$ -- we shall discuss in Section \ref{sec:Approximation of the energy} the extension to the case of $O_I$ depending on $x$ -- we have, denoted $\tau_t(A)=\e^{itH(x)}A\e^{-itH(x)}$,
\begin{align}    
    \label{eq:nabla_f}
    &\|\nabla_J
    f(O_I,x)
    \|
    =
    \|
    \Tr(O_I [D_{J}(x),\rho(x)])
    \|
    =
    \|
    \Tr([O_I, D_{J}(x)]\rho(x))
    \|
    \\
    &\le
    \int_{-\infty}^{+\infty}\dd t |W_\gamma(t)|\,
    \|
    \Tr([\tau_t(\nabla_{J}H(x)), O_I] \rho(x))
    \|\le
    \int_{-\infty}^{+\infty}\dd t |W_\gamma(t)|\,
    \|
    [\tau_t(\nabla_{J}h_J(x_J)), O_I] 
    \|
    \,,
\end{align}
where the last inequality follows from Von Neumann's trace inequality: 
if $A,B$ have eigenvalues $\alpha_i, \beta_i$, then
$\Tr(AB)\le \sum_i \alpha_i\beta_i\le \max_i \alpha_i \sum_j \beta_j$.

As next step we are going to bound the quantity 
$\|[\tau_t(\nabla_{J}h_J(x_J)), O_I]  \|$.
Let us assume that $R=d(I,J)>0$.
We introduce an invertible function of $R$, $t_*=t_*(R)$, that we shall determine later, depending on the exact decay of the interactions. 
We distinguish two regimes: one for $|t|<t_*$, where
the commutator is small since the operator $\tau_t(\nabla_{J}h_J(x_J))$ has not spread enough in the region $I$. In this regime we can use a Lieb-Robinson bound \cite{hastings2010locality}. The second regime is $|t|>t_*$, where the Lieb-Robinson bound is vacuous and we can use instead the trivial bound:
\begin{align}
    \|
    [\tau_t(\nabla_{J}h_J(x_J)), O_I] 
    \|
    \le 
    2
    \|
    \tau_t(\nabla_{J}h_J(x_J)) 
    \|
    \,
    \|
    O_I
    \|
    =
    2
    \|
    \nabla_{J}h_J(x_J)
    \|
    \,
    \|
    O_I 
    \|
    \,.
\end{align}
Then we break down the quantity to be estimated as the sum of two terms 
\begin{align}
    &\int_{-t_*}^{+t_*}\dd t |W_\gamma(t)|\,
    \|
    [\tau_t(\nabla_{J}h_J(x_J)), O_I] 
    \|
    \le
    \|
    \nabla_{J}h_J(x_J)
    \|
    \,
    \|
    O_I 
    \|
    \underbrace{
    \frac{\int_{0}^{t_*}\dd t 
    \|
    [\tau_t(\nabla_{J}h_J(x_J)), O_I] 
    \|}{\|
    \nabla_{J}h_J(x_J)
    \|
    \,
    \|
    O_I 
    \|}}_{{\cal I}_1}
    \,,\\
    &
    \left(
    \int_{-\infty}^{-t_*}
    +
    \int_{t_*}^{+\infty}
    \right)
    \dd t |W_\gamma(t)|\,
    \|
    [\tau_t(\nabla_{J}h_J(x_J)), O_I] 
    \|
    \le
    \|
    \nabla_{J}h_J(x_J)
    \|
    \,
    \|
    O_I 
    \|
    \underbrace{
    4\int_{t_*}^{+\infty}\dd t
    |W_\gamma(t)|
    }_{
    {\cal I}_2
    }
    \,,    
\end{align}
so that
\begin{align}
\label{eq:I1_I2}
    \|\nabla_J
    f(O_I,x)
    \|
    \le 
    \|
    \nabla_{J}h_J(x_J)
    \|
    \,
    \|
    O_I 
    \|
    (
    {\cal I}_1
    +
    {\cal I}_2
    )\,.
\end{align}
We can bound ${\cal I}_2$ using the following result on the function $W_\gamma$:
\begin{lemma}[Lemma 2.6 (iv) of \cite{Bachmann_2011}]\label{lemma:2.6}
For $t>0$, let
\begin{align}
    I_\gamma(t)
    =
    \int_t^{\infty}\dd \xi 
    |W_\gamma(\xi)|\,.
\end{align}
Then $I_\gamma(t)\le G(\gamma|t|)$, where
\begin{align}
    G(\xi)
    =
    \frac{1}{\gamma}
    \begin{cases}
        \frac{K}{2}& 0\le \xi \le \xi_*\,,\\
        130\e^2\xi^{10}u_{2/7}(\xi)& 
        \xi >\xi_*\,,
    \end{cases}
    \quad
u_a(\xi) = \exp(-a\tfrac{\xi}{\log^2(\xi)})
\,,
\end{align}
with $K\sim 14708$ and $ 36057<\xi_*<36058$.
\end{lemma}
Thus ${\cal I}_2 \le 4 G(\gamma t_*)$.
If we assume that
\begin{align}
    \gamma t_*(d(I,J)) > \xi_*
    \Rightarrow
    d(I,J)> t_*^{-1}(\gamma^{-1}\xi_*)\,,
\end{align}
then
\begin{align}
    \label{eq:I_2_gamma_t_star}
    {\cal I}_2
    \le
    \frac{4}{\gamma}
    130\e^2 (\gamma t_*)^{10}u_{2/7}(\gamma t_*)
    \,.
\end{align}
Next, we shall discuss different decays of the long range interactions which lead to different choices of $t_*$ and different bounds for ${\cal I}_1$ and thus of $|f(O_I, x')-f(O_I, x)|$.
We start with exponential decay in Section \ref{sec:main - Exponential decay} and then move to power law in Sections \ref{sec:main - Power law > 2D} -- \ref{sec:main - power law alpha gt 2d plus one}.

\subsubsection{Exponential decay}
\label{sec:main - Exponential decay}

In this section we shall present the results and comment on their proofs for Hamiltonians with interactions decaying exponentially with the distance, while the full derivation is presented in \ref{sec:Exponential decay}. We wish to prove the following proposition:
\begin{prop}[Corollary \ref{cor:delta_eps_exp}]
    Assume the Hamiltonian has interactions decaying as \begin{align}
    \|h_I\|
    \le
    \e^{-\nu \diam(I)}
    P(\diam(I))
    \,, \quad     
    \|\nabla_I h_I\|
    \le
    \e^{-\nu \diam(I)}
    \,,    
\end{align} where $\nu$ is a positive constant and $P$ a positive polynomial. Then we have \begin{align}
    |f(O_I,x')-f(O_I,x)| \leq \epsilon \| O_I \|,
\end{align} where $x' = \chi_{S_{I,\delta}}(x)$ is the truncated vector of parameters, if $0 < \epsilon \leq e^{-1}$ and $\delta \geq c \log^2(\tilde{c}/\epsilon)$, with $c$ and $\tilde c$ being constants specified in Appendix \ref{sec:Exponential decay}.
\end{prop} 
This means that we can approximate $f(O_I,x)$ with a function $f(O_I,\chi_{S_{I,\delta}}(x))$ that is supported only on $\chi_{S_{I,\delta}}(x)$ if we choose $\delta$ as above.

The proof of this result is similar to the case of short range interactions \cite{lewis2023improved,onorati2023efficient} and relies on Lieb-Robinson bounds for exponentially decaying interactions \cite{hastings2010locality} to bound $\mathcal{I}_1$
appearing in \eqref{eq:I1_I2}.
There are however two main differences with respect to the case of short range interactions discussed previously in the literature.
First, we consider a general observable supported on the qubits $I$, while previous works consider only a geometrically local observable. 
Second, as discussed in Section \ref{sec:Dependency of observables on Hamiltonian parameters},  bounding $|f(O_I,x')-f(O_I,x)|$ requires to bound contributions from Hamiltonian terms $J$ such that $d(I,J)\le \delta$ and $\diam(J)>\delta$, which are not present in the short range case.
The bulk of the proof develops combinatorial arguments to make sure that these terms do not change the asymptotic behavior of the short range case.

\subsubsection{Power law with $\alpha\in (2D,2D+1)$}
\label{sec:main - Power law > 2D}

In this section we shall present the results and comment on their proofs for Hamiltonians with interactions decaying polynomially with the distance, while the full derivation is presented in \ref{sec:power law alpha gt 2d}. We wish to prove the following proposition:

\begin{prop}[Corollary \ref{cor:delta_eps_alpha_gt_2D}]
Let $\alpha\in (2D,2D+1)$ and assume that the Hamiltonian has interactions decaying as follows: for all $I$ such that $|I|>1$,
    \begin{align} \|h_I\|
    \le
    g_{|I|}\diam(I)^{-\gamma(|I|)}
    \,,  \quad
    \|\nabla_I h_I\|
    \le
    \diam(I)^{-\gamma(|I|)}
    \,,    \quad 
    \gamma(\ell) = (\ell-2)D+\alpha\,,
    \end{align} 
for some positive constants $g_{\ell}$. Then we have \begin{align}
    |f(O_I,x')-f(O_I,x)| \leq \epsilon \| O_I \|,\end{align} where $x' = \chi_{S_{I,\delta}}(x)$, with $\delta$ now being bounded like 
    $\delta \ge \tilde{c}\max(
    \epsilon^{-1/(\nu-D)},
    \log^{2/\mu}(3C_2/\epsilon))$, with the constants being specified in Appendix \ref{sec:power law alpha gt 2d}.

\end{prop}

Power law interactions with $\alpha\in(2D,2D+1)$ introduce two new features in the scaling of $\delta$ with $\epsilon$ compared to exponential decay: a polynomial scaling with $1/\epsilon$ with exponent $1/(\nu-D)$ and the logarithmic dependency $\log^{2}(1/\epsilon)$ gets replaced with $\log^{2/\mu}(1/\epsilon)$.
We show in Figure \ref{fig:exponents_epsilon} that both $1/(\nu-D)$ and $1/\mu$ are monotonically decreasing with $\alpha$, and as $\alpha\to 2D$ both exponents diverge.

\begin{figure}[ht!]
\centering
    \begin{subfigure}[t]{0.5\textwidth}
        \centering
        \includegraphics[width=.75\textwidth]{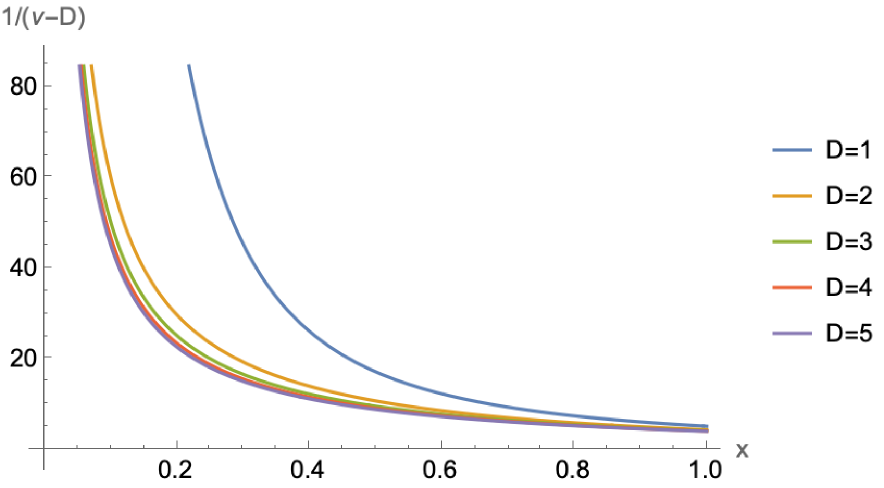}
        \caption{}
    \end{subfigure}%
    ~ 
    \begin{subfigure}[t]{0.5\textwidth}
        \centering
        \includegraphics[width=.75\textwidth]{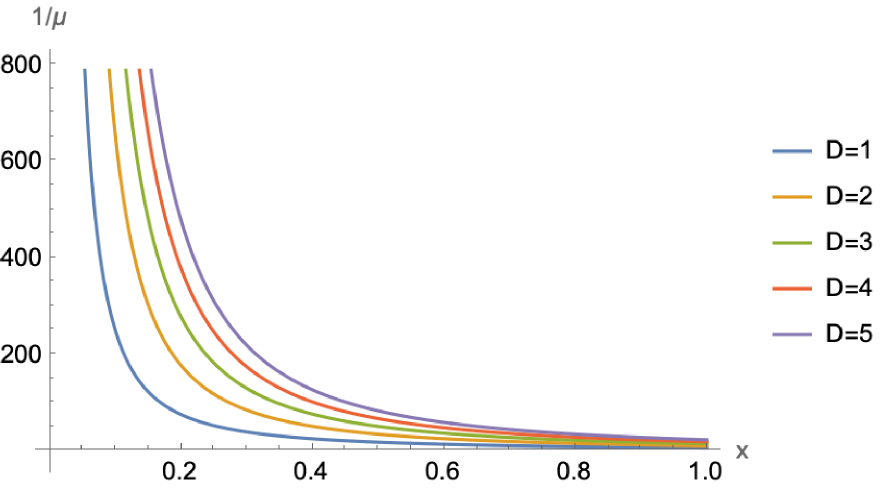}
        \caption{}
    \end{subfigure}
    \caption{
    \textbf{(a)} $1/(\nu-D)$ as a function of $x=\alpha-2D$ for $D=1,\dots,5$.
    \textbf{(b)} $1/\mu$ as a function of $x=\alpha-2D$ for $D=1,\dots,5$.
    }
    \label{fig:exponents_epsilon}
\end{figure}

The proof of this result is similar to that of the case of exponentially decaying interactions and relies on previously derived Lieb Robinson bounds for systems with interactions decaying as a power law \cite{Tran_2021}.
The light cone, which is linear in the case of exponentially decaying or short range interactions, is algebraic for power law interactions.
The core of the proof generalises the bounds to series involving exponentially decaying functions of Appendix \ref{sec:Exponential decay} to the case of power law decay. The summability conditions of the series result in the restrictions to the power law exponent.

\subsubsection{Power law with $\alpha>2D+1$}
\label{sec:main - power law alpha gt 2d plus one}

We now note that the results of Section \ref{sec:main - Power law > 2D} can be applied to Hamiltonians with $\alpha>2D+1$ as well. Indeed if the interactions in the Hamiltonian decay as $1/r^{\alpha}$ with $\alpha>2D+1$ we can use that for $r\ge 1$, $1/r^{\alpha}<1/r^{2D+1-\zeta}$ for any $\zeta>0$ and apply Proposition \ref{prop:approx_alpha_gt_2d} and Corollary \ref{cor:delta_eps_alpha_gt_2D} with $\alpha=2D+1-\zeta$, $\zeta\in (0,1)$:
\begin{align}
    \frac{|f(O_I, x')-f(O_I, x)|}{\|O_I\|}\le \epsilon\,,
\end{align}
if 
\begin{align}
    \delta &\ge \tilde{c}\max(
    \epsilon^{-1/(\nu-D)},
    \log^{2/\mu}(3C_2/\epsilon))|_{\alpha=2D+1-\zeta}\,,\\
    \tilde{c}&=\max((3C_1)^{1/(\nu-D)},(c/b)^{1/\mu})|_{\alpha=2D+1-\zeta}\,.
\end{align}
In Figure \ref{fig:exponents_epsilon_2D_1} we plot the exponents $1/(\nu-D)$ and $1/\mu$ at $\alpha=2D+1-\zeta$ with $\zeta=10^{-12}$.

\begin{figure}[ht!]
\centering
    \begin{subfigure}[t]{0.5\textwidth}
        \centering
        \includegraphics[width=.75\textwidth]{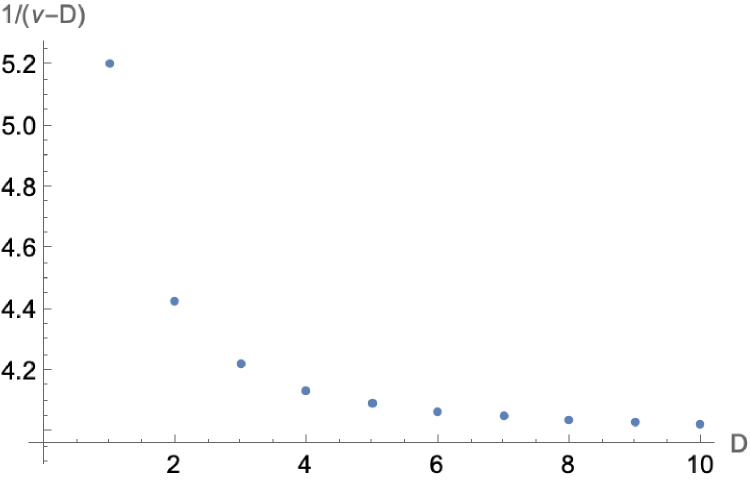}
        \caption{}
    \end{subfigure}%
    ~ 
    \begin{subfigure}[t]{0.5\textwidth}
        \centering
        \includegraphics[width=.75\textwidth]{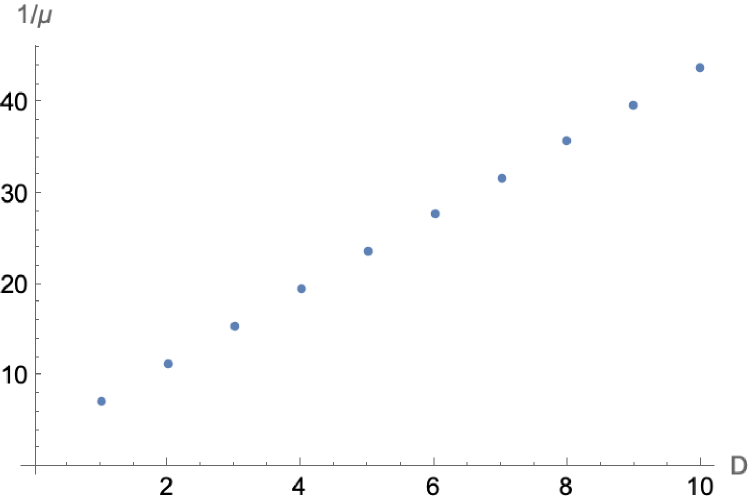}
        \caption{}
    \end{subfigure}
    \caption{
    \textbf{(a)} $1/(\nu-D)$ as a function of $D$ at $\alpha=2D+1-\zeta$ with $\zeta=10^{-12}$.
    \textbf{(b)} $1/\mu$ as a function of $D$ at $\alpha=2D+1-\zeta$ with $\zeta=10^{-12}$.
    }
    \label{fig:exponents_epsilon_2D_1}
\end{figure}

So we can extend the results derived for $\alpha\in(2D,2D+1)$ to all $\alpha>2D$. 
In the following, when we use Corollary \ref{cor:delta_eps_alpha_gt_2D} for all $\alpha>2D$, we will understand that if $\alpha>2D+1$ we should replace $\nu,\mu$ with their value at $2D+1-\zeta$.
A more refined analysis for $\alpha>2D+1$ with improved scaling is possible if we use the stronger Lieb-Robinson bounds in this regime \cite{Kuwahara_2020}. In this case, the dependency on $\alpha$ is such that we recover the short range result as $\alpha\to\infty$.

\subsection{Approximation of general $p$-body observables}

Now we consider general observables of the form
\begin{align}
    O = \sum_{I\in {\cal P}_p(\lambda)} O_I\,,
\end{align}
where $O_I$ is supported on the qubits $I$ as before, and we assume $p=\mathcal{O}(1)$.
Then, if we choose $\delta$ as in corollary
\ref{cor:delta_eps_exp} and \ref{cor:delta_eps_alpha_gt_2D} respectively for exponential and power law decay interactions with exponent $\alpha>2D$, we have
\begin{align}
    &|\Tr(O\rho(x)) - \sum_{I\in {\cal P}_p(\Lambda)}f(O_I,\chi_{S_{I,\delta}}(x))|
    \le
    r(O)
    \epsilon\,,\\
    \label{eq:rO}
    &r(O) = \sum_{I\in {\cal P}_p(\Lambda)} \|O_I\| \,.
\end{align}
So we can approximate general $p$-body observables $O$ provided $r(O)$ is small.
If we decompose $O$ as a sum over Paulis, $O=\sum_{P} \alpha_P P$, then we can bound the difference with $\sum_{P} \alpha_P f(P,\chi_{\text{dom}(P), \delta})$ by $\sum_{P} |\alpha_P| \epsilon$.
In \cite{lewis2023improved} it was shown that if  $\alpha_P$ is non-zero only for geometrically local $P$, then one can bound $\sum_{P}|\alpha_P|\le C \|O\|$.
In the following, we are going to assume that we know explicitly $\|O_I\|$, or a bound of it, which is typically the case in practice where we know what observable we want to measure, so that we can renormalise $O'=O/r(O)$ so that $r(O')=1$ and we achieve the bound:
\begin{align}
    |\Tr(O'\rho(x)) - \sum_{I\in {\cal P}_p(\Lambda)}f(O'_I,\chi_{S_{I,\delta}}(x))|
    \le
    \epsilon\,.
\end{align}

\subsubsection{Approximation of the energy}
\label{sec:Approximation of the energy}
In the previous sections we assumed that the observable $O_I$ does not depend on $x$. This leaves out the important task of predicting the energy of the system, where $O_I = h_I(x_I)$. For this we need to modify the derivations above as follows.
We go back to the computation of the gradient of $f(h_I(x_I),x)$:
\begin{align}
    &\|\nabla_J
    f(h_I,x)
    \|
    =
    \|
    \Tr(\nabla_J h_I \rho(x))
    \|
    +    
    \|
    \Tr(h_I [D_{J}(x),\rho(x)])
    \|
    \,.
\end{align}
We need to bound the first term:
\begin{align}
    \|
    \Tr(\nabla_J h_I \rho(x))
    \|
    \le
    \|
    \nabla_J h_I 
    \|
    =
    \delta_{I,J}
    \|
    \nabla_I h_I
    \|
\end{align}
so that the difference $|f(h_I,x)-f(h_I,x')|$ has the following extra term w.r.t.~that of Section \ref{sec:Dependency of observables on Hamiltonian parameters}, with $R=d(I,J)$:
\begin{align}
\left[\sum_{J\in {\cal P}_k(\Lambda)}
    \id\left(R\le \delta
    \wedge
    \diam(J)> \delta
    \right)
+
    \sum_{J\in {\cal P}_k(\Lambda)}
    \id\left(
    R>\delta\right)
\right]
\delta_{I,J}
    \|
    \nabla_I h_I
    \|=
    \id\left(
    \diam(I)> \delta
    \right)
    \|
    \nabla_I h_I
    \|
    \,.
\end{align}
For exponentially decaying interactions, we can bound this by $\e^{-\nu \delta}$, which is subleading compared to the terms in \eqref{eq:diff_OI_exp}, and thus does not alter the scaling of $\delta$ with $\epsilon$ to ensure that $|f(O_I, \chi_{S_{I,\delta}}(x))-f(O_I, x)|\le \epsilon \|O_I\|$
also in the case of $O_I=h_I(x_I)$.

In the case of power law interactions with $\alpha>2D$, we have from Proposition \ref{prop:approx_alpha_gt_2d} that
\begin{align}
    \id\left(
    \diam(I)> \delta
    \right)
    \|
    \nabla_I h_I
    \|
    \le 
    \id\left(
    \diam(I)> \delta
    \right)
    \diam(I)^{-(|I|-2)D-\alpha}
    \le 
    \delta^{-(|I|-2)D-\alpha}
    \le 
    \delta^{-\alpha}\,.
\end{align}
We assumed $|I|>1$ since for $|I|=1$ the condition $\diam(I)>\delta$ cannot be satisfied.
After some algebra we have, with $x=\alpha-2D$,
\begin{align}
    &(\nu-D-\alpha)16 (2 D+2 x-1) \left(D+x^2+x\right)^2
    =
    -32 D^3 (2 D-1)-72 D^2 (3 D-1) x-(80 D+39) x^5\\
    &-2 (D (32 D+91)+7) x^4-8( D (37 D+8)-1) x^3-4 D (D (32 D+51)-12) x^2-22
   x^6
\end{align}
and $\nu-D-\alpha\le 0$ since it is the ratio of a polynomial with negative coefficients and one with positive coefficients for all $D=1,2,\dots$
So 
$-\alpha\le -\nu +D$ and 
$\delta^{-\alpha}$ is subleading in the bound of Proposition \ref{prop:approx_alpha_gt_2d} and so does not alter the scaling of $\delta$ with $\epsilon$ to ensure that $|f(O_I, \chi_{S_{I,\delta}}(x))-f(O_I, x)|\le \epsilon \|O_I\|$
for $O_I=h_I(x_I)$ also in the case of power law interactions.

For illustration we now compute $r(H)$, with $r(O)$ as in \eqref{eq:rO}, for the case of power law decay with exponent $\alpha$ and $k=2$.
Let us assume that we have a cubic system of side $L$, so that $n=L^D$.
We have, assuming also that the $k=1$ terms to have norm $1$:
\begin{align}
    r(H)&=\sum_{ij}
    \|h_{ij}\|
    +
    \sum_i \|h_i\|
    \le
    \sum_{ij}
    (1+d(i,j))^{-\alpha}
    +
    \sum_i 1
    \\
    &=\sum_i \sum_j 
    (1+d(i,j))^{-\alpha}
    +n
    \le
    \sum_i 
    \int_{\mathbb{R}^D} \dd^D x
    (1+\|x-i\|)^{-\alpha}+n
    \\
    &\le  
    \Omega_D \sum_i 
    \int_{0}^\infty \dd r
    (1+r)^{-\alpha}+n
    \le 
    \frac{\Omega_D}{1-\alpha}\sum_i 1+n
    = 
    \frac{\Omega_D}{1-\alpha}n +n \,.
\end{align}
Here $\Omega_D$ is a constant and 
so $\sum_I\|h_I\|\le C n$ and
 the observable $H/n$ can be approximated by 
$\sum_{I}
    f(h_I(x_I),
    \chi_{S_{I,\delta}}(x)
    )$.
A similar computation can be done for $k>2$ and for
exponentially decaying interactions.

\section{Machine learning model and sample complexity bounds}
\label{sec:main - Machine learning model}

\subsection{Generalisation bounds for single observable}

We assume we are given data $S=\{x^{(i)}, y^{(i)}\}_{i=1}^N$ with
\begin{align}
\label{eq:eps_2}
|y^{(i)} -\Tr(O\rho(x^{(i)}))|\le \epsilon_2\,,\quad 
\end{align}
for an observable as in \eqref{eq:O}:
\begin{align}
    O=\sum_{I\in {\cal S}}O_I\,,\quad 
    {\cal S}\subseteq {\cal P}_p(\Lambda)\,.
\end{align}
The machine learning algorithm learns the weights $w$ of 
\begin{align}
    h(x)
    =
    w\cdot \phi(x)
    =
    \sum_{I\in {\cal S}}\sum_{x'\in X_{I,\delta}} (w)_{I,x'}\phi(x)_{I,x'}
    \,,\quad \phi(x)_{I,x'} = \id(x\in T_{x',I})\,,
\end{align}
with $X_{I,\delta}$ and $T_{x',I}$ as in \eqref{eq:XIdelta} and \eqref{eq:TxIdelta}.
We are going to study the generalisation properties of the predictor obtained by minimising the training error 
\begin{align}
    R_S(h)
    =
    \frac{1}{N}\sum_{i=1}^N |h(x^{(i)}) - y^{(i)}|^2
    \,.
\end{align}
Given the bound on the norm of $w$ in Lemma \eqref{lemma:bounds_m_phi_B}, we consider the optimisation problem
\begin{align}
    \label{eq:lasso}
    \min_{w\in {\cal C}_B} R_S(h_w)
    \,,\quad 
    {\cal C}_B=\{ w\in \mathbb{R}^m\,\,\,|\,\,\,\|w\|_1\le B\}\,.
\end{align}
This algorithm is called LASSO in the statistics literature and has been extensively studied \cite{mohri2018foundations}.

We have the following result which adapts Theorem 1 of \cite{lewis2023improved} to our setting of long range interactions and observables. 

\begin{theorem}\label{thm:sample_complexity_single_obs}
Consider an observable $O$ 
\begin{align}
    O=\sum_{I\in {\cal S}}O_I\,,\quad 
    {\cal S}\subseteq {\cal P}_p(\Lambda)\,.
\end{align}
and a dataset $S=\{x^{(i)}, y^{(i)}\}_{i=1}^N$ with
\begin{align}
|y^{(i)} -\Tr(O\rho(x^{(i)}))|\le \epsilon_2\,.
\end{align}
Choose $\delta$ as 
\begin{align}
    \delta > \max(b^{-1}\xi_*, b^{-1}\eta_*, c\log^2(1/\epsilon_1))
\end{align}
for exponentially decaying interactions, where the constants are as in Proposition \ref{prop:approx_exp_decay} and Corollary \ref{cor:delta_eps_exp} and as
\begin{align}
    \delta > \max((b^{-1}\xi_*)^{1/\mu}, (b^{-1}\eta_*)^{1/\mu}, \tilde{c}\log^{2/\mu}(1/\epsilon_1),
    \tilde{c}\epsilon_1^{-1/(\nu-D)})\,,
\end{align}
for power law decay with $\alpha>2D$, where the constants are as in Proposition \ref{cor:delta_eps_alpha_gt_2D}, Corollary \ref{cor:delta_eps_alpha_gt_2D} , and for $\alpha>2D+1$, $\nu,\mu$ are evaluated at $\alpha=2D+1-\zeta$.
Then form the linear predictor $h_*(x)=w_*\cdot \phi(x)$ with $w_*$ a solution of the optimisation problem \eqref{eq:lasso} such that
\begin{align}
    \label{eq:training_bound_eps3}
    R_S(h_*)
    \le 
    \frac{\epsilon_3}{2}
    +
    \min_{w\in {\cal C}_B} R_S(h_w)\,.  
\end{align} 
If the number of samples is
\begin{align}
\label{eq:ScalingComplexityEpsilonDependence}
    N 
    =
    r(O)^4
    \epsilon_3^{-2}
    {\cal N}(\epsilon_1) 
    \mathcal{O}(\log(|{\cal S}|/\delta))    
    \,,
\end{align}
with 
\begin{align}
    \omega = \frac{kD}{\nu-D}\,,\quad
    {\cal N}(\epsilon)
    =
    \begin{cases}
    2^{\mathcal{O}(\operatorname{polylog}(1/\epsilon))} & \textup{ exp decay},\\
    2^{\mathcal{O}(\epsilon^{-\omega} \log(1/\epsilon))} & \textup{ power law }\alpha>2D\,,
\end{cases}
\end{align}
the following generalisation bound holds with probability at least $1-\gamma$:
    \begin{align}
    \label{eq:gen_bound_eps}
    \mathbb{E}_{(x,y)\sim {\cal D}}|h_*(x)-y|^2
    \le 
    (\epsilon_1+\epsilon_2)^2+\epsilon_3\,.
    \end{align}
\end{theorem}

The sample complexity depends polynomially on $r(O)$ and logarithmically on $|{\cal S}|$. If we normalise $O$ such that $r(O)=\mathcal{O}(1)$, then the dependency on $n$ will be $\mathcal{O}(\log(n))$ as long as $O$ is the sum of $k$-body interactions with $|{\cal S}|=\mathcal{O}(n^k)$ with $k=\mathcal{O}(1)$. This is the same for both cases of exponentially decaying interactions and power law interactions with $\alpha>2D$. 
However, the dependency on $1/\epsilon$ goes from quasi-polynomial to exponential as we go from exponential to power law interactions. Now we will prove this theorem.
\begin{proof}
The proof relies on the generalisation bound for LASSO.
\begin{lemma}[\cite{mohri2018foundations}]\label{lemma:gen_bound}
    Denote the input space ${\cal X}\subseteq \mathbb{R}^A$ and the output space ${\cal Y}$. We consider a class of linear predictors $h(x) = w\cdot x$ with $\|w\|\le B$. If ${\cal D}$ is the distribution over ${\cal X}\times {\cal Y}$ from which the data $S$ of size $N$ is drawn, then with probability at least $1-\gamma$ we have
    \begin{align}
        \mathbb{E}_{(x,y)\sim {\cal D}}|h(x)-y|^2
        \le R_S(h)
        +2r_\infty B M \sqrt{\frac{2\log(2A)}{N}}+
        M^2 \sqrt{\frac{\log(1/\gamma)}{2N}}
    \end{align}
    with $\|x\|_\infty \le r_\infty$, $|h(x)-y|\le M$ for all $(x,y)\in {\cal X}\times {\cal Y}$, and $R_S$ is the training error.
\end{lemma}
To apply this, first we need to bound the training error.
\begin{lemma}[Lemma 15 of \cite{lewis2023improved}]
    Let $w_*$ be a solution of the optimisation problem \eqref{eq:lasso} such that
    \begin{align}
    R_S(h_*)
    \le 
    \frac{\epsilon_3}{2}
    +
    \min_{w\in {\cal C}_B} R_S(h_w)\,,\quad 
    h_*(x)=w_*\cdot \phi(x)\,.  
    \end{align}
    Also assume that 
    \begin{align}
        |f(O, x) - g(O,x)|\le \epsilon_1\,.
    \end{align}
    Then, with $\epsilon_2$ as in \eqref{eq:eps_2},
    \begin{align}
    R_S(h_{*})
    \le 
    \frac{\epsilon_3}{2}
    +
    (\epsilon_1 + \epsilon_2)^2\,.
\end{align}
\end{lemma}
Next we use Lemma \ref{lemma:gen_bound} with ${\cal X}=\{0,1\}^{m_\phi}$, so that $r_\infty=1$ and 
\begin{align}
    A=m_\phi \le 
    |{\cal S}| {\cal N}(\epsilon)\,.
\end{align}
Since $\Tr(O\rho)-\epsilon_2\le y\le \Tr(O\rho)+ \epsilon_2$, we have $|y|\le |\Tr(O\rho)|+|\epsilon_2|$, which can be verified by checking the two cases $y>0$ and $y<0$. Then from Lemma \ref{lemma:bounds_m_phi_B},
\begin{align}
    |h(x)-y|
    \le
    |w\cdot \phi(x)|+|y|
    \le 
    \|w\|_1 \|\phi(x)\|_\infty + \|O\|+\epsilon_2
    \le 
    r(O){\cal N}(\epsilon_1)+r(O)+\epsilon_2
    \le r(O){\cal N}(\epsilon_1)
    \,,
\end{align}
where we used that $\epsilon_2 = \mathcal{O}(1)$ since we consider asymptotics in $1/\epsilon_2$, and so $M=B=r(O){\cal N}(\epsilon_1)$.
Now we want to find $N$ such that the generalisation bound is small
\begin{align}
    G = 2r_\infty B M \sqrt{\frac{2\log(2A)}{N}}+
    M^2 \sqrt{\frac{\log(1/\gamma)}{2N}}
    \le \frac{\epsilon_3}{2}\,.
\end{align}
We have, since ${\cal N}(\epsilon_1)^2={\cal N}(\epsilon_1)$, $2{\cal N}(\epsilon_1)={\cal N}(\epsilon_1)$:
\begin{align}
    G
    &\le
    r(O)^2
    {\cal N}(\epsilon_1)
    \frac{1}{\sqrt{N}}
    \left(
    \sqrt{2}
    \sqrt{\log(|{\cal S}|{\cal N}(\epsilon_1))}+
    \frac{1}{\sqrt{2}}
    \sqrt{\log(1/\gamma)}
    \right)\\
    &\le 
    r(O)^2{\cal N}(\epsilon_1)
    \frac{1}{\sqrt{N}}
    \left(
    \sqrt{\log(|{\cal S}|{\cal N}(\epsilon_1))}+
    \sqrt{\log(1/\gamma)}\right)\,.
\end{align}
So we need 
\begin{align}
    N 
    &= \frac{4}{\epsilon_3^2}
    r(O)^4 {\cal N}(\epsilon_1)
    \left(\sqrt{\log(|{\cal S}|)+\log{\cal N}(\epsilon_1)}+
    \sqrt{\log(1/\gamma)} \right)^2\\
    &=
    r(O)^4 
    \epsilon_3^{-2}
    {\cal N}(\epsilon_1)
    \mathcal{O}(\log(|{\cal S}|/\gamma))    \,.
\end{align}
In the second equality we used that $(a+b)^2\le 2(a^2+b^2)$, which follows from $0\le (a-b)^2=a^2+b^2-2ab$, and that 
${\cal N}(\epsilon_1)\log{\cal N}(\epsilon_1)={\cal N}(\epsilon_1)$.
\end{proof}

Assuming $r(O)=\mathcal{O}(1), \gamma=\mathcal{O}(1)$ in the definition of the data size $N$, as remarked in \cite{lewis2023improved} the training time of LASSO is $\mathcal{O}(m_\phi \log(m_\phi)/\epsilon_3^2)
=
\mathcal{O}(|{\cal S}| N)
$.

\subsection{Sample complexity reduction from equivariance}
\label{sec:SampleComplexityReductionEquivariance}

\subsubsection{Equivariance of observables}
\label{sec:Equivariance of observables}

We first recall given a map $f:{\cal X}\to {\cal Y}$ and a group $G$ acting on ${\cal X}$, ${\cal Y}$,
$f$ is called $G$-equivariant if $f(g\cdot x)=g\cdot f(x)$ for all $g\in G$ and $x$ in the input space ${\cal X}$, where $g\cdot$ denotes the group action.

Given a Hamiltonian 
    \begin{align}
    H(x)
    =
    \sum_{I\in {\cal E}}h_{I}(x_{I})
    \,,
\end{align}
with ${\cal E}\subseteq {\cal P}(\Lambda)$, we define the interaction hypergraph as the triple $\mathcal{I}=(\Lambda,\mathcal{E},h)$, where $h$ is a function that assigns to each $I\in \mathcal{E}$ the interaction term $h_I(x_I)$.
The automorphism group $G=\text{Aut}(\mathcal{I})$ of the interaction hypergraph is the set of permutations of vertices that permute the interactions terms; specifically, $g\in G$ if
\begin{align}
    \hat{g} h_I(x_I) \hat{g}^{-1}
    =
    h_{gI}(x_I)\,,
\end{align}
where $\hat{g}$ is the unitary representation of a permutation $g$ on the Hilbert space of the quantum theory and $gI$ is the permuted set of vertices, namely
if $I=(i_1,\dots,i_\ell)$, $gI=(gi_1,\dots,gi_\ell)$.
For example, for a system with periodic boundary conditions and pairwise interactions
$h_{ij}(x_{ij})=x_{ij}f(d(i,j)) H_{ij}$ with $d(i,j)$ the Euclidean distance and $f$ an arbitrary function, $G$ is the Euclidean group of symmetries of the lattice under which $d(i,j)$ is invariant.
This setting includes nearest neighbor and long-range interactions.

With this definition, we have the following result:
\begin{lemma}\label{lemma:equi_obs}
Consider a Hamiltonian with interaction hypergraph ${\cal I}$.
Then the function $f(O_I,x)=\Tr(O_I\rho(x))$ is $\text{Aut}({\cal I})$-equivariant: 
\begin{align}
f(O_{gI},x)=
f(O_I,g\cdot x)
\,,
\end{align} 
for all $x$ and $g\in \text{Aut}({\cal I})$, where $(g\cdot x)_J=x_{gJ}$. 
\end{lemma}
\begin{proof}
To prove the Lemma, we first we show that the Hamiltonian, seen as a map from $x$ to an operator on the qubits, is $\text{Aut}({\cal I})$-equivariant.
If $g\in \text{Aut}({\cal I})$ is a permutation of the vertices whose action on vertex $i$ is denoted by $gi$, the action on the hyperedge $I=\{i_1,\dots,i_{\ell}\}$ is $gI=\{gi_1,\dots,g i_\ell\}$ and by definition it fixes ${\cal E}=\{ I_0,\dots, I_{m-1} \}$, i.e. permutes the sets $I_0,\dots,I_{m-1}$:
\begin{align}
    g{\cal E}
    =
    \{ gI_0,\dots, gI_{m-1} \}
    =
    {\cal E}\,.
\end{align}
Then we have, denoted by $\hat{g}$ the unitary representation of $g$ on the space of the qubits,
\begin{align}
    \hat{g}H(x)\hat{g}^{-1}
    =
    \sum_{I\in {\cal E}}
    \hat{g}h_{I}(x_{I})\hat{g}^{-1}
    =
    \sum_{I\in {\cal E}}
    h_{gI}(x_{I})
    =
    \sum_{J\in {\cal E}}
    h_{J}(x_{g^{-1}J})
    =
    H(g^{-1}\cdot x)\,,
\end{align}
where we set $J=gI$ and used that $g{\cal E}={\cal E}$ to move the action from $h_{I}$ to its argument.
Further,
\begin{align}
    &Z_\beta(x)
    =
    \Tr(\e^{-\beta H(x)})
    =
    \Tr(\hat{g}\e^{-\beta H(x)}\hat{g}^{-1})
    =
    \Tr(\e^{-\beta H(g^{-1}\cdot x)})
    =Z_\beta(g^{-1}\cdot x)
    \,,\\
    \label{eq:g_rho_g}
    &\hat{g}\rho(x)\hat{g}^{-1}
    =
    \lim_{\beta\to\infty}
    \hat{g}
    \frac{\e^{-\beta H(x)}}{
    Z_\beta(x)
    }
    \hat{g}^{-1}
    =
    \lim_{\beta\to\infty}
    \frac{\e^{-\beta H(g^{-1}x)}}{
    Z_\beta(g^{-1}\cdot x)
    }
    =
    \rho(g^{-1}\cdot x)\,,
\end{align}
so we have the equivariance property of the function $f$ for any $g\in \text{Aut}({\cal I})$:
\begin{align}
    f(O_{gI}, x)
    =
    \Tr(O_{gI} \rho(x))
    =
    \Tr(\hat{g}O_{I}\hat{g}^{-1} \rho(x))
    =
    \Tr(O_{I}\rho(g\cdot x))
    =
    f(O_I, g\cdot x)\,.
\end{align}
\end{proof}
We stress that this notion of equivariance is not the same as invariance of the Hamiltonian, which is a  more restrictive concept and of more limited use.
Also, note that for the case of interactions over edges of a graph, $\text{Aut}({\cal I})$ reduces to the automorphism group of the graph. 
For example, if $H$ is the Hamiltonian for qubits on a closed chain with nearest neighbor interactions, 
\begin{align}
    H(x)
    =
    \sum_{i=0}^{n-1}
    h_{i,i+1}(x_{i,i+1})\,,
\end{align}
with indices understood modulo $n$, $\text{Aut}({\cal I})$ contains the cyclic group of order $n$ corresponding to translations of the chain, whose generator $g$ acts as $gi=i+1\mod n$ and cyclically permutes the entries of $x$ as:
\begin{align}
    g\cdot 
    (x_{0,1},x_{1,2},\dots,x_{n-1,0})
    =
    (x_{1,2}, x_{2,3}, 
    \dots,
    x_{0,1})\,.
\end{align}

\subsubsection{Equivariance of the ML model}

A ML model for the observable $f(O_I,x)$ should also satisfy the equivariance constraint of Lemma \ref{lemma:equi_obs}. In fact we show here that, provided that the model satisfies the equivariance constraints, we can reduce the sample complexity.

\begin{prop}\label{prop:equi_ml_model}
    The model  
    \begin{align}
    h^{O_I}_w(x)
    =
    \sum_{x'\in X_{I,\delta}}
    w_{x'} \phi(x)_{I,x'}\,,\quad
    \phi(x)_{I,x'} = \id(x\in T_{x',I})\,,
    \end{align}
    satisfies the equivariance constraint $h^{O_{gI}}_{\tilde{w}}(x)=h^{O_{I}}_w(g\cdot x)$ 
    of Lemma \ref{lemma:equi_obs}
    if the weights transform as
\begin{align}
\tilde{w}_{g^{-1}\cdot x'}
=
w_{x'}
\,,
\end{align}
for any $x'$ and $g\in \text{Aut}({\cal I})$.
\end{prop}
\begin{proof}
We want to relate 
\begin{align}
    h^{O_{gI}}_{\tilde{w}}(x)
    =
    \sum_{x'\in X_{gI,\delta}}
    \tilde{w}_{x'} \phi(x)_{gI,x'}\,,\quad
    \phi(x)_{gI,x'} = \id(x\in T_{x',gI})\,,
\end{align}
to $h^{O_I}_w(g\cdot x)$. So we first discuss the transformation of the various sets involved.
We have
\begin{align}
    S_{gI,\delta}
    &=
    \{
    J\in {\cal P}_k(\Lambda)
    \,|\,
    d(gI,J)\le \delta
    \wedge
    \diam(J)\le \delta
    \}\\
    &=
    \{
    J\in {\cal P}_k(\Lambda)
    \,|\,
    d(I,g^{-1}J)\le \delta
    \wedge
    \diam(g^{-1}J)\le \delta
    \}    \\
    \label{eq:Sg}
    &=
    \{
    gJ'\in {\cal P}_k(\Lambda)
    \,|\,
    d(I,J')\le \delta
    \wedge
    \diam(J')\le \delta
    \} 
    =
    g S_{I,\delta}
    \,.
\end{align}
In the second line we use invariance of the distance under permutation, $d(I,J)=d(gI,gJ)$, and that as a consequence the diameter is also invariant under permutations.
In the third line we relabelled $g^{-1}J=J'$ and used that ${\cal P}_k(\Lambda)=g{\cal P}_k(\Lambda)$.
Thus
\begin{align}
    X_{gI,\delta}
    &=
    \{ x\in [-1,1]^m \,|\,
    x_J = 0 \text{ if }J\not\in gS_{I,\delta}\,,\quad
    x_J\in W_{\delta_2}^{|J|} \text{ if }J\in gS_{I,\delta}
    \}\\
    &=
    \{ x\in [-1,1]^m \,|\,
    x_{gJ'} = 0 \text{ if }J'\not\in S_{I,\delta}\,,\quad
    x_{gJ'}\in W_{\delta_2}^{|J|} \text{ if }J'\in S_{I,\delta}
    \}\\
    &=
    \{ g^{-1}\cdot x'\in [-1,1]^m \,|\,
    x'_{J'} = 0 \text{ if }J'\not\in S_{I,\delta}\,,\quad
    x'_{J'}\in W_{\delta_2}^{|J|} \text{ if }J'\in S_{I,\delta}
    \}\\
    &=
    g^{-1}\cdot X_{I,\delta}
    \,,
\end{align}
where in the second line we defined $J=gJ'$ and used that $gJ\in gS$ if and only if $J\in S$.
In the third line we used that if $x'=g\cdot x$, then $x'_J = x_{gJ}$. Similarly,
\begin{align}
    T_{x,gI}
    &=
    \{
    x'\in [-1,1]^m \,|\,
    -\frac{\delta_2}{2}
    < x_{J} - x'_{J}<
    +\frac{\delta_2}{2}
    \,,\quad J\in gS_{I,\delta}
    \}\\
    &=
    \{
    x'\in [-1,1]^m \,|\,
    -\frac{\delta_2}{2}
    < x_{gJ'} - x'_{gJ'}<
    +\frac{\delta_2}{2}
    \,,\quad J'\in S_{I,\delta}
    \}\\
    &=\{
    g^{-1}\cdot z'\in [-1,1]^m \,|\,
    -\frac{\delta_2}{2}
    < (g\cdot x)_{J'} - z'_{J'}<
    +\frac{\delta_2}{2}
    \,,\quad J'\in S_{I,\delta}
    \}\\
    &=
    g^{-1}\cdot T_{g\cdot x, I}
    \,.
\end{align}
Here we have used the same manipulations as in $X_{gI,\delta}$ and defined $z'=g\cdot x'$.
Then 
\begin{align}
   \phi(x)_{gI,x'}=\id(x\in T_{x',gI})=
   \id(x\in g^{-1}\cdot T_{g\cdot x',I})
   =
   \id(g\cdot x\in T_{g\cdot x',I})
   =
   \phi(g\cdot x)_{I,g\cdot x'}
\end{align}
and we can prove the result:
\begin{align}
    h^{O_{gI}}_{\tilde{w}}(x)
    &=
    \sum_{x'\in X_{gI,\delta}}
    \tilde{w}_{x'} \phi(x)_{gI,x'}
    =
    \sum_{x'\in g^{-1}\cdot X_{I,\delta}}
    \tilde{w}_{x'}
    \phi(g\cdot x)_{I,g\cdot x'}
    \\
    &=
    \sum_{z'\in X_{I,\delta}}
    \tilde{w}_{g^{-1}\cdot z'}
    \phi(g\cdot x)_{I,z'}    
    =
    h^{O_{I}}_{w}(g\cdot x)
    \,,
    \quad w_{z'}=\tilde{w}_{g^{-1}\cdot z'}\,,
\end{align}
where we called $z'=g\cdot x'$.
\end{proof}
This means that if we have a model
\begin{align}
    h(x)
    =
    w\cdot \phi(x)
    =
    \sum_{I\in {\cal S}}\sum_{x'\in X_{I,\delta}} (w)_{I,x'}\phi(x)_{I,x'}
    \,,\quad \phi(x)_{I,x'} = \id(x\in T_{x',I})\,,
\end{align}
for predicting $\Tr(O\rho(x))$, where
\begin{align}
    O=\sum_{I\in {\cal S}}O_I\,,\quad 
    {\cal S}\subseteq {\cal P}_k(\Lambda)\,,
\end{align}
we can take weights that satisfy the equivariance condition of the Proposition \ref{prop:equi_ml_model},
\begin{align}
 w_{I,x'}=w_{gI,g^{-1}\cdot x'}\,,   
\end{align}
for all $g\in \text{Aut}({\cal I})$.
This is indeed the equivariance of the true weights $(w')_{I,x'}$ of \eqref{eq:target_w}.
This leads to a reduction in the number of learnable parameters and thus a sample complexity reduction.
More precisely, we consider the quotient space ${\cal S}/\text{Aut}({\cal I})$ where two sets $I,J\in {\cal S}$ are equivalent if there is a $g\in \text{Aut}({\cal I})$ such that $I=gJ$. Then equivariance relates all the weights related to sets within an equivalence class of ${\cal S}/\text{Aut}({\cal I})$, and leaves a number of independent weights equal to the cardinality of the space ${\cal S}/\text{Aut}({\cal I})$. We summarise this result in the following corollary.

\begin{corollary}\label{cor:sample_complexity_equiv}
    The sample complexity of Theorem \ref{thm:sample_complexity_single_obs} is reduced by using an equivariant model to
    \begin{align}
    N 
    =
    r(O)^4
    \epsilon_3^{-2}
    {\cal N}(\epsilon_1) 
    \mathcal{O}(\log(|{\cal S}/\text{Aut}(\cal I)|/\gamma))    
    \,.
\end{align}
\end{corollary}

Continuing our example of the $D=1$ periodic chain of Section \ref{sec:Equivariance of observables}, if we consider an observable that is a sum of local terms such that
\begin{align}
    O=\sum_{i=0}^{n-1}O_i\,,\quad 
    O_i = O_{g^i 0}\,,
\end{align}
with $g$ the translation by one site, $gi=i+1\mod n$, we have ${\cal S}=\{0,\dots,n-1\}$ and $|{\cal S}/\text{Aut}({\cal I})|=1$, leading to a $\mathcal{O}(1)$ complexity of learning to predict this observable, using the following ML model
\begin{align}
    h(x)
    &=
    \sum_{i=0}^{n-1}
    \sum_{x'} (w)_{i,x'}\phi(x)_{i,x'}
    =
    \sum_{i=0}^{n-1}
    \sum_{x'}     
    (w)_{0,g^i\cdot x'}
    \phi(x)_{i,x'}
    =
    \sum_{i=0}^{n-1}
    \sum_{x''}     
    (w)_{0,x''}
    \phi(x)_{i,g^{-i}\cdot x''}
    \\
    &=
    \sum_{i=0}^{n-1}
    (w_0 * \Phi(x))_i
    \,,
\end{align}
with $*$ the convolution with filter $w_0$ and $\Phi(x)$ is obtained by concatenating $\{\phi(x)_{i,g^{-i}\cdot x''}\}_{x''}$ over $i$.

\subsubsection{Equivariance of the random feature model}
\label{sec:Equivariance of the random feature model}

We discuss here equivariance of the random feature model that is used in the experiments Section \ref{sec:Experiments}.
For a given $\delta$, we define the set of parameters that influence $f(O_I,x)$ according to the results of Section \ref{sec:Dependency of observables on Hamiltonian parameters}:
\begin{align}
    Z_I(x) = 
    ( x_J\,|\, J\in S_{I,\delta} )
    \in \bigoplus_{J\in S_{I,\delta}} [-1,1]^{q_J}\,.
\end{align}
The random feature model is defined as
\begin{align}
    F^{O_I}(x)
    =
    \sum_{i=1}^R a_i \sigma(b_i^T Z_{I}(x))
    \,,
\end{align}
where $a\in \mathbb{R}^R$ are the learnable parameters and $b_i$'s are random vectors of the same size of $Z_{I}$. $\sigma$ is a non-linearity and thus $F^{O_I}$ can be thought as a two-layer neural network with random weights in the first layer \cite{E_2020}.
\begin{prop}\label{prop:equi_rff}
    The random feature model satisfies the equivariance constraint of Lemma \ref{lemma:equi_obs},
    \begin{align}
        F^{O_{gI}}(x)
        =
        F^{O_I}(g\cdot x)
    \end{align}
    for all values of $a, b$.
\end{prop}
\begin{proof}
    From \eqref{eq:Sg}, we get, with $gJ'=J$
    \begin{align}
    Z_{gI}(x) = 
    ( x_J\,|\, J\in S_{gI,\delta} )
    =
    ( x_J\,|\, J\in gS_{I,\delta} )
    =
    ( x_{gJ'}\,|\, J'\in S_{I,\delta} )
    =
    Z_{I}(g\cdot x)\,.
    \end{align}
    Thus for any values of the parameters $a,b$:
    \begin{align}
        F^{O_{gI}}(x)
        =
        \sum_{i=1}^R a_i \sigma(b_i^T Z_{gI}(x))
        =
        \sum_{i=1}^R a_i \sigma(b_i^T Z_{I}(g\cdot x))
        =
        F^{O_{I}}(g\cdot x)\,
        .
    \end{align}
\end{proof}
As in the case of the feature map $\phi(x)$, this means that we can reduce the number of parameters when predicting an observable by associating a set of $a,b$'s to each equivalence class in ${\cal S}/\text{Aut}({\cal I})$.
For example, for the case of a $D=1$ periodic chain and observable
\begin{align}
    O = \sum_{i=0}^{n-1} O_i\,,\quad O_i=O_{g^i 0}
\end{align}
we can use the random feature model with $\mathcal{O}(1)$ parameters:
\begin{align}
    F(x)
    =
    \sum_{i=0}^{n-1}
    F^{O_{g^i 0}}(x)
    =
    \sum_{i=0}^{n-1}
    F^{O_{0}}(g^i\cdot x)
    =
    \sum_{i=0}^{n-1}  
    \sum_{j=1}^R a_{j}       
    \sigma(b_{j}*x)_i
    =
    \sum_{i=0}^{n-1}  
    (a * \Phi(x))_i    
    \,.
\end{align}
Here we have first noted that 
\begin{align}
(b_{j}*x)_i=b_j^T Z_{0}(g^i\cdot x)
=
\sum_{k=1}^{|S_{0,\delta}|}
b_{j,k} x_{k+i\!\!\!\mod n}
\end{align}
is the $i$-th output of the 1D convolution with filter $b_j$ and input $x$, with $\delta$ determining the filter size.
We can think of $j$ as indexing the channel dimension of size $R$ and $i$ the space dimension of size $n$.
The output of the second layer can be also interpreted as a convolution with filter $a$ of size $1$ in the space dimension and $R$ in the channel dimension. $\Phi(x)$ is the concatenation of 
$\{\sigma(b_{j}*x)_i\}_{j=1}^R$ over $i$.
So the model is a two-layer convolutional neural network, with final layer a global pooling over space.

\subsection{Classical shadows}
\label{sec:Classical shadows}

A classical shadow of the density matrix $\rho$ is obtained by averaging repeated random measurements \cite{Huang_2020}.
One selects  uniformly at random whether to measure $X,Y,Z$ for each qubit and stores the associated measurement results as classical data $s_i^t\in \{ 0_X,1_X,0_Y,1_Y,0_Z,1_Z\}$ for the measurement outcome of qubit $i$ at time $t$. 
Here we denote by $0_A,1_A$ the possible states after measurement of the Pauli $A$.
After $T$ measurements, the classical shadow is 
\begin{align}
    \sigma_T(\rho)
    =
    \frac{1}{T}
    \sum_{t=1}^T
    \left(3\ket{s_1^t}\bra{s_1^t}-\id_2 \right)\otimes \cdots \otimes
    \left(3\ket{s_n^t}\bra{s_n^t}-\id_2 \right)\,.
\end{align}
The power of classical shadows is that they can $\epsilon$-approximate the reduced density matrices of an $\mathcal{O}(1)$ subsystem with $T=\mathcal{O}(\log(n)/\epsilon^2)$.
Our first result is the extension of \cite[Corollary 5]{lewis2023improved} to our setting.

\begin{theorem}\label{thm:sample_complexity_classical_shadows}
    Suppose we have data $\{x^{(i)}, \sigma_T(\rho(x^{(i)}) \}_{i=1}^N$ with
\begin{align}
    \label{eq:T_shadow}
    T &= \mathcal{O}(\log(nN/(\gamma/2))/(\epsilon_2')^2)
    \,,\quad
    \epsilon_2'=\epsilon_2/C\,,\\
    N &=
    (\epsilon_3/C^2)^{-2}
    {\cal N}(\epsilon_1/C)
    \mathcal{O}(\log(n/\gamma)) 
    \,,
\end{align}
with $\epsilon_1,\epsilon_2,\epsilon_2,C>0$ and $\gamma\in(0,1)$.
    Then we can learn a predictor $\hat{\rho}(x)$ that achieves
    \begin{align}
        \mathbb{E}_{x\sim {\cal D}}|\Tr(O\rho(x))-\Tr(O\hat{\rho}(x))|^2
    \le 
    (\epsilon_1+\epsilon_2)^2+\epsilon_3\,,
    \end{align}
    with probability at least $1-\gamma$ for all observables $O$ such that 
    \begin{align}
    \label{eq:O_alpha_P}
    O = \sum_{P\in P_k} \alpha_P P
    \,,\quad 
    \sum_{P\in P_k}|\alpha_P|\le C\,,
\end{align}
where $P_k$ is the set of Pauli strings with weight at most $k$. 
\end{theorem}
This shows that given classical shadow data we can predict all observables of the form $\sum_{I\in {\cal P}_k}O_I$ 
with a sample complexity similar to that of predicting a single one, presented in Theorem \ref{thm:sample_complexity_single_obs}.
The proof adapts \cite[Corollary 5]{lewis2023improved} to our setting and is presented in Appendix \ref{sec:Classical shadows and prediction of many observables}.

We also extend the results for equivariant Hamiltonians to data obtained from classical shadows. Recall the definition of $\text{Aut}(\mathcal{I})$
from \ref{sec:Equivariance of observables}. We have the following result:
\begin{corollary}
    Under the hypothesis of Theorem \ref{thm:sample_complexity_classical_shadows} and assuming further that ${\cal D}$ is a probability distribution over $x$ that is invariant under $\text{Aut}({\cal I})$, 
    we can predict all observables of the form \eqref{eq:O_alpha_P} with
    \begin{align}
    N 
    =
    (\epsilon_3/C^2)^{-2}
    {\cal N}(\epsilon_1/C)
    \mathcal{O}(\log(|P_k/\text{Aut}({\cal I})|/\gamma)) 
    \,,
\end{align}
where $\text{Aut}({\cal I})$ is the automorphism group of the interaction hypergraph.
\end{corollary}
This result is most powerful if $\mathcal{O}(|P_k/\text{Aut}({\cal I})|)=\mathcal{O}(1)$, since it reduces the complexity from $\mathcal{O}(\log(n))$ to $\mathcal{O}(1)$.
As remarked before, this is the case for example if we consider observables that are sum of geometrically local terms (which is the setting of \cite{lewis2023improved,onorati2023efficient,onorati2023provably}) on a lattice with periodic boundary conditions.

The core of the proof of this result is to show that if
$h^P_*(x)$ is an equivariant model as in Proposition \ref{prop:equi_ml_model}, then
    \begin{align}
        \mathbb{E}_{x\sim {\cal D}}
        \left( |h^P_*(x)-y^P(x)|^2 \right)
        =
        \mathbb{E}_{x\sim {\cal D}}
        \left( |h^{P'}_*(x)-y^{P'}(x)|^2 \right)
        \,,\quad P'=\hat{g}P\hat{g}^{-1}\,,
    \end{align}
where $y^{(P)}(x)$ is the expectation value of $P$ in the classical shadow, 
$y^{(P)}(x)=\Tr(P\sigma_T(\rho(x)))$, and $g\in\text{Aut}(\mathcal{I})$.
In turn this follows by showing the equivariance of the measure on $y^{(P)}(x)$ induced by the classical shadow measurement protocol.
See Appendix \ref{sec:Equivariant classical shadows} for details.

\section{Experiments} 
\label{sec:Experiments}

In this section, we will cover the details and the setup for the numerical experiments carried out and present the results affirming the theoretically guaranteed bounds. We will discuss implementation of the simulations for different systems and of the machine learning model. For most of the experiments, our main focus will be on demonstrating the predicted efficient scaling of the algorithm. We will also cover an important phenomenon observed for practical usability of the algorithm regarding normalisation of the observables and concentration of expectation values in Section \ref{sec:ShrinkingPhenomenon}.

\subsection{Setup}

The simulations of the quantum systems were typically obtained using exact diagonalisation for up to $20$ sites, while larger systems were simulated using tensor networks and the Density Matrix Renormalisation Group (DMRG) method \cite{Schollw_ck_2011}. The tensor network calculations were carried out using the Python library \textsc{TeNPy} \cite{TeNPy}. Note that DMRG is a variational method, and so the resulting state is just an approximation of the true ground state, hence there can be possible issues with convergence and small numerical errors. The practical applicability of DMRG is the reason why we focus on one dimensional systems for the numerical experiments; either with open or periodic boundary conditions. For DMRG, the SVD cut-off was set to $10^{-9}$ while the convergence criterion for the ground energy was a relative change of $10^{-6}$. The maximal number of sweeps used was typically set between $100$ and $200$, and the maximal bond dimension of the MPS was typically gradually increased up to $100$ or $200$.

To determine whether or not the Hamiltonians remain gapped across the chosen range of parameters, we would also at first measure the spectral gaps for many random choices of parameters as the system sizes increase to see if they aren't closing. To do this for larger systems, the DMRG would be run on a Hamiltonian with projection onto the orthogonal space to the ground state to obtain the first excited state.

Following Appendix D of \cite{lewis2023improved}, the machine learning part is based on LASSO, an $l_1$-regularised linear regression model, with feature mapping with built in local-geometric bias. Instead of using the discretisation as was considered in the theoretical guarantees, we will use random Fourier features, as discussed in Section \ref{sec:Equivariance of the random feature model}, which are commonly used in practice \cite{rndFourierFeatures}.

Given an observable $O = \sum\limits_{I\in \mathcal{S}} O_I$, for any $I \in \mathcal{S}$ with $|I|=g$ occurring non-trivially in $O$, one would consider all the parameters $x_J$'s in its $\delta$-neighbourhood for all $J$ such that $d(I,J) \leq \delta$ and $\diam(J) \leq \delta$ to create a vector
$$Z_I^{(\delta)} = (x_J \,|\, J \in S_{I,\delta} )\,.$$
This vector would then be mapped to its randomised Fourier features via
\begin{equation}\phi:Z \mapsto \left( \begin{matrix}
\cos\left(\frac{\gamma}{\sqrt{l}} (\omega_1 \cdot Z)\right)\\
\sin\left(\frac{\gamma}{\sqrt{l}} (\omega_1 \cdot Z)\right)\\
\vdots\\
\cos\left(\frac{\gamma}{\sqrt{l}} (\omega_R \cdot Z)\right)\\
\sin\left(\frac{\gamma}{\sqrt{l}} (\omega_R \cdot Z)\right)\\
\end{matrix} \right),\end{equation}
where $l$ is the length of the vector $Z$, and $\gamma$ and $R$ are hyperparameters of the model, to create vectors
\begin{equation}\phi_I ^{(\delta)} = \phi\left(Z_I^{(\delta)}\right).\label{eq:FeatureVectorOneI}\end{equation} 
The $\omega_i$'s are vectors of length $l$ of independent standard normal random variables, which are the same for all $I$ with $|I|=g$ for a specific $g$, and all of which are set for a given LASSO model and used for all its training and testing samples.
These vectors would then be concatenated over all such $I$ into a $g$-body interaction feature vector \begin{equation}\Phi ^{(g,\delta)} =\left( \phi_I ^{(\delta)} \,|\, I \in \mathcal{S} \wedge |I|=g \right),\label{eq:FeatureVectorAllIOnep}\end{equation}
which would finally be concatenated into the full feature vector \begin{equation}\Phi ^{(\delta)} =\left( \Phi ^{(g,\delta)} \,|\, 1 \leq g \leq p \right),\label{eq:FullFeatureVector}\end{equation}
used as the input for the LASSO model. This means that the full feature dimension for an observable $O$ as defined above would be at most 
$$m_\Phi = 2R \cdot \sum_{g=1}^p {n \choose g} = 2R \cdot  |{\cal P}_p(\Lambda)| = \mathcal{O}(n^p)\,.$$ 
Note that in the experiments that will be presented shortly, we won't actually meet an observable for which we would use parameters acting on different numbers of sites (due to the specific implementation of the Ising chain), and hence the last concatenation step won't be used.

For tuning the hyperparameters $R$, $\gamma$, and also the $l_1$-regularisation parameter $\alpha$, appearing in the LASSO minimisation problem \begin{equation}\min_{\substack{w\in \mathbb{R}^{m_\Phi}\\ y_0 \in \mathbb{R}}} \left\{\frac{1}{2N} \sum\limits_{i=1}^N \left|y^{(i)}-y_0 - w \cdot [\Phi^{(\delta)}]^{(i)} \right|^2 + \alpha \cdot \|w\|_1\right\},\end{equation} we used $5$-fold cross-validation. As per \cite{lewis2023improved}, the possible options for these parameters were taken to be $$\alpha \in \left\{ 2^{-8}, 2^{-7}, 2^{-6}, 2^{-5} \right\},$$ $$\gamma \in \{0.4, 0.5, 0.6, 0.65, 0.7, 0.75 \},$$ $$R \in \{5, 10, 20, 40\}.$$ Note that for a given model, the hyperparameters used are the same everywhere; in particular, when creating the vectors $\Phi^{(g,\delta)}$, they are the same for all $I$ and $g$. The LASSO model and the cross-validation were implemented using the Python library \textsc{scikit-learn} \cite{scikit-learn}.

Having implemented the machine learning algorithm, we wanted to demonstrate the scaling of the number of training samples with respect to the system size. Hence, for any given set of testing samples, we've run the algorithm repeatedly with increasingly more training samples up until a fixed additive root mean squared (RMS) error was achieved. The plots of training samples needed for distinct numbers of qubits will then be used to show the scaling of this algorithm. The corresponding code for these simulations is available at \cite{Smid_Efficient_Learning_of_2024}.

\subsection{Scaling of observables, self-averaging of expectation values, and the central limit theorem}
\label{sec:ShrinkingPhenomenon}

In certain scenarios, for example when evaluating the performance of this algorithm as in the following sections, one needs to be careful when normalising global observables. To illustrate this point, let's consider the nearest-neighbour Heisenberg model on one-dimensional chain with open boundary conditions, as described in Section \ref{sec:Heisenberg}. The Hamiltonian for this model is $$H = \sum\limits_{\langle i j \rangle} J_{ij} (X_i X_j + Y_i Y_j + Z_i Z_j),$$
and we would like to  measure correlations $C_{ij} = \frac{1}{3} (X_i X_j + Y_i Y_j + Z_i Z_j)$ between nearest neighbours. 

When one considers these correlations individually, then for any given chain of length $n$ with the coupling coefficients $J_{ij}$ sampled uniformly at random from $[0,2]$, their expectation values range throughout the whole possible interval $[-1,0]$, being determined accurately by the local parameters $J_{ij}$ in their neighbourhood as per the theoretical guarantees discussed in Section \ref{sec:Dependency of observables on Hamiltonian parameters}. 

But now, let's consider the observable to be the average of these correlations, $$O = \frac{1}{n-1} \sum\limits_{i=0}^{n-2} C_{i,i+1} = \frac{1}{n-1} \sum\limits_{i=0}^{n-2} \frac{1}{3} (X_i X_{i+1} + Y_i Y_{i+1} + Z_i Z_{i+1}).$$ This is clearly correctly normalised to $r(O) = \mathcal{O}(1)$ as per the assumptions of Theorems \ref{thm:sample_complexity_single_obs}, \ref{thm:sample_complexity_classical_shadows}.
But we know from Appendix \ref{sec:CLTGlobalObservables} that in the large $n$ limit, the expectation values will concentrate around their average over the distribution of the parameters, provided that their variance tends to $0$ as $\frac{1}{n}$, which we will check numerically.

Table \ref{tab:ShrinkingObservables} presents the data collected from simulations of Heisenberg chains of different lengths, each being based on 100 samples, when measuring the observable $O$ as defined previously. Even though each individual correlation in the chain ranges from $-1$ to $0$, it takes only $n = 8$ sites to get the possible range of $O$ down to [$-0.642$, $-0.468$], which continues to shrink, and gets down to only [$-0.557$,  $-0.528$] for $n = 128$ sites.  Fitting the results for $\frac{\operatorname{Var}}{\operatorname{Mean}^2}$ from 16 qubits onwards with a function of the form $a+\frac{b}{n}$ yields values $a = -0.000039$ and $b = 0.024$, confirming the applicability of the central limit Theorem as given by Proposition \ref{prop:clt} to this case.

\begin{table}[h]
    \centering
    \begin{tabular}{|c||c|c|c|c|}
        \hline
         System size $n$ & Range of $\Tr(O \rho(x))$ & Size of range & Standard deviation & $\sqrt{n}\cdot $SD \\\hhline{|=#=|=|=|=|}
         4 & [$-0.718,-0.433$] & 0.285 & 0.0662 & 0.1324 \\\hline
         8 & [$-0.642, -0.468$] & 0.174 & 0.0444 & 0.1256 \\\hline
         16 & [$-0.605, -0.500$] & 0.105 &  0.0206 & 0.0824 \\\hline
         32 & [$-0.571, -0.494$] & 0.078 & 0.0162 & 0.0916 \\\hline
         64 & [$-0.563, -0.526$] & 0.038 & 0.0079 & 0.0632 \\\hline
         128 & [$-0.557,  -0.528$] & 0.029 & 0.0063 & 0.0713\\\hline
    \end{tabular}
    \caption{Demonstration of self-averaging of expectation values for the average correlation $O = \frac{1}{n-1}\sum_i C_{i,i+1}$ in the Heisenberg model obtained from 100 samples for each system size. Note that the scaled version of standard deviation is actually asymptotically constant, even though it appears here to be oscillating - this is just a result of the intrinsic randomness of the systems.}
    \label{tab:ShrinkingObservables}
\end{table}

This behaviour means that if one would like to predict the observable to a given absolute additive RMS error $\epsilon$, it will be actually much easier to do so for larger systems; so much so, that for  any sensible choice of the error for smaller systems, there will be a system size whose range of possible expectation values is smaller than this error, so any single sample of this size would be trivially within this error, making the sample complexity at this size be just 1. In other words, we expect that as $n \to \infty$, the expectation values $y$ become independent of the particular parameter choices $x$, and so one can trivially predict $y_*$ for a new parameter choice $x_*$ by simply outputting the value $y$ at a single training data point. Note that this requires using LASSO with a possibly non-zero $y$-intercept $y_0$, as otherwise predicting a constant becomes a non-trivial task.

To remedy this behaviour when demonstrating the scaling complexity of this algorithm, one can utilise few different approaches. One possibility would be based on classical shadows \cite{Huang_2020}, similarly to as what was done in \cite{Huang_2022,lewis2023improved}, to just consider and predict each of the local terms separately, calculate the RMS error for each of them separately, and finally average over all of these errors. This approach indeed avoids the global effects and gives much more information about the whole chain, but is a priori less accurate when predicting observables such as the ground state energy, as its guaranteed prediction error is a constant $\epsilon$, while the size of the range of possible expectation values will decrease with $n$ below this error bound, and hence training on a single sample of the entire observable eventually becomes trivially more accurate. We used a similar approach when considering an equivariant system in Section \ref{sec:HeisenbergPeriodic}, where we trained a single ML model on a single geometrically local observable, but then used it for predicting individual local observables over the whole Heisenberg chain to demonstrate a constant sample complexity.

A different approach, which we used for non-equivariant systems, and which is advantageous when predicting global observables such as the ground state energy directly, is to introduce a scaling of the observable, which is chosen such that the size of the range of sampled expectation values and their standard deviation don't scale with the system size. This would mean considering a new observable $O' \coloneqq \sqrt{n} \cdot O$, where $O$ is a correctly normalised observable $\|O\| = \mathcal{O}(1)$ which represents some type of an average behaviour over the whole system. 

The factor of $\sqrt{n}$ can be argued to be appropriate in the following way: Assuming that $O$ is averaging over a sufficiently large disordered sample, such as the whole system, one could observe the physical phenomenon of self-averaging. Looking at the expectation values of the individual terms $O_I$'s making up the whole observable $O= \sum O_I$, assuming them to be weakly correlated, and hence that we can apply generalised central limit theorem as per the discussion in Appendix \ref{sec:CLTGlobalObservables}, we would find that $y = f(O,x) = \Tr (O\rho(x))$ is distributed according to a normal distribution with mean $\mu$ and standard deviation $\frac{\sigma}{\sqrt{n}}$, denoted as $y \sim \mathcal{N}\left(\mu,\frac{\sigma^2}{n}\right)$, where $\mu = \mathcal{O}(1)$ and $\sigma = \mathcal{O}(1)$ do not scale with $n$. Hence, for $O' \coloneqq \sqrt{n} \cdot O$, we would get that $y' \sim \mathcal{N}(\sqrt{n} \mu ,\sigma^2)$ is indeed a random variable whose standard deviation $\sigma$ does not scale with $n$ as we wanted, though its mean $\sqrt{n}\mu$ does grow with $n$.

For example, in the cases of the Heisenberg and Ising model considered in the following sections, when measuring the ground state energy, this scaling means that we have looked at $O'= \frac{H}{\sqrt{n}}$ instead of $O = \frac{H}{n}$. Note that this means $\| O'\| = \mathcal{O}(\sqrt{n}) > \mathcal{O}(1)$, which implies that the predicted complexity scaling isn't naïvely applicable. But we can further show that the rescaling by $\sqrt{n}$ does not affect the sample complexity for those suitable observables $O$ which exhibit the self-averaging behaviour. Indeed, when looking at the observable $O'_\mu \coloneqq \sqrt{n}(O-\mu)$, we have that the corresponding expectation value $y'_\mu$ is distributed as $y'_\mu \sim \mathcal{N}(0,\sigma^2)$, which does not depend on $n$.
This means that this observable is correctly normalised and hence the sample complexity scaling is indeed applicable in this case. But this observable differs from $O'$ only by a shift of $\sqrt{n}\mu$, which for any given $n$ represents only a constant translation of the outcomes, which does not affect the sample complexity of the machine learning model nor its prediction error, as we can simply determine the $y$-intercept $y_0$ as the average value of the outputs, and the remaining weights of the model would then be the same as for a new data set $\left\{x^{(i)},y^{(i)}-y_0\right\}$, which has zero mean. But that makes the predicted efficient sample complexity also applicable for the case of $O'$, even though it is not correctly normalised. 

This behaviour can also be reworded in terms of dependency of the sample complexity on the error, as predicting $O$ with error $\epsilon$ is analogous to predicting $O'$ with error $\epsilon'= \sqrt{n}\epsilon$. But since the predicted sample complexity with error dependence given by Equation \eqref{eq:ScalingComplexityEpsilonDependence} is applicable to $O'$ with error $\epsilon'=\sqrt{n}\epsilon$, this provides a significant sample complexity reduction for observables $O$ exhibiting the self-averaging behaviour, as the dependence of their sample complexity on $\epsilon$ gets scaled by $\sqrt{n}$.

\subsection{Heisenberg model}
\label{sec:Heisenberg}

The Heisenberg model considered is given by the Hamiltonian
\begin{equation}H = \sum\limits_{\langle i j \rangle} J_{ij} (X_i X_j + Y_i Y_j + Z_i Z_j),\end{equation}
where $\langle i j \rangle$  represents nearest neighbours on the lattice. Here we will use only a 1D chain with either open or periodic boundary conditions.

The coupling parameters $J_{ij}$ were sampled uniformly at random from $[0,2]$ as in \cite{lewis2023improved}. In this section, we consider a local observable $C_{ij} = \frac{1}{3} (X_i X_j + Y_i Y_j + Z_i Z_j)$; depending on the situation, either at some specific edge (i.e. $C_{01}$), or taking some sort of average over all of the nearest neighbours, as will be described in the separate subsections.

\subsubsection{Open boundary conditions}

For the following simulations, we've looked at a fixed RMS error $\epsilon = 0.55$ when predicting the ground state energy $O = \frac{H}{\sqrt{n}} = \frac{3}{\sqrt{n}}\sum\limits_{i=0}^{n-2} J_{i,i+1}C_{i,i+1}$, where the normalisation was chosen such that the standard deviations in this regime would not scale with the system size, as was discussed in Section \ref{sec:ShrinkingPhenomenon}.

The presented data on Figure \ref{fig:HeisenbergOpen}, obtained for the number of training samples needed for a given system size, are illustratively fitted with a function of the form $N = a\cdot \log(n) + b + \frac{c}{n}$ (for $n \geq 8$), showing a good correspondence to the predicted asymptotic sample complexity as given by Theorem \ref{thm:sample_complexity_single_obs}. Each shown data point corresponds to a single trial with $N_\text{test} = 40$ test samples, using $\delta = 4$. Note that the variance in these results is caused primarily by our testing method rather than the algorithm itself, as each set of $N_\text{test}$ samples can have noticeably different standard deviation, even though the standard deviation over all parameter choices is normalised. Hence increasing $N_\text{test}$ would decrease the variance of our results significantly.

\subsubsection{Periodic boundary conditions}
\label{sec:HeisenbergPeriodic}

To demonstrate the absence of scaling in this case, as discussed in Section \ref{sec:SampleComplexityReductionEquivariance}, we look into predicting each $C_{i,i+1}$ individually (where the addition is understood modulo $n$), but we train the model only on $C_{01}$ and then use the same model for all $i$, simply by cycling through the input parameters. 

Note that when training the model for $C_{01}$ with just $\mathcal{O}(1)$ samples, the set $\mathcal{S}$ of non-trivial interactions of $O$ contains a single element, $I = \{0,1\}$, and so the concatenated feature vectors as defined by equations \eqref{eq:FullFeatureVector} and \eqref{eq:FeatureVectorAllIOnep} are the same as the one local feature vector given by Equation \eqref{eq:FeatureVectorOneI} with $I = \{0,1\}$.

Here we fix a number of training samples at $N = 40$, and we measure the RMS error over all of the testing samples for each $C_{i,i+1}$ individually, which we subsequently average over all $i$'s. And given the sample complexity reduction, we would expect the RMS error to be asymptotically constant. These results are presented on Figure \ref{fig:HeisenbergPeriodicRMSE}, which indeed shows that the error did not grow between the $n = 16$ and $n = 128$ site cases. Note that saying that the error does not grow with the system size for a fixed number of training samples is tantamount to saying that the number of training samples needed to obtain a fixed error does not grow with the system size, affirming the predicted sample complexity reduction.

\begin{figure}[H]
\centering

\begin{subfigure}[t]{0.5\textwidth}
\centering
\includegraphics[width = \textwidth]{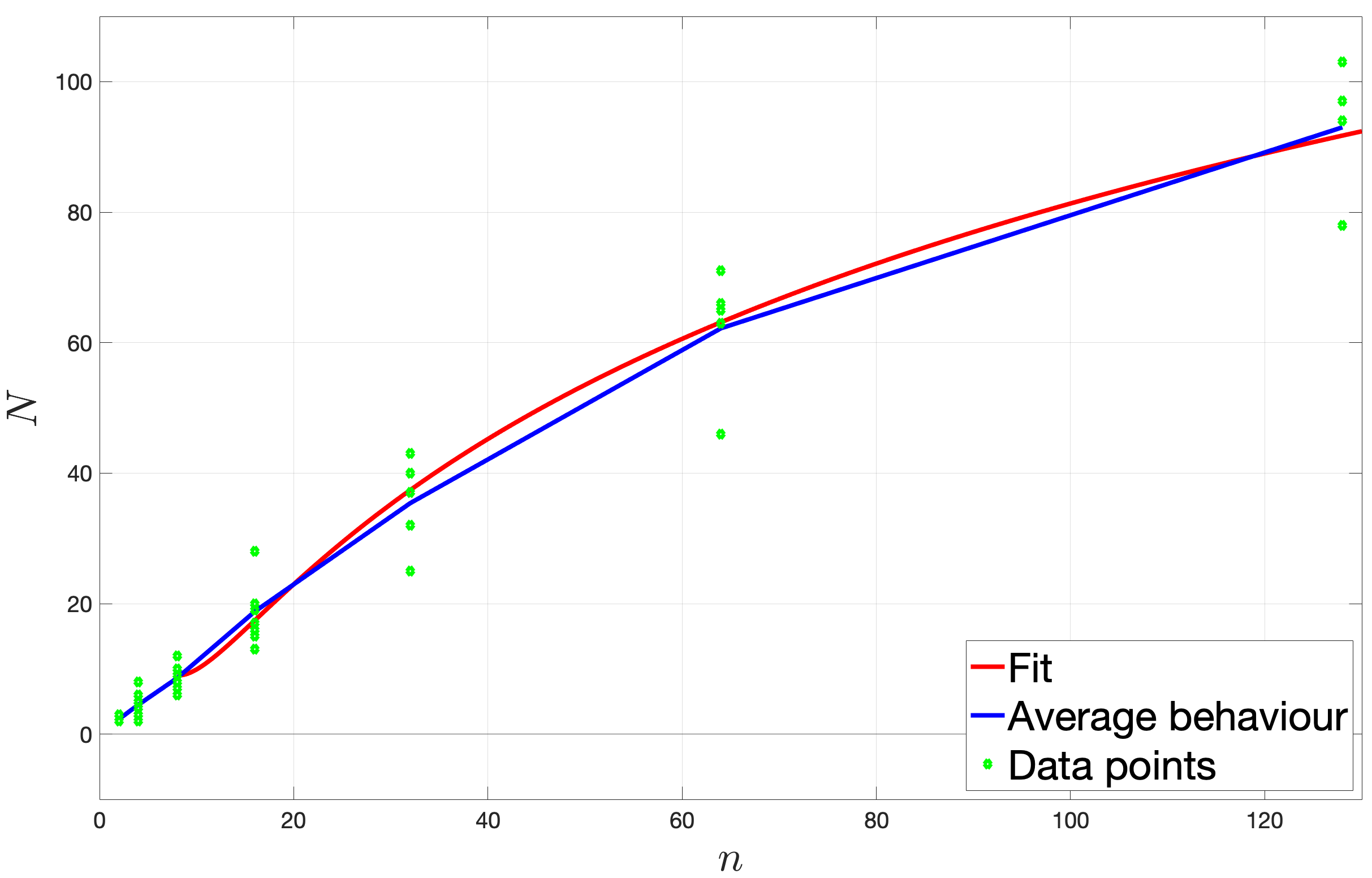}
\caption{Open boundary conditions}\label{fig:HeisenbergOpen}  
\end{subfigure}%
\begin{subfigure}[t]{0.5\textwidth}
\centering
\includegraphics[width = \textwidth]{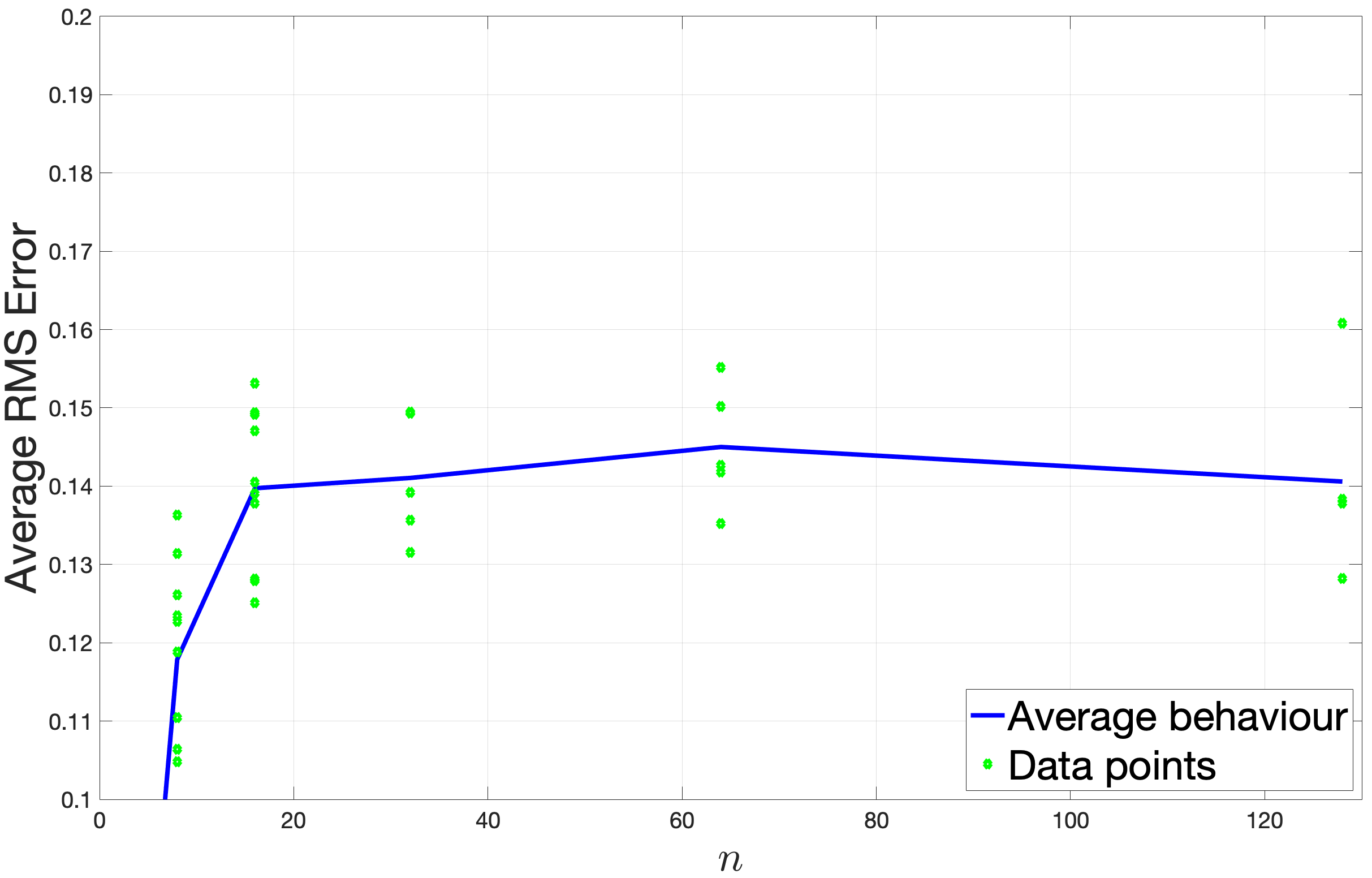}    
\caption{Periodic boundary conditions}\label{fig:HeisenbergPeriodicRMSE}
\end{subfigure}

\caption{\textbf{Results for Heisenberg model.} Each shown data point corresponds to a single trial with $N_\text{test} = 40$ randomly generated test samples. 
\textbf{(a)} Plot of the number of training samples needed to obtain a fixed additive RMS error of $\epsilon = 0.55$ when measuring $\frac{H}{\sqrt{n}}$ in the Heisenberg chain with open boundary conditions, using $\delta = 4$. The logarithmic fit is obtained from all of the data points starting from 8 qubits. \textbf{(b)} Plot of the average RMS error for a fixed number of training samples $N = 40$ when measuring all $C_{i,i+1} = \frac{1}{3} (X_i X_{i+1} + Y_{i} Y_{i+1} + Z_i Z_{i+1})$ individually, but training only on $C_{01}$, in the Heisenberg chain with periodic boundary conditions, using $\delta = 4$.}
\end{figure}

\subsection{Long-range Ising model}

The long-range Ising model considered here is given by the Hamiltonian 
\begin{equation}H = \sum\limits_{i<j} \frac{1+J_i J_j}{d(i,j)^\alpha} Z_i  Z_j + \sum\limits_i h_i  X_i,\end{equation}
where $d(i,j)$ is the distance between sites $i$ and $j$. Here we will only consider a 1D chain with open boundary conditions, so that $d(i,j) = |i-j|$.
This model was studied theoretically in \cite{Juhasz_2014}. It can be realised experimentally, especially in the case of dipolar interactions at $\alpha=3$ studied below, in trapped ions \cite{hauke2010complete} and polar molecules \cite{yan2013observation}.

For the following simulations, the parameters $J_i$ were sampled uniformly at random from $[0,2]$, while all the $h_i$'s were taken to be $h_i = \e$ for simplicity while also maintaining a spectral gap of a reasonable size. The observable was taken to be $O = \frac{H}{\sqrt{n}}$, where the normalisation was again chosen such that the standard deviations in this regime would not scale with the system size, based on the computation of the norm in Section \ref{sec:Approximation of the energy}
and the discussion in Section \ref{sec:ShrinkingPhenomenon}.

\subsubsection{Open boundary conditions}
\label{sec:IsingChainOpenBC}

For the following simulations, we have looked at a fixed RMS error $\epsilon = 0.3$. Data presented on Figure \ref{fig:IsingOpen} are again illustratively fitted with a function of the form $N = a\cdot \log(n) + b + \frac{c}{n}$ (for $n \geq 8$) to demonstrate the expected asymptotic sample complexity of the algorithm given by Theorem \ref{thm:sample_complexity_single_obs}. Each shown data point again corresponds to a single trial with $N_\text{test} = 40$ test samples, using $\delta = 4$.

\subsubsection{The case of $\alpha \leq 2D$}

In this section, we considered the Ising chain with power-law decay with exponent $\alpha = 1.5$. For this regime, we expect the efficient learning guarantees to not hold up as per the discussion in Appendix \ref{sec:power law alpha le 2d}.

Indeed the scaling we observe, plotted on Figure \ref{fig:IsingBadAlpha}, seems to be noticeably worse than in the previous section with $\alpha = 3$, which is in the correct regime $\alpha > 2D$. As the average behaviour appears to scale seemingly linearly, we illustratively fit the data with a function of the form $N = a\cdot n + b + \frac{c}{n}$. But because of high variance in the obtained data, this fit is not significantly better than the logarithmic one, even though the average behaviour seems rather linear. Further, given the finite sample size and the fact that the number of exponentials needed to fit $x^{-1.5}$ for DMRG simulations is not appreciably larger than to fit $x^{-3}$ to the same precision at this size, we can not definitely conclude that this scaling is not logarithmic.

\begin{figure}[H]
\centering

\begin{subfigure}[t]{0.5\textwidth}
\centering
\includegraphics[width = \textwidth]{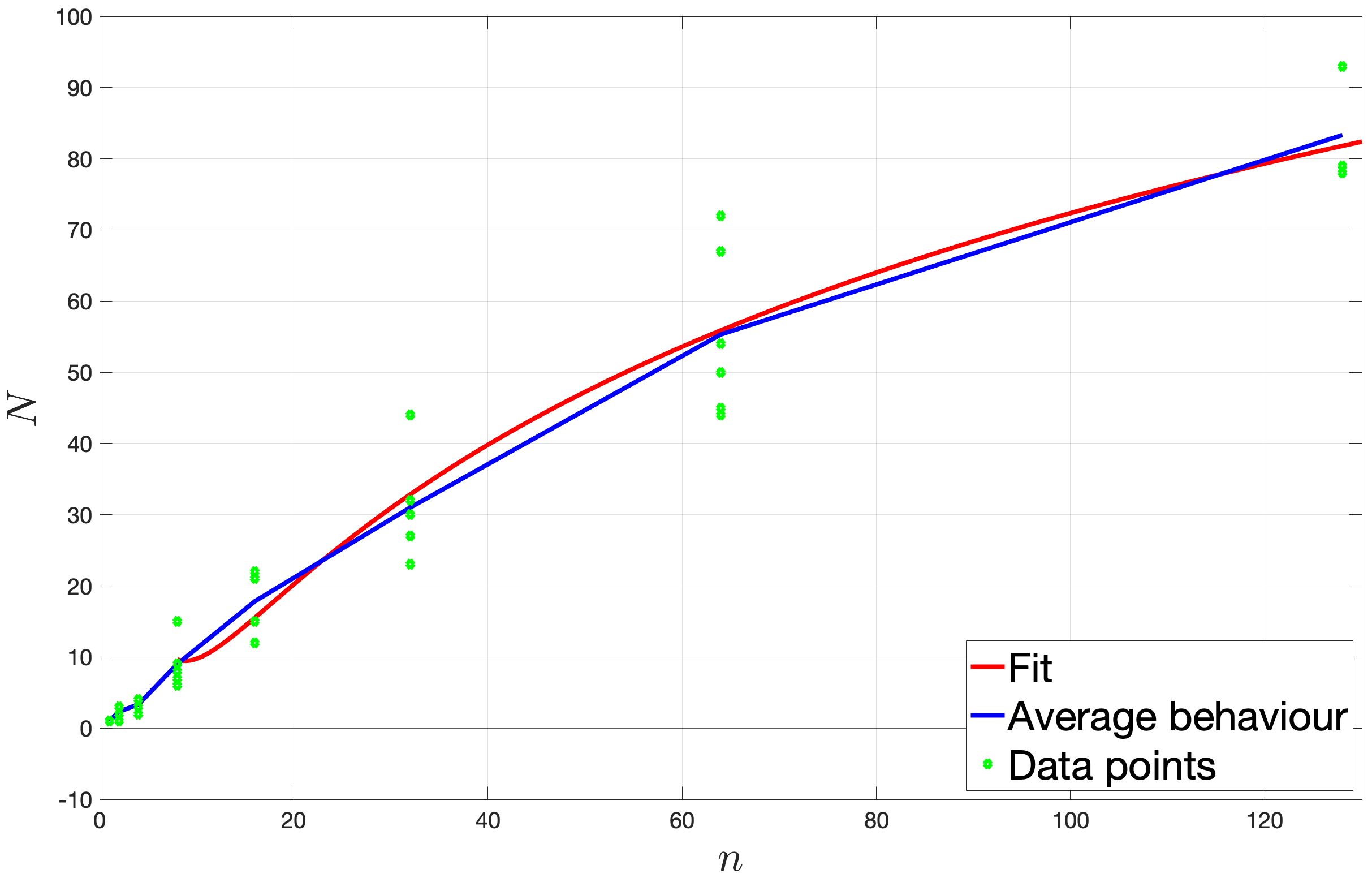}
\caption{$\alpha = 3$}\label{fig:IsingOpen}
\end{subfigure}%
\begin{subfigure}[t]{0.5\textwidth}
\centering
\includegraphics[width = \textwidth]{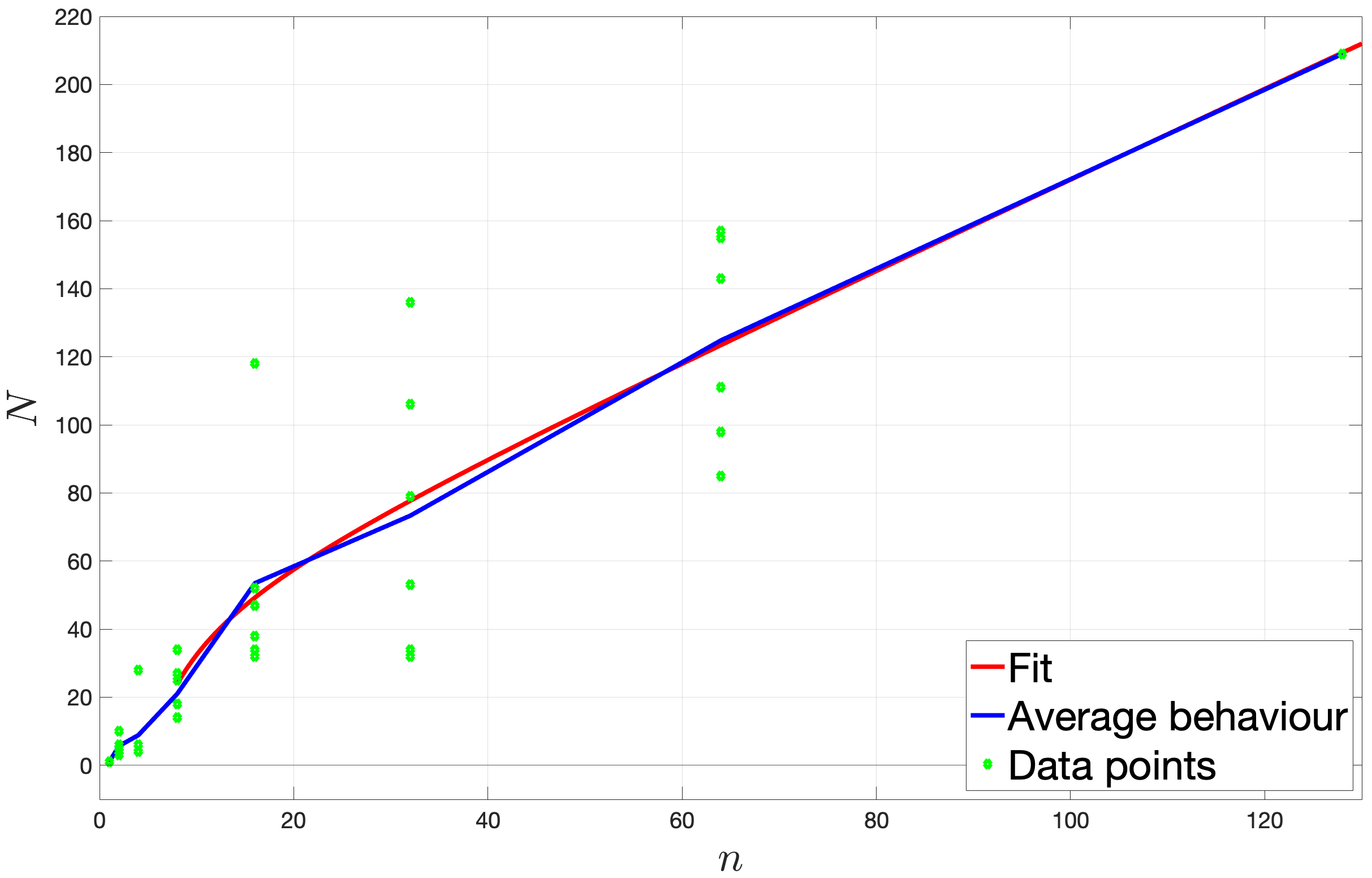}
\caption{$\alpha = 1.5$}\label{fig:IsingBadAlpha}
\end{subfigure}

\caption{\textbf{Results for Ising model.} Each shown data point corresponds to a single trial with $N_\text{test} = 40$ randomly generated test samples. The fitting curves are obtained from all of the data points starting from 8 qubits. \textbf{(a)} Plot of the number of training samples needed to obtain a fixed additive RMS error of $\epsilon = 0.3$ when measuring $O = \frac{H}{\sqrt{n}}$ in the Ising chain with $\alpha = 3$, using $\delta = 4$. \textbf{(b)} Plot of the number of training samples needed to obtain a fixed additive RMS error of $\epsilon = 0.15$ when measuring $O = \frac{H}{\sqrt{n}}$ in the Ising chain with $\alpha = 1.5$, using $\delta = 4$.}
\end{figure}

\subsection{Rydberg atom chain with position disorder}

The Hamiltonian representing 1D chain with open boundary conditions of interacting Rydberg atoms at displaced positions $i + \delta_i$, i.e. with the metric $d(i,j) = |i+\delta_i - j - \delta_j|$, is given by \begin{equation}
   H =  \sum\limits_{i<j} \frac{V}{|i+\delta_i - j - \delta_j|^6} N_i N_j + \sum\limits_i \left( \frac{\Omega}{2} X_i + \Delta N_i \right),
\end{equation} where $N = \frac{I-Z}{2} = |1\rangle\langle 1|$ is the Rydberg number operator. Such a system has been previously studied in e.g. \cite{Marcuzzi2017}. This disordered geometry is illustrated on Figure \ref{fig:Rydberg-Positions}. 
\begin{figure}[H]
    \centering
    \includegraphics[width=\textwidth]{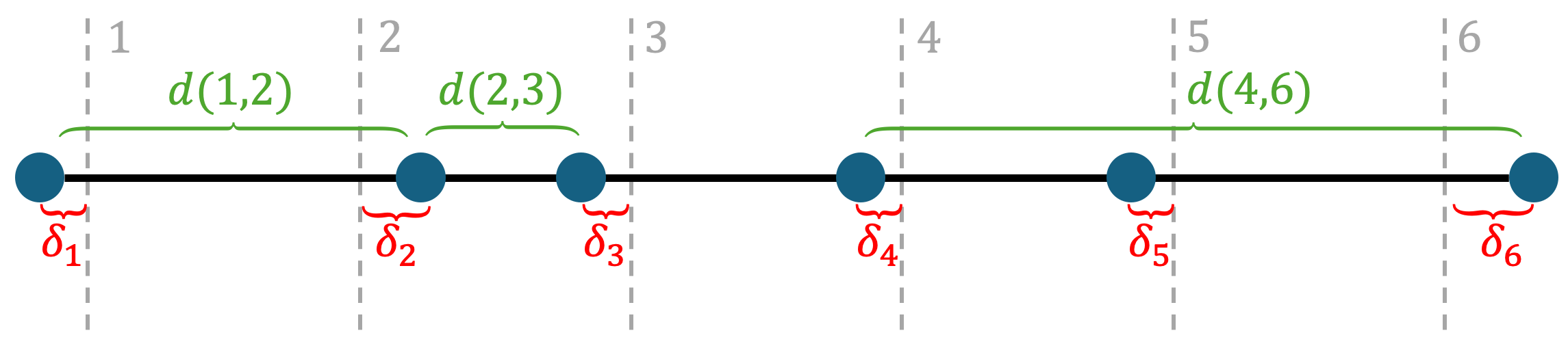}
    \caption{Illustration of the position disorder in the 1D Rydberg chain with 6 sites.}
    \label{fig:Rydberg-Positions}
\end{figure}

For these simulations, the atom displacements $\delta_i$ were sampled uniformly at random from $[-0.25 , 0.25]$. Note that if we were to allow displacements such that $d(i,j) \to 0$, then the magnitude of the energy of such configurations wouldn't be bounded. The repulsive coupling coefficient was taken to be $V = -1$, the detuning $\Delta = 2$, and the Rabi frequency $\Omega = 1$. Data presented on Figure \ref{fig:Plot-Rydberg} were obtained for a fixed error $\epsilon = 2$, and are again illustratively fitted with a function of the form $N = a\cdot \log(n) + b + \frac{c}{n}$ (for $n \geq 8$). Each shown data point again corresponds to a single trial with $N_\text{test} = 40$ test samples, using truncation $\delta = 2$.

\begin{figure}[H]
\centering

\includegraphics[width = \textwidth]{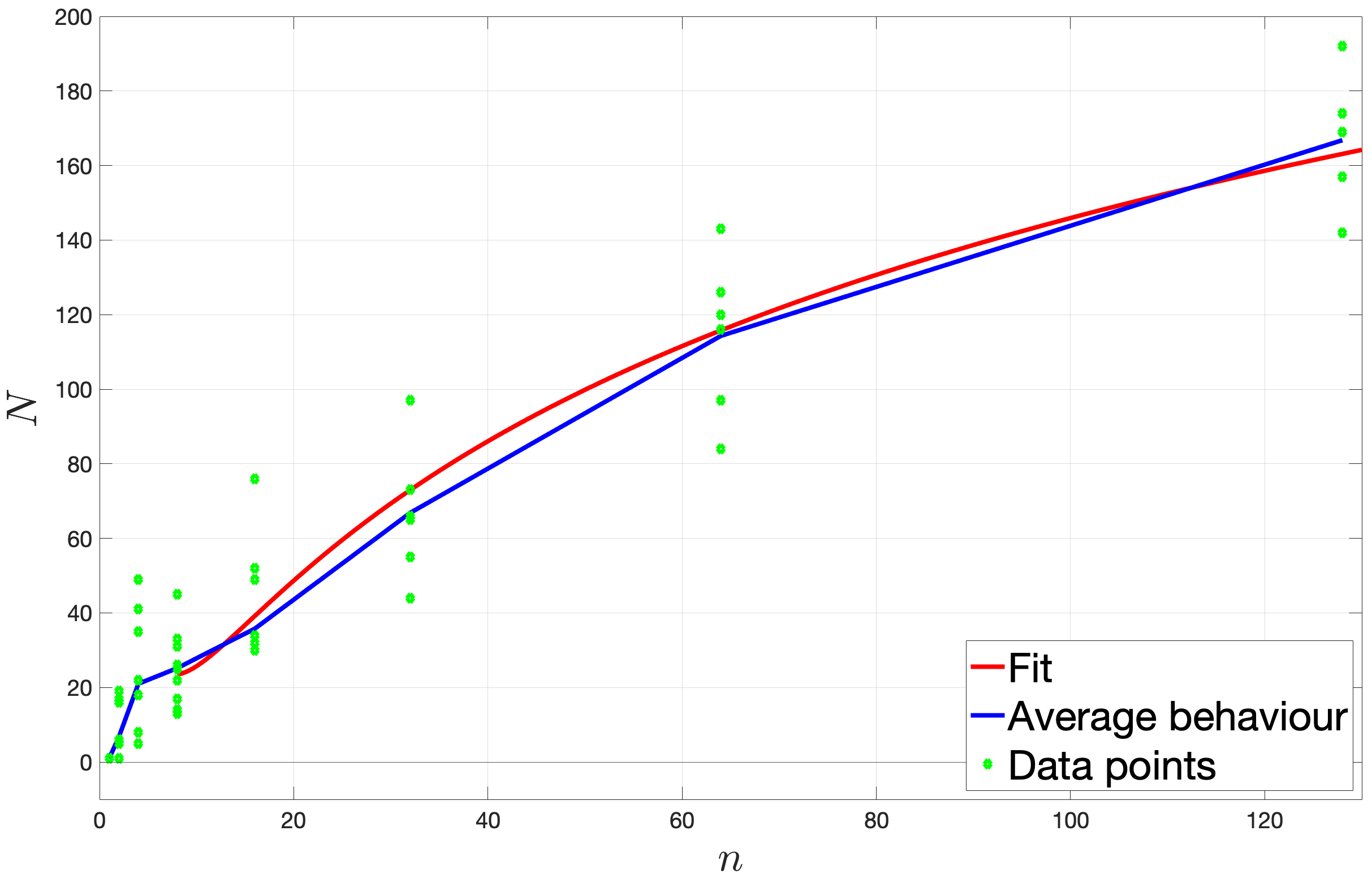}

\caption{\textbf{Results for Rydberg model.} Each shown data point corresponds to a single trial with $N_\text{test} = 40$ randomly generated test samples. The fitting curves are obtained from all of the data points starting from 8 qubits. Plotting the number of training samples needed to obtain a fixed additive RMS error of $\epsilon = 2$ when measuring $O = \frac{H}{\sqrt{n}}$ in the Rydberg chain, using $\delta = 2$.}\label{fig:Plot-Rydberg}
\end{figure}

\section{Conclusions and outlook}

We have shown rigorous guarantees when using a classical machine learning algorithm for predicting ground state properties of Hamiltonians with long range interactions decaying at least polynomially in separation of sites, with the decay exponent $\alpha$ being greater than twice the dimension of the system. Albeit relevant to many systems of interest, this still leaves out important examples with smaller $\alpha$'s, such as the Coulomb interaction. Though our bounds are not applicable for this case, and our simulations do not suggest the efficient logarithmic scaling, it still remains unknown if one can come up with some theoretical guarantees for these stronger interactions, and if the logarithmic scaling is achievable. 

There are also still many other questions of interest when one leaves ground states of gapped Hamiltonians. Is it possible to extend this theory to thermal states? Can one get similar guarantees when working with gapless Hamiltonians, training across different topological phases?
We expect that the results implying  dependence of the expectation value of a $p$-local observable on only the parameters $x$ close to it fail if we 
drop the assumption of a gapped phase.
Indeed if it were true, then we could
efficiently compute ground state properties of NP-hard Hamiltonians, such as the 3SAT one considered in \cite[App. H.1]{Huang_2022}, by reparametrising the 
Hamiltonian in such a way that there exists a parameter choice that would separate the Hamiltonian into $k$-local subsystems, allowing for efficient simulation.

In our work, we have explored the importance of self-averaging effects on the concentration of measure for the expectation values of global observables. We have observed this phenomenon for every correctly normalised global observable we have considered. As this is a property of the systems and observables themselves, it will affect any algorithm learning from those data, not only the specific one we have considered. Is this a universal property for global observables within gapped phases, or would some specific distributions of the parameters with long-range correlations not exhibit this behaviour?

\appendix

\section{Properties of observables and expectation values}
\label{sec:Theoretical guarantees for approximating observables}

\subsection{Dependency of observables on Hamiltonian parameters}

\subsubsection{Exponential decay}
\label{sec:Exponential decay}

We shall use the following result:
\begin{lemma}[\cite{hastings2010locality}]\label{lemma:lr_exp}
Assume that  the interactions are such that
\begin{align}
\label{eq:assumption_lr_exp}
    \sum_{I\ni i,j}
    \|h_I\| 
    \le 
    \lambda 
    \e^{-\mu d(i,j)}
\end{align}
for some positive constants $\lambda, \mu$.
Then, if $d(I,J)>0$, we have the Lieb-Robinson bound, with $s>0$:
\begin{align}
\|[\tau_t(A_J), B_I] \|
\le
2\|A_J\|\,\|B_I\|\,|J|
\e^{2s|t|-\mu d(I,J)}
\,.
\end{align}
\end{lemma}
We start by noting that the assumptions of this Theorem are satisfied if the interactions have exponential decay times a polynomial of the diameter.
\begin{lemma}\label{lemma:hI_exp_decay}
Let $\nu=\mu+\epsilon, \epsilon>0$. If
\begin{align}
    \|h_I\|
    \le
    \e^{-\nu \diam(I)}
    P(\diam(I))
    \,,
\end{align}
with $P$ a positive polynomial, then \eqref{eq:assumption_lr_exp} holds.
    \end{lemma}
\begin{proof}    
We first rewrite:
\begin{align}
    \sum_{I\ni i,j}
    \|h_I\| 
    &\le 
    \sum_{R\ge 1}\e^{-\nu R} P(R)
    \sum_{\ell=2}^k \widetilde{M}_\ell(ij,R)\,,
    \\
        \widetilde{M}_\ell(ij,R)
        &=
        \sum_{i_1,\dots,i_\ell} \delta\left(\max_{1\le a<b\le \ell}(d_{ab}),R-1\right) \id((i_1,\dots,i_\ell)\ni ij)\\
        &\le
        C_\ell
        \left( 
        \delta(d(i,j), R-1)
        +
        [\tfrac{1}{2}\ell(\ell-1)-1]
        \id(d(i,j)\le R-1)
        R^{- 1}  \right)
        R^{(\ell-2) D}  
        \,,
\end{align}
where the inequality is proved in Lemma \ref{lemma:Mtilde} for some $C_\ell>0$.
Defining
\begin{align}
    C_{\text{max}}
    =
    \sum_{\ell=2}^k C_\ell
    \max\left(1,
    \tfrac{1}{2}\ell(\ell-1)-1
    \right)\,,
\end{align}
we have
\begin{align}
    \sum_{\ell=2}^k \widetilde{M}_\ell(ij,R)
    \le 
    R^{(k-2) D}  
    C_{\text{max}}
        \left( 
        \delta(d(i,j), R-1)
        +
        \id(d(i,j)\le R-1)
        \right)
    \,,
\end{align}
and get
\begin{align}
    \sum_{I\ni i,j}
    \|h_I\| 
    &\le 
     \sum_{R\ge 1}\e^{-\nu R} P'(R)
    \left( 
        \delta(d(i,j), R-1)
        +
        \id(d(i,j)+1\le R)
        \right)\\
        &=
     \e^{-\nu (d(i,j)+1)} 
    \left( 
        P'(d(i,j)+1)
        +
        \sum_{R\ge 1+d(i,j)}\e^{-\nu (R-d(i,j)-1)} P'(R)
        \right)
        \,,
\end{align}
for a new positive polynomial $P'(R)$.
The second sum can be performed by changing variables to $r=R-1-d(i,j)$, to give another positive polynomial $P''(d(i,j)+1)$.
The Lemma follows by noting that there exists a constant $C_P$ such that $\e^{-\epsilon (d(i,j)+1)}P(d(i,j)+1)\le C_P$, so that
\begin{align}
    \sum_{I\ni i,j}
    \|h_I\| 
    &\le 
    \e^{-\mu d(i,j)} (C_{P'}+C_{P''})\,.
\end{align}
\end{proof}

We are going to prove the following result:
\begin{prop}\label{prop:approx_exp_decay}
    Assume that the hypothesis of Lemma \ref{lemma:hI_exp_decay} is satisfied, and that
    \begin{align}
    \|\nabla_I h_I\|
    \le
    \e^{-\nu \diam(I)}
    \,,    
    \end{align}    
    and let $\nu,\mu,s$ be as in Lemmas \ref{lemma:hI_exp_decay} and \ref{lemma:lr_exp}.
    Then for 
    $0<a<\frac{1}{2s}(\mu - \log k)$, define
    \begin{align}
    b=\gamma a\,,\quad 
    \mu' = \mu-2sa \,.
    \end{align}
    Also, assume that 
    \begin{align}
        \delta > \max(b^{-1}\xi_*, b^{-1}\eta_*)
    \end{align}
    with $\xi_*$ as in Lemma \ref{lemma:2.6} and $\eta_*$ such that the function $x^{D+9}u_{2/7}(x)$ is decreasing for $x>\eta_*$ and satisfies the hypothesis of Lemma \ref{lemma:extension_2.5} for $s=D+9$.
    Then, there exist positive constants $c_0,c_1',c_2'$ so that:
    \begin{align}
    \label{eq:diff_OI_exp}
    &\frac{|f(O_I, x')-f(O_I, x)|}{\|O_I\|}
    \le
    c_0\delta^{Dk+9}\e^{-\nu \delta}
    +
    c_1'
    \delta^{D-1}
    \e^{-\mu'\delta}
    +
    c_2' (b\delta)^{2D+20}u_{2/7}(b\delta)\,.
\end{align}
\end{prop}
\begin{proof}  
Define
\begin{align}
    t_* = a d(I, J)
    \,,
\end{align}
and note that the Lieb-Robinson bound holds for $|J|
\e^{2s|t|-\mu d(I,J)} < 1$, which holds for $|t|<t_*$ for all $I,J$ with $d(I,J)>0$ since we chose 
\begin{align}
    a<\min_{I,J}
    \frac{1}{2s}\left(\mu -\frac{\log |J|}{d(I, J)}\right)
    =
    \frac{1}{2s}(\mu - \log k)
    \,.
\end{align}
Then we can compute ${\cal I}_1$ appearing in \eqref{eq:I1_I2} as:
\begin{align}
    {\cal I}_1
    \le
    2|J|
    \e^{-\mu d(I,J)}
    \int_0^{t_*} \e^{2s t}\dd t
    =
    s^{-1}
    |J|
    \e^{-\mu d(I,J)}
    (\e^{2s t_*}-1)
    \le
    s^{-1}
    |J|
    \e^{-\mu' d(I,J)}
    \,.
\end{align}
We have that, denoting $R=d(I,J)$,
\begin{align}
    &\|\nabla_J
    f(O_I,x)
    \|
    \le
    \|\nabla_Jh_J\|\,\|O_I\|
    \left(
    c_1/q_*
    \e^{-\mu' R}
    +
    4 G(b R)
    \right)
    \,, \quad c_1 = s^{-1}|J|q_*\,,
\end{align}
with $G$ as in Lemma \ref{lemma:2.6}.
Next, under the assumption
(all the calculations below can be easily adopted to the case where we have a decay of $\e^{-\nu \diam(I)}$ times a polynomial of $\diam(I)$):
\begin{align}
    \delta>b^{-1}\xi_*\,,\quad 
    \|\nabla_I h_I\|
    \le
    \e^{-\nu \diam(I)}
    \,,    
\end{align}
we have that, with $c_2 = \tfrac{4}{\gamma}
    130\e^2 q_*$,
\begin{align}
    &\frac{|f(O_I, x')-f(O_I, x)|}{\|O_I\|}
    \le S_1 + S_2\,,\\
    &S_1=
    \sum_{J\in {\cal P}_k(\Lambda)}
    \id\left(d(I,J)\le \delta
    \wedge
    \diam(J)> \delta
    \right)
    \e^{-\nu \diam(J)}
    \left(
    c_1
    \e^{-\mu' d(I,J)}
    +
   4G(b d(I,J))
    \right)\,,
    \\
    &S_2=
    \sum_{J\in {\cal P}_k(\Lambda)}
    \id\left(
    d(I,J)>\delta\right)
    \e^{-\nu \diam(J)}
    \left(
    c_1
    \e^{-\mu' d(I,J)}
    +
    c_2 (b d(I,J))^{10}u_{2/7}(b d(I,J))
    \right)\,.
\end{align}
Denoting $R_*=b^{-1}\xi_*$,
we can rewrite these as follows:
\begin{align}
 &S_1=
    \sum_{r>\delta}
    \e^{-\nu r}
    \left( 
    c_1
    \sum_{R\le \delta}
    \e^{-\mu' R}
    +    
    \sum_{R\le R_*}
    2\frac{Kq_*}{\gamma}
    +
    c_2
    \sum_{R_*<R\le \delta}
    (b R)^{10}u_{2/7}(b R)
    \right)
    M(I, r, R)\,,
    \\
&S_2=
\sum_{R>\delta}\sum_{r \ge 1}
    \e^{-\nu r}
    \left(
    c_1
    \e^{-\mu' R}
    +
    c_2 (b R)^{10}u_{2/7}(b R)
    \right)
    M(I, r, R)\,,
    \\
&M(I, r, R)=
    \sum_{J\in {\cal P}_k(\Lambda)}
    \delta(r, \diam(J))
    \delta(R, d_{IJ})
    \le 
    C_{\text{sum}}
    r^{(k-1)D-1}
    R^{D-1}
    \,,
\end{align}
where the bound is according to Corollary \ref{corollary:M}.
Putting things together 
we have
\begin{align}
&S_1\le \Phi
    \left(
    c_1
    \sum_{R\le \delta}
    R^{D-1}
    \e^{-\mu' R}
    +
    \sum_{R\le R_*}
    2\frac{Kq_*}{\gamma} R^{D-1}
    +
    c_2 
    \sum_{R_*<R\le \delta}
    R^{D-1}(b R)^{10}u_{2/7}(b R)
    \right)
    \,.
\end{align}
Setting $d=D(k-1)-1$, $\Phi$ 
can be upper bounded by
\begin{align}
    \label{eq:sum_r_exp}
    \Phi/C_{\text{sum}}
    =
    \sum_{r>\delta}
    r^{d}
    \e^{-\nu r}   
    =
    \e^{-\nu \delta}   
    \sum_{y>0}
    (y+\delta)^{d}
    \e^{-\nu y}   
    =
    \e^{-\nu \delta}   
    \sum_{\ell=0}^{d}
    \delta^\ell
    \binom{d}{\ell}
    \sum_{y>0}
    y^{d-\ell}    
    \e^{-\nu y}   
    \le 
    C' 
    \e^{-\nu \delta}   
    \delta^{d}\,.
\end{align}
Now we bound each term in parenthesis in $S_1$.
We use that $\e^{-\mu' R},u_{2/7}(bR)\le 1$, to see that the term in parenthesis is bounded by
\begin{align}
c_1 \delta^D + c_1' + c_2 b^{10} \delta^{D+10}
\le \max(c_1, c_1', c_2 b^{10})
\delta^{D+10}
\,,
\end{align}
for a constant $c_1'>0$.
Thus we get:
\begin{align}
S_1\le
    c_0
    \delta^{D(k-1)-1+D+10}
    \e^{-\nu\delta}
    =
    c_0
    \delta^{Dk+9}
    \e^{-\nu\delta}
    \,,\quad
    c_0
    =
    C_{\text{sum}} C' \max(c_1, c_1', c_2 b^{10})
    \,.
\end{align}
Note that the case $R=0$, which was excluded when deriving the Lieb-Robinson bound is included in the sum, and just contributes as an additive constant.

Now we get back to $S_2$.
Defining
\begin{align}
    g(R) 
    =
    R^{D-1}
    \left(
    c_1
    \e^{-\mu' R}
    +
    c_2 (b R)^{10}u_{2/7}(b R)
    \right)\,,
\end{align}
it is equal to
\begin{align}
    S_2=
    C_{\text{sum}}
\sum_{r \ge 0}
r^{D(k-1)-1}
    \e^{-\nu r}
\sum_{R >\delta}
g(R)\,.
\end{align}
The sum over $r$ gives the constant
\begin{align}
    \Phi
    =
    \sum_{r \ge 1}
r^{D(k-1)-1}
    \e^{-\nu r}\,.
\end{align}
The first sum over $R$ can be bounded as in \eqref{eq:sum_r_exp}:
\begin{align}
    \sum_{R>\delta}
    R^{D-1}
    \e^{-\mu' r}   
    \le 
    C'
    \delta^{D-1}
    \e^{-\mu'\delta}\,.
\end{align}
The remaining sum is
\begin{align}
\sum_{R >\delta}
(bR)^{D+9}u_{2/7}(b R)\,.
\end{align}
We are going to use Lemma \ref{lemma:extension_2.5}.
Since $b\delta>\eta_*$, we can bound the sum as
\begin{align}
    \sum_{R >\delta}
(bR)^{D+9}u_{2/7}(b R)
    &\le
    \int_{\delta}^\infty
    \dd x
    (bx)^{D+9}u_{2/7}(bx)
    =
    b^{-1}
    \int_{b\delta}^\infty
    \dd \eta
    \eta^{D+9}u_{2/7}(\eta)\\
    &\le 
    \frac{7}{2b}
    (2D+21) (b\delta)^{2D+20}u_{2/7}(b\delta)\,.
\end{align}
Putting things together we have
\begin{align}
    &S_2\le    
    c_1'
    \delta^{D-1}
    \e^{-\mu'\delta}
    +
    c_2' (b\delta)^{2D+20}u_{2/7}(b\delta)\,,
    \\
    &c_1'=C_{\text{sum}}
    \Phi C' c_1
    \,,\quad
    c_2'=
    C_{\text{sum}}
    \Phi \frac{7}{2b}
    (2D+21)
    c_2\,,
\end{align}
which proves the result.
\end{proof}

We have the following corollary:

\begin{corollary}\label{cor:delta_eps_exp}
    Under the same assumptions of Proposition \ref{prop:approx_exp_decay}, we have 
    \begin{align}
    |f(O_I, x')-f(O_I, x)|
    \le
    \epsilon \|O_I\|
    \end{align}
    if $0<\epsilon\le \e^{-1}$ and
    \begin{align}
    \label{eq:delta_eps_exp}
    \delta \ge c \log^2(\tilde{c}/\epsilon)\,,
    \quad 
    \tilde{c} = 3\max(c_0,c_1',c_2')
    \,.
\end{align}
\end{corollary}
The proof follows from Lemmas \ref{lemma:6} and \ref{lemma:7} to bound each term in \eqref{eq:diff_OI_exp} by $\epsilon/3$. 

\subsubsection{Power law with $\alpha\in (2D,2D+1)$}
\label{sec:power law alpha gt 2d}

We start to discuss the case of power law interactions with exponent $\alpha\in (2D,2D+1)$ and then show how the results obtained in this case can be extended to  $\alpha\ge 2D+1$.
We shall use the following Lieb-Robinson bound:

\begin{lemma}\label{lemma:lr_alpha_ge_2d}
Assume that
\begin{align}
    \label{eq:H_decay_power_law}
    H=\sum_{I\in {\cal P}_k(\Lambda)} h_I\,,\quad 
    \sum_{I\ni i,j}\|h_I\|\le g d(i,j)^{-\alpha}\,,
\end{align}
for a positive constant $g$.
    For any $\alpha\in (2D,2D+1)$,  
    let
\begin{align}
    \epsilon_*= \frac{(\alpha-2D)^2}{(\alpha-2D)^2+\alpha-D}\,,\quad 
    \alpha' = \alpha-2D-\epsilon
    \,,\quad
    \beta = \frac{\alpha-D}{\alpha-2D}-\frac{\epsilon}{2}\,,
\end{align}
with $\epsilon\in (0, \epsilon_*)$.
Then there exist constants $c,C_1,C_2\ge 0$ s.t. for all $0\le t\le c r^{\alpha'}$ and $|I|, |J|=\mathcal{O}(1)$
    \begin{align}
    \|[A_I(t),B_J]\|
&\le
\|A_I\|\|B_J\|
f(t, d(I,J)) \,,\quad
f(t,r)=
    C_1 \left( 
    \frac{t}{r^{\alpha'}}
    \right)^{\beta}
    +
    C_2
    \frac{t}{r^{\alpha-d}}
    \,.
\end{align}
\end{lemma}
This is proven in Corollary \ref{corollary:lr_alpha_ge_2d}
based on results in the literature.
The proof relies on a conjecture that was made in \cite{Tran_2021}, which applies to the case at hand of a Hamiltonian given by a sum of terms with at most $k$-body interactions, while only the case $k=2$ was proved in \cite{Tran_2021}.

The assumptions of the Lieb-Robinson bound are satisfied if the interactions decay as in the following lemma:

\begin{lemma}\label{lemma:beta_ell}
    If for all $I$ such that $|I|>1$,
    \begin{align}
        \|h_I\|\le g_{|I|} \diam(I)^{-\gamma(|I|)}\,,\quad 
        \gamma(\ell) = (\ell-2)D + \alpha\,,
    \end{align}
    for some constant $g_\ell>0$,
    then \eqref{eq:H_decay_power_law} holds.
\end{lemma}
\begin{proof}
    Let us fix $|I|=\ell$ and consider the sum over only those sets:
    \begin{align}
        f_\ell(i,j) 
        &\coloneqq \sum_{I,|I|=\ell} \|h_I\| \id(i,j\in I)
        \le  g_{\ell} \sum_{I,|I|=\ell} \diam(I)^{-\gamma(\ell)}\id(i,j\in I)\\
        &\,= g_{\ell} 
        \sum_{R\ge 1} R^{-\gamma(\ell)} \widetilde{M}_\ell(ij,R)
        \,,\quad 
        \widetilde{M}_\ell(ij,R)=
        \sum_{i_1,\dots,i_\ell} \delta\left(\max_{1\le a<b\le \ell}(d_{ab}),R-1\right) \id(i,j\in I)\,.
    \end{align}
    Here we denoted $d_{ab}=d(i_a,i_b)$. 
    From Lemma \ref{lemma:Mtilde}, we know \begin{align}
        \widetilde{M}_\ell(ij,R)
        \le 
        C_\ell
        \left( 
        \delta(d(i,j), R-1)
        R^{(\ell-2) D}  
        +
        [\tfrac{1}{2}\ell(\ell-1)-1]
        \id(d(i,j)\le R-1)
        R^{(\ell-2) D - 1}  \right)
        \,,
    \end{align}
    and, denoting $C_\ell' = g_\ell C_\ell$,
    \begin{align}
        f_\ell(i,j) 
        = C_{\ell}' 
        \sum_{R\ge 1} R^{-\gamma(\ell)+(\ell-2) D} \left( 
        \delta(d(i,j), R-1)
        +
        [\tfrac{1}{2}\ell(\ell-1)-1]
        \id(d(i,j)\le R-1)
        R^{- 1}  \right)\,.
    \end{align}
    Now if we plug in the value $\gamma(\ell)=(\ell-2)D+\alpha$, we get for the first term
    \begin{align}
        \sum_{R\ge 1} R^{-\alpha}
        \delta(d(i,j), R-1)
        =
        (1+d(i,j))^{-\alpha}
    \end{align}
    and the second term
    \begin{align}
        \sum_{R\ge 1} R^{-1-\alpha}
        \id(d(i,j)\le R-1)
        = 
        \sum_{R\ge 1+d(i,j)} R^{-1-\alpha}
        \le \int_{d(i,j)}^\infty x^{-1-\alpha}\dd x 
        = \frac{1}{\alpha} d(i,j)^{-\alpha}\,.
    \end{align}
    The Lemma follows since this holds for all $\ell$.
\end{proof}

\begin{prop}\label{prop:approx_alpha_gt_2d}
Let $\alpha\in (2D,2D+1)$ and assume that the hypothesis of Lemma \ref{lemma:beta_ell} are satisfied and that for all $I$ such that $|I|>1$,
    \begin{align}
    \|\nabla_I h_I\|
    \le
    \diam(I)^{-\gamma(|I|)}
    \,,    \quad 
    \gamma(\ell) = (\ell-2)D+\alpha\,.
    \end{align}     
    With the definitions of Lemma \ref{lemma:lr_alpha_ge_2d}, let 
    \begin{align}
    \label{eq:mu}
    b=ac\gamma\,,\quad 
    \eta = \frac{1}{2}\frac{\alpha-2D}{2\alpha-2D-1}\,,\quad 
    \epsilon=\frac{\epsilon_*}{2}\,,\quad 
    \mu = \alpha'\eta\,,\quad 
    \nu = \alpha' \beta - \mu (\beta+1)\,.
    \end{align}
    Also, assume that 
    \begin{align}
        \delta > \left(b^{-1}\max(\xi_*,\eta_*)\right)^{1/\mu}
    \,,
    \end{align}
    with $\xi_*$ as in Lemma \ref{lemma:2.6} and $\eta_*$ is such that the function $t^{9+\frac{D}{\mu}}u_{2/7}(t)$ is monotonically decreasing for $t>\eta_*$ and the hypothesis of Lemma \ref{lemma:extension_2.5} are satisfied.
    Then, there exist positive constants $C_1,C_2$ such that:
    \begin{align}
\frac{|f(O_I, x')-f(O_I, x)|}{\|O_I\|}
\le   
C_1
    \delta^{-\nu+D}
+
C_2
(b\delta^\mu)^{21 + 2\frac{D}{\mu}}
    u_{2/7}(b\delta^\mu)\,.
\end{align}
\end{prop}
\begin{proof}

We first note that by assumption we can use the Lieb-Robinson bound of Lemma \ref{lemma:lr_alpha_ge_2d}, and discuss the bounds on the various exponents used there.
For $\alpha\in (2D,2D+1)$ we have 
\begin{align}
    0<\epsilon_*<\frac{1}{D+2}
\end{align}
since $\epsilon_*$ is an increasing function of $\alpha$ with $\lim_{\alpha\to 2D}\epsilon_*=0$, $\lim_{\alpha\to 2D+1}\epsilon_*=1/(D+2)$.
Next, note that $\alpha'$ is a decreasing function of $\epsilon$, so we can bound its range by first using the bounds on $\epsilon$ and then on $\alpha$:
\begin{align}
    \alpha'>\alpha'|_{\epsilon=\epsilon_*}=
    \alpha-2D-\epsilon_*>\alpha-2D>0
    \,,\quad 
    \alpha'<\alpha'|_{\epsilon=0}=
    \alpha-2D<1\,.
\end{align}
Similarly for $\beta$:
\begin{align}
    \beta>\beta|_{\epsilon=\epsilon_*}>\lim_{\alpha\to 2D+1}\beta|_{\epsilon=\epsilon_*}=D+1-\frac{1}{2(D+1)}\,,\quad 
    \beta<\beta|_{\epsilon=0}=\frac{\alpha-D}{\alpha-2D}<\infty\,.
\end{align}
where we used that $\beta|_{\epsilon=\epsilon_*}$ is the sum of two decreasing functions of $\alpha$, so their minimum is at $2D+1$.
We can do the same analysis for $\alpha'\beta$ using bounds for each:
\begin{align}
    \alpha'\beta>\alpha'\beta|_{\epsilon=\epsilon_*}>\lim_{\alpha\to 2D}
    \alpha'\beta|_{\epsilon=\epsilon_*}=D
    \,,\quad 
\alpha'\beta<\alpha'\beta|_{\epsilon=0}=\alpha-D<D+1\,,
\end{align}
where to obtain the lower bound we noticed that $\alpha'\beta|_{\epsilon=\epsilon_*}$ is an increasing function of $\alpha$.

Now we evaluate the integrals ${\cal I}_1, {\cal I}_2$ of \eqref{eq:I1_I2}. The Lieb-Robinson bound holds for $|t|\le c R^{\alpha'}$, 
with $R=d(I,J)$.
So we define
\begin{align}
    t_* = a c R^{\alpha'\eta}\,, \quad 
    a\le 1\,,\quad
    0\le \eta\le 1\,,
\end{align}
so that the bound holds for $|t|\le t_*$.
Then, 
\begin{align}
    &{\cal I}_1
    \le
    \int_0^{t_*} \dd t\left(
    C_1 t^\beta R^{-\alpha'\beta}
    +
    C_2 t R^{D-\alpha}
    \right)
    =
    C_1 R^{-\alpha'\beta} 
    \frac{t_*^{\beta+1}}{\beta+1} 
    +
    C_2 R^{D-\alpha}
    \frac{t_*^2}{2} 
    \\
    &=
    C_1' R^{-\alpha'\beta + (\beta+1)\alpha'\eta}
    +
    C_2'R^{D-\alpha+2\alpha'\eta}
    \,,\quad 
    C_1'=C_1 \frac{(ac)^{\beta+1}}{\beta+1}
    \,,\quad 
    C_2'=C_2 \frac{(ac)^{2}}{2}\,.
\end{align}
Now from the bounds on $\alpha'\beta, \beta$ discussed above, we know that
\begin{align}
    -\alpha'\beta > D-\alpha
    \,,\quad
    \beta+1> D+2-\frac{1}{2(D+2)}>2\,,
\end{align}
so that the first power law in ${\cal I}_1$ dominates the sum:
\begin{align}
    {\cal I}_1
    \le
    C
    R^{-\nu}
    \,,\quad
\nu = \alpha'(\beta - (\beta+1)\eta)\,,\quad 
C=C_1'+C_2'\,.
\end{align}
Next we want so show that we can choose $\eta$ such that $\nu > D$ for all $\alpha\in (2D,2D+1)$.
This is required for the sum in $|f(O_I,x)-f(O_I,x')|$ to converge.
We set $\epsilon=\epsilon_*/2$ to simplify the discussion. We now set
\begin{align}
    \eta = \frac{1}{2}\frac{\alpha-2D}{2\alpha-2D-1}
    \in \left(0, \frac{1}{2}\frac{1}{2D+1}\right)\,.
\end{align}
The range of $\eta$ can be computed by noting that it is an increasing function of $\alpha$ and $\lim\limits_{\alpha\to 2D}\eta=0$,
$\lim\limits_{\alpha\to 2D+1}\eta=1/(2(2D+1))$.
Now we claim that with this choice we have $\nu > D$.
Indeed after some algebra we have that $\nu-D$ is the ratio of two polynomials of $x=\alpha-2D$ with positive coefficients for $D=1,2,\dots$,
\begin{align}
    \label{eq:nu-D}
    \nu=
    D+
    \frac{8 (D-1) D^2 x+
    4D (5 D-4) x^2+
    \left(24 D^2-8\right) x^3
    +(42 D-14) x^4+
    (16 D+9) x^5+
    10 x^6}{16 (2 D-1+2 x)
   \left(D+x^2+x\right)^2}
\end{align}
showing that $\nu-D>0$.
$\mu=\alpha'\eta$ is also an increasing function of $\alpha$ with range
\begin{align}
    0<\mu<\left(1-\frac{1}{2(D+2)} \right)\frac{1}{2(2D+1)} \,.
\end{align}

Thus, from Lemma \ref{lemma:2.6} we have
\begin{align}
    &\|\nabla_J
    f(O_I,x)
    \|
    \le
    \|\nabla_Jh_J\|\,\|O_I\|
    \left(
    c_1
    R^{-\nu}
    +
    4G(bR^\mu)
    \right)
    \,.
\end{align}
Note the difference with respect to the case of exponential decay or short range Hamiltonians: the exponential in their first term has been replaced by a power law with exponent $\nu$, while in the second term the distance is shrunk with a power $\mu$.
We then have:
\begin{align}
    &\frac{|f(O_I, x')-f(O_I, x)|}{\|O_I\|}
    \le S_1 + S_2\,,\\
    &S_1=
    \sum_{J\in {\cal P}_k(\Lambda)}
    \id\left(R\le \delta
    \wedge
    \diam(J)> \delta
    \right)
    \|\nabla_Jh_J\|
    \left(
    c_1
    R^{-\nu}
    +
    4G(bR^\mu)
    \right)\,,
    \\
    &S_2=
    \sum_{J\in {\cal P}_k(\Lambda)}
    \id\left(
    R>\delta\right)
    \|\nabla_Jh_J\|
    \left(
    c_1
    R^{-\nu}
    +
    c_2 
    (bR^\mu)^{10}u_{2/7}(bR^\mu)
    \right)\,.
\end{align}
Here the constants are $c_1=q_* C$ and $c_2=4/\gamma\times 130 \e^2 q_*$. 
We can rewrite these, since $\|\nabla_Jh_J\|\le \diam(J)^{-\gamma(|J|)}$ and denoting $R_* = (b^{-1}\xi_*)^{1/\mu}$, as follows:
\begin{align}
&S_1=
    \sum_{r>\delta}
    \sum_{\ell=1}^k
    r^{-\gamma(\ell)}
    \left(
    c_1
    \sum_{R\le \delta}
    R^{-\nu}
    +
    \sum_{R\le R_*}
    2\frac{Kq_*}{\gamma}
    +
    c_2
    \sum_{R_*<R\le \delta}
    (b R^\mu)^{10}u_{2/7}(b R^\mu)
    \right)
    M_\ell(I, r, R)
    \\
&S_2=
\sum_{r > 0}
\sum_{R>\delta}
    \sum_{\ell=1}^k
    r^{-\gamma(\ell)}
    M_\ell(I, r, R)
    \left(
    c_1
    R^{-\nu}
    +
    c_2 
    (bR^\mu)^{10}u_{2/7}(bR^\mu)
    \right)\,,
    \\
&M_\ell(I, r, R)=
    \sum_{J, |J|=\ell}
    \delta(r, \diam(J))
    \delta(R, d_{IJ})
    \,.
\end{align}
Note that the sum over $r$ starts from $1$ since the $\diam(I)>0$.
Now use Lemma \ref{lemma:Mell}, which proves that for $\ell>1$,
\begin{align}
    M_\ell(I, r, R)
    \le 
    C_{\ell} r^{(\ell-1)D-1} R^{D-1}\,,\quad 
    M_1(I,r,R)\le \delta(r,1) C R^{D-1}\,,
\end{align}
to bound $S_1$. The sum over $r$ is, using that $\delta>1$:
\begin{align}
    &
    \sum_{r>\delta}
    (
    C \delta(r,1) 
    +
    \sum_{\ell=2}^k
    r^{-\gamma(\ell)}
    C_{\ell} r^{(\ell-1)D-1}
    )
    =
    \sum_{\ell=2}^k
    C_{\ell}
    \sum_{r>\delta}
    r^{(\ell-1)D-1-(\ell-2)D-\alpha}\\
    &=
    \sum_{\ell=2}^k
    C_{\ell}
    \sum_{r>\delta}
    r^{D-1-\alpha}
    \le C_{\text{sum}}
    \delta^{D-\alpha}\,,
\end{align}
where in the bound we used
that $D-\alpha-1<-1$  
so that
\begin{align}
    \sum_{r>\delta}
    r^{D-\alpha-1}
    &\le 
    \int_{\delta}^\infty x^{D-\alpha-1}\dd x 
    = \frac{1}{D-\alpha} \delta^{D-\alpha}\,.
\end{align}
Then we perform the sums over $R$ in $S_1$:
\begin{align}
    &\left(
    c_1
    \sum_{R\le \delta}
    R^{-\nu}
    +
    \sum_{R\le R_*}
    2\frac{Kq_*}{\gamma}
    +
    c_2
    \sum_{R_*<R\le \delta}
    (b R^\mu)^{10}u_{2/7}(b R^\mu)
    \right)
    R^{D-1}\\
    &\le 
    c_1
    \delta^{-\nu+D}
    +
    c_1'
    +
    c_2
    \sum_{1\le R\le \delta}
    (b R^\mu)^{10}u_{2/7}(b R^\mu)
    \,.
\end{align}
Then we use the assumption that $b\delta^{\mu} > \eta_*$, so that the second term in $g(R)$ attains its maximum at a value $\tilde{R} < \delta$:
\begin{align}
    c_2 \sum_{1\le R\le \delta} (bR^\mu)^{D-1+10\mu}u_{2/7}(bR^\mu)
    \le 
    c_2 \tilde{C} \delta\,,\quad 
    \tilde{C} = 
    (b\tilde{R}^\mu)^{D-1+10\mu}u_{2/7}(b\tilde{R}^\mu)\,.
\end{align}
Putting things together we can bound $S_1$ as follows, since $\nu>D$,
\begin{align}
    S_1
    \le 
    C_{\text{sum}}
    \delta^{D-\alpha}
    (c_1
    \delta^{-\nu+D}
    +
    c_1'
    +
    c_2 \tilde{C} \delta)
    \le 
    C'
    \delta^{D-\alpha+1}
    \,,\quad 
    C'=\max(c_1,c_1',c_2\tilde{C})\,.
\end{align}
Now we turn to $S_2$.
Defining
\begin{align}
    g(R)
    =
    R^{D-1}
    \left(
    c_1
    R^{-\nu}
    +
    c_2 
    (bR^\mu)^{10}u_{2/7}(bR^\mu)
    \right)
    \,,
\end{align}
we get
\begin{align}
S_2\le 
\sum_{R>\delta}
g(R)
\sum_{r > 0}
\left(C \delta(r,1) + 
    \sum_{\ell=2}^k
    r^{-(\ell-2)D-\alpha}
    C_\ell  r^{(\ell-1)D-1}    
    \right)=
\sum_{R>\delta}
g(R)
\left(C 
+ 
C_{\text{sum}}
\sum_{r > 0}
    r^{D-\alpha-1}
    \right)
    \,.
\end{align}
Using
\begin{align}
    \sum_{r > 0}
    r^{D-\alpha-1}&=1+
    \sum_{r > 1}
    r^{D-\alpha-1}
    \le 
    1+
    \int_{1}^\infty x^{D-\alpha-1}\dd x 
    =1+\frac{1}{\alpha-D}\,,
\end{align}
we have
\begin{align}
    S_2 \le C'' 
    \sum_{R>\delta}
g(R)\,.
\end{align}
To estimate $S_2$ we use Lemma \ref{lemma:extension_2.5}.
We assumed that $b\delta^{\mu} > \eta_*$, with $\eta_*$ is such that the function $t^{9+\frac{D}{\mu}}u_{2/7}(t)$ is monotonically decreasing for $t>\eta_*$ and the hypotheses of Lemma \ref{lemma:extension_2.5} are satisfied.
Therefore, we first bound the sum with an integral -- recall also that $\nu>D$ so that the first integral converges:
\begin{align}
    &\sum_{R > \delta}
    g(R)
    \le
    \int_{\delta}^\infty 
    \dd R
    \left(
    c_1
    R^{-\nu+D-1}
    +
    c_2 
    R^{D-1}
    (bR^\mu)^{10}u_{2/7}(bR^\mu)
    \right)\\
    &=
    \frac{c_1}{\nu-D}
    \delta^{-\nu+D}
    +
    c_2 b^{-\frac{D-1}{\mu}}
    \int_{\delta}^\infty 
    \dd R
    (bR^\mu)^{10+\frac{D-1}{\mu}}u_{2/7}(bR^\mu)
\,.
\end{align}
Then we change variables to $\eta=bR^\mu$, so that $\dd R=
(\mu b^{1/\mu})^{-1} \eta^{1/\mu-1} \dd \eta$ and
\begin{align}
    \int_{\delta}^\infty 
    \dd R
    (bR^\mu)^{10+\frac{D-1}{\mu}}u_{2/7}(bR^\mu)
    =
(\mu b^{1/\mu})^{-1}
    \int_{b\delta^\mu}^\infty 
\dd \eta\,
    \eta^{9+\frac{D}{\mu}}u_{2/7}(\eta)\,.
\end{align}
Finally we apply Lemma \ref{lemma:extension_2.5} to get
\begin{align}
    S_2
    \le 
    \frac{c_1 C''}{\nu-D}
    \delta^{-\nu+D}
    +
    C_2'
    (b\delta^\mu)^{21 + 2\frac{D}{\mu}}
    u_{2/7}(b\delta^\mu)
    \,,\quad 
    C_2'=
    C'' (\mu b^{1/\mu})^{-1} c_2 b^{-\frac{D-1}{\mu}}\,.
\end{align}
Putting things together, we get
\begin{align}
\frac{|f(O_I, x')-f(O_I, x)|}{\|O_I\|}
\le    
C'
\delta^{-\alpha+D+1}
+
\frac{c_1C''}{\nu-D}
    \delta^{-\nu+D}
+
C_2'
(b\delta^\mu)^{21 + 2\frac{D}{\mu}}
    u_{2/7}(b\delta^\mu)\,.
\end{align}
Finally we claim that
\begin{align}
-\alpha+1\le 
    -\nu \,.
\end{align}
Indeed, after some algebra we have that, if $x=\alpha-2D$:
\begin{align}
    &(\nu-\alpha+1)16 (2 D-1+2 x) \left(D+x^2+x\right)^2
    =
    -16 (D-1) D^2 (2 D-1)-(48 D+7) x^5\\
    &\quad -(2 D (16 D+51)-34) x^4-8(D (21 D-8)-1) x^3-4 (D+2) (D (16 D-13)+2) x^2\\
    &\quad -8 D (3 D-1) (5
   D-4) x-22 x^6\,.
\end{align}
This shows that $\nu-\alpha+1\le 0$ since it is the ratio of a polynomial of $x$ with negative coefficients for all $D=1,2,\dots$ and a polynomial with positive coefficients.
Then, since $\delta>1$ we have that $\delta^{-\alpha+D+1} \le \delta^{-\nu+D}$, and we get the result of the Lemma:
\begin{align}
\frac{|f(O_I, x')-f(O_I, x)|}{\|O_I\|}
\le    
C_1'
    \delta^{-\nu+D}
+
C_2'
(b\delta^\mu)^{21 + 2\frac{D}{\mu}}
    u_{2/7}(b\delta^\mu)\,,\quad 
    C_1'=
\max(\tfrac{c_1C''}{\nu-D}, C')\,.
\end{align}
\end{proof}

We see that compared to the case of exponentially decaying interactions of Proposition \ref{prop:approx_exp_decay}, in the second term we have replaced $\delta$ with $\delta^{\mu}$ and a different power that depends on $\mu$. Also, the first term which was exponentially decaying is now a power law.
We have the following Corollary whose proof again follows from Lemma \ref{lemma:7}:

\begin{corollary}\label{cor:delta_eps_alpha_gt_2D}
    Under the same assumptions of Proposition \ref{prop:approx_alpha_gt_2d}, we have 
    \begin{align}
    \label{eq:delta_eps_alpha_gt_2d}
    \delta &\ge \tilde{c}\max(
    \epsilon^{-1/(\nu-D)},
    \log^{2/\mu}(3C_2/\epsilon))\,,\\
    \tilde{c}&=\max((3C_1)^{1/(\nu-D)},(c/b)^{1/\mu})\,.
\end{align}
\end{corollary}

\subsubsection{Failure of bounding the local approximation for power law decay with $\alpha\le 2D$}
\label{sec:power law alpha le 2d}

We have seen in Corollary \ref{cor:delta_eps_alpha_gt_2D} that $\delta\to \infty$ as $\alpha\to 2D$. In this section we study this phenomenon by using a Lieb-Robinson bound that applies for all $\alpha\in (D,2D)$ and implies a logarithmic light cone for operator spreading.
We start by recalling the following Lieb-Robinson bound.
\begin{lemma}[\cite{Hastings_2006}]\label{lemma:lr_alpha_lt_2d}
    Assume that
\begin{align}
    H=\sum_{I\in {\cal P}_k(\Lambda)} h_I\,,\quad 
    \sum_{I\ni i,j}\|h_I\|\le g d(i,j)^{-\alpha}\,,
\end{align}
for a positive constant $g$.  For any $\alpha\in (D,2D)$,  we have
    \begin{align}
        \|[A_I,B_J]\|
        \le C \|A_I\|\,\|B_J\|\,
        |I|\,|J|\, \frac{\e^{v|t|}-1}{(1+d(I,J))^\alpha}
    \end{align}
    with $C,v$ depending only on the Hamiltonian and the metric of the lattice.
\end{lemma}
We have the following result:
\begin{prop}
    Assume the hypotheses of Lemma \ref{lemma:lr_alpha_lt_2d} are satisfied.
    Then the bound of the quantity $|f(O_I,x)-f(O_I,x')|$ introduced in Section \ref{sec:Dependency of observables on Hamiltonian parameters} becomes vacuous in this regime.
\end{prop}
\begin{proof}
We use the bound in Lemma \ref{lemma:lr_alpha_lt_2d} to compute ${\cal I}_1$ appearing in \eqref{eq:I1_I2}:
\begin{align}
    &\frac{{\cal I}_1}{
    |I\|J|\, }
    \le
    C \frac{1}{(1+d(I,J))^\alpha}
    \int_0^{t_*}
    (\e^{vt}-1)
    \le
    \frac{C}{v} \frac{\e^{vt_*}}{(1+d(I,J))^\alpha}
    \,.
\end{align}
We now take 
\begin{align}
    t_*(R)
    =
    \frac{\beta}{v}
    \log(1+R)\,,
\end{align}
where we will fix $\beta$ later. If 
\begin{align}
    \delta >\gamma^{-1}
    t_*^{-1}(\xi_*)
    =
    \gamma^{-1}
    (\e^{\frac{v}{\beta}\xi_*}-1)
\end{align}
with $\xi_*$ as in Lemma \ref{lemma:2.6},
we have that, from \eqref{eq:I_2_gamma_t_star} with $R=d(I,J)$ and $b=\beta \gamma v^{-1}$:
\begin{align}
    \|\nabla_J
    f(O_I,x)
    \|
    \le 
    \|
    \nabla_{J}h_J(x_J)
    \|
    \,
    \|
    O_I 
    \|
    (
    c_1
    (1+R)^{\beta-\alpha}
    +
    4 G(b\log(1+R))
    )\,,
\end{align}
with $G$ as in Lemma \ref{lemma:2.6}.
As in the proofs of propositions \ref{prop:approx_exp_decay} and \ref{prop:approx_alpha_gt_2d} we can write
\begin{align}
    &\frac{|f(O_I, x')-f(O_I, x)|}{\|O_I\|}
    \le S_1 + S_2\,,\\
    &S_1=
    \sum_{J\in {\cal P}_k(\Lambda)}
    \id\left(r\le \delta
    \wedge
    \diam(J)> \delta
    \right)
    \|\nabla_Jh_J\|
    \left(
    c_1
    R^{\beta-\alpha}
    +
    4 G(b\log(1+R))
    \right)
    \\
    &S_2=
    \sum_{J\in {\cal P}_k(\Lambda)}
    \id\left(
    R>\delta\right)
    \|\nabla_Jh_J\|
    \left(
    c_1
    R^{\beta-\alpha}
    +
    c_2
    (b\log(1+R))^{10}u_{2/7}(b\log(1+R))
    \right)\,.
\end{align}
We will show that $S_2=\infty$.
If $\|\nabla_Jh_J\|=h(\diam(J))$:
\begin{align}
&S_2=
\sum_{r > 0}
    h(r)
\sum_{R>\delta}
    \left(
    c_1
    R^{\beta-\alpha}
    +
    c_2
    (b\log(1+R))^{10}u_{2/7}(b\log(1+R))
    \right)
    M(I, r, R)\,,
    \\
&M(I, r, R)=
    \sum_{J\in {\cal P}_k(\Lambda)}
    \delta(r, \diam(J))
    \delta(R, d_{IJ})
    \,.
\end{align}
To simplify the discussion and notation we will consider below the case of Hamiltonians with terms $h_J$ with $|J|=2$ only.
Then Lemma \ref{lemma:Mell} gives:
\begin{align}
    M(I, r, R)
    \le
    C
    r^{D-1} R^{D-1}
    \,.
\end{align}
Since for $|J|=2$ we have
\begin{align}
    \|\nabla_Jh_J\|\le \diam(J)^{-\alpha}\,,
\end{align}
defining
\begin{align}
    g(R)
    =
    R^{D-1}
    \left(
    c_1
    R^{\beta-\alpha}
    +
    c_2
    (b\log(1+R))^{10}u_{2/7}(b\log(1+R))
    \right)
    \,,
\end{align}
we get the following result by replacing the sum over $r$ with an integral, since the summand is a decreasing function and $\alpha>D$ so $-\alpha+D-1< -1$ and the integral converges:
\begin{align}
S_2=
\sum_{r > 0}
    r^{-\alpha+D-1}
\sum_{R > \delta}
    g(R)
    \le 
    C
\sum_{R > \delta}
    g(R)
    \,.
\end{align}
Now we consider the  sum:
\begin{align}    
    \sum_{R > \delta}
    \left(
    c_1
    R^{D-1+\beta-\alpha}
    +
    c_2
    R^{D-1}
    (b\log(1+R))^{10}u_{2/7}(b\log(1+R))
    \right)    
    \,.
\end{align}
We need $D-1+\beta-\alpha<-1$ for the first term to converge, so $0<\beta<\alpha-D$ which can be satisfied since $\alpha-D>0$.
For the second term, we have
\begin{align}
    u_{2/7}(b\log(1+R))
    =
    \exp\left(-\frac{2}{7}\frac{b\log(1+R)}{\log^2(b\log(1+R))}\right)
    =
    (1+R)^{-a(R)}\,,
    \quad 
    a(R)
    =
    \frac{2}{7}b \log^{-2}(b\log(1+R))
\end{align}
For the sum to converge, we want $a(R) > D$ for large $R$, so, defining $z=\exp(\sqrt{\frac{2b}{7D}})$,
\begin{align}
  \log^{2}(b\log(1+R)) < \frac{2b}{7D}
  \Rightarrow
  b\log(1+R) < z
  \Rightarrow
  R < \e^{z/b}-1\,.
\end{align}
This means that for $R>\e^{z/b}-1$ the summand will decay as a non-summable power of $R$, leading to a divergent bound: $S_2=\infty$.
\end{proof}
This shows that for $\alpha\le 2D$ the observable $f(O_I,x)$ does not only depend on $x_J$ with $J\in S_{I,\delta}$, but also on parameters outside this region.

\subsubsection{Bound on the gradient}

In Section \ref{sec:Dependency of observables on Hamiltonian parameters} we bounded $\|\nabla_J \Tr(O_I,\rho(x))\|$ and its sum over $S_{\delta,I}^c$. It will be useful for deriving the machine learning model to have a bound on the norm of the gradient itself.
We now generalize \cite[Lemma 4]{Huang_2022} to exponentially decaying and power law interactions with $\alpha>2D$:
\begin{lemma}\label{lemma:bound_gradient}
    The gradient of $\Tr(O\rho(x))$ for $O=\sum_I O_I$ is bounded by
    \begin{align}
        \|\nabla \Tr(O\rho(x))\|
        \le 
        C'
        \sum_{I}\|O_I\| 
    \end{align}
    with $C'>0$ a constant for both exponentially decaying interactions and power law interactions with $\alpha>2D$.
\end{lemma}
\begin{proof}
    The computation is very similar to the one done above for $\|\nabla f(O_I,\rho(x))\|$.
    As in \cite[Lemma 4]{Huang_2022}, we use that
    \begin{align}
        \label{eq:u_nabla}
        \|\nabla \Tr(O\rho(x))\|
        =
        \|u\cdot \nabla \Tr(O\rho(x))\|
        \,,\quad u = \frac{\nabla \Tr(O\rho(x))}{\|\nabla \Tr(O\rho(x))\|}\,,
    \end{align}
    and bound the quantity $\|\partial_u \Tr(O\rho(x))\|$ for a generic unit vector $u$. $u$ has the same size as $x$ and we denote its $J$-th component as $u_J$. Since $\partial_u = \sum_{J}u_{J}\cdot  \nabla_{J}$, we have
    \begin{align}
        |\partial_u f(O_I,x)|
        \le 
        \sum_J |u_J| |\nabla_J f(O_I,x)|
        \le 
        \|O_I\|
        \sum_J \|\nabla_J h_J\|
        ({\cal I}_1+{\cal I}_2)\,,
    \end{align}
where in the second inequality we used $|u_J|\le 1$ (since $u$ is a unit vector) and \eqref{eq:I1_I2}. We now discuss the case of exponentially decaying interactions. From the proof of Proposition \ref{prop:approx_exp_decay} we know that, with $G$ as in Lemma \ref{lemma:2.6},
\begin{align}
\sum_J \|\nabla_J h_J\|
        ({\cal I}_1+{\cal I}_2)
    \le 
    C_{\text{sum}} 
    \sum_{R\ge 0}\sum_{r \ge 1}
    \e^{-\nu r}
    r^{(k-1)D-1}R^{D-1}
    \left(
    c_1
    \e^{-\mu' R}
    +
    4G(bR)
    \right)
    \,.
\end{align}
The sum over $r$ as well as the sum over $R$ of the term multiplying $c_1$ produce a constant. We then claim that the remaining term, which is proportional to
\begin{align}
    \sum_{R\ge 0}
    G(bR)
    \,,
\end{align}
is also constant. This is because the sum is convergent, as used in the computation of $S_2$ in Proposition \ref{prop:approx_exp_decay}, and so it gives a finite constant.

Similarly, we can deal with power law interactions with $\alpha>2D$. We know from Proposition \ref{prop:approx_alpha_gt_2d} that 
\begin{align}
    &\sum_J \|\nabla_J h_J\|
    ({\cal I}_1+{\cal I}_2)
    \le 
    g
    \sum_{R\ge 0}
    R^{D-1}
    \left(
    c_1
    R^{-\nu}
    +
    4G(bR^\mu)
    \right)
    \,,\\
    g &= 
    C + C\sum_{\ell=2}^k
    \sum_{r>0}
    r^{-\gamma(\ell)+(\ell-1)D-1} 
    =
    C + C(k-1)
    \sum_{r>0}
    r^{-\alpha + D-1} 
    < \infty\,,
\end{align}
since $-\alpha+D-1<-1$. The term proportional to $c_1$ is a constant since $-\nu+D-1<-1$ because $\nu>D$. The second term proportional to $c_2$ is also a constant because the sum is convergent, as used in the computation of $S_2$ in Proposition \ref{prop:approx_exp_decay}.
\end{proof}

\subsection{Approximation by discretisation}
\label{sec:Approximation by discretisation}

It is straightforward to adapt the discretisation of the function $f(O_I,\chi_{S_{I,\delta}}(x))$ discussed in \cite{lewis2023improved} to our case, as we detail next.
For given $\epsilon>0$ we define the fraction
\begin{align}
    \label{eq:delta_2}
    \delta_2 = 
    \Big\lceil \sqrt{C' |S_{I,\delta}|}\epsilon^{-1}
    \Big\rceil^{-1}
\end{align}
with $C'$ as in Lemma \ref{lemma:bound_gradient} and $S_{I,\delta}$ as in \eqref{eq:S_I_delta}. Next we define the discretised space:
\begin{align}
    \label{eq:XIdelta}
    X_{I,\delta}
    &=
    \left\{ x\in [-1,1]^m \,\left|\,
    x_J = 0 \text{ if }J\not\in S_{I,\delta}\,,\quad
    x_J\in W_{\delta_2}^{|J|} \text{ if }J\in S_{I,\delta}
    \right.\right\}\,,\\
    W_{\delta_2}
    &=\{0, \pm \delta_2, \pm 2 \delta_2,  \dots, \pm 1\}
    \,,
\end{align}
and the space of points close to an $x\in X_{I,\delta}$:
\begin{align}
    \label{eq:TxIdelta}
    T_{x,I}
    =
    \left\{
    x'\in [-1,1]^m \,\left|\,
    -\frac{\delta_2}{2}
    < x_{J} - x'_{J}<
    +\frac{\delta_2}{2}
    \,,\quad J\in S_{I,\delta}
    \right.\right\}\,,
\end{align}
where for a vector $x$, $x>0$ means that every component is greater than zero.
We now adapt \cite[Lemma 5]{lewis2023improved} to our notation. The proof relies on the bound on the gradient of Lemma \ref{lemma:bound_gradient}.
\begin{lemma}[Lemma 5 of \cite{lewis2023improved}]
    Given $\epsilon>0$ and defining the discretised version of $f$:
    \begin{align}
    g(O_I,x)
    =
    \sum_{x'\in X_{I,\delta}} f(O_I, x')\id(x\in T_{x',I})
    \,,
    \end{align}
    if we choose $\delta$ in terms of $\epsilon$ as in \eqref{eq:delta_eps_exp} or
    \eqref{eq:delta_eps_alpha_gt_2d}, with   
    our choice of $\delta_2$ in \eqref{eq:delta_2}, we have that
    \begin{align}
        \label{eq:g_minus_f}
    |f(O_I,x)-g(O_I,x)|<\epsilon \|O_I\|\,.
    \end{align}
\end{lemma}
We see that $g$ has the form of a linear function of the parameters
\begin{align}
    \label{eq:target_w}
    (w')_{I,x'}=f(O_I, x')\,.
\end{align}
Then for any observable 
\begin{align}
    \label{eq:O}
    O = \sum_{I\in {\cal S}} O_I\,,\quad {\cal S}\subseteq {\cal P}_p(\Lambda)\,,
\end{align}
we get
\begin{align}
    \label{eq:diff_f_target}
    |f(O, x) - g(O,x)|\le 
    \sum_{I\in {\cal S}}
    |f(O_I, x) - g(O_I,x)|\le 
    \epsilon
    r(O)\,,
\end{align}
with $r(O)$ as in \eqref{eq:rO} and
\begin{align}
    g(O,x)
    =
    w'\cdot \phi(x)
    =
    \sum_{I\in {\cal S}}\sum_{x'\in X_{I,\delta}} (w')_{I,x'}\phi(x)_{I,x'}
    \,,\quad \phi(x)_{I,x'} = \id(x\in T_{x',I})\,.
\end{align}
While the true weights are unknown, we show below they can be learned efficiently. Before moving on, we note the following bounds that will be useful in the derivation of the sample complexity of the machine learning model:
\begin{lemma}\label{lemma:bounds_m_phi_B}
Denote by $m_\phi$ the size of the weight vector $w'$ of \eqref{eq:target_w} associated to the observable \eqref{eq:O}. 
With $\nu$ as in Proposition \ref{prop:approx_alpha_gt_2d} define
\begin{align}
    \label{eq:n_eps}
    \omega = \frac{kD}{\nu-D}\,,\quad
    {\cal N}(\epsilon)
    =
    \begin{cases}
    2^{\mathcal{O}(\operatorname{polylog}(1/\epsilon))} & \textup{ exp decay},\\
    2^{\mathcal{O}(\epsilon^{-\omega} \log(1/\epsilon))} & \textup{ power law }\alpha>2D\,.
\end{cases}
\end{align}
We have the following bounds:
\begin{align}
m_\phi \le 
|{\cal S}| {\cal N}(\epsilon)\,,
\quad 
\|w'\|_1 \le 
r(O)
{\cal N}(\epsilon)
\equiv B
\,.
\end{align}
\end{lemma}
\begin{proof}
We note that the size of weights $w$ and the feature map $\phi$ is
\begin{align}
    m_\phi
    =
    \sum_{I\in {\cal S}} |X_{I,\delta}|\,.
\end{align}
We can bound $|X_{I,\delta}|$ as
\begin{align}
    |X_{I,\delta}|\le 
    (2\delta_2^{-1} + 1)^{|S_{I,\delta}|}
    \le 
    \left( 2 \frac{\sqrt{C' |S_{I,\delta}|}}{\epsilon}
     + 3\right)^{|S_{I,\delta}|}
\end{align}
and 
\begin{align}
    |S_{I,\delta}|
    =
    \sum_{J\in {\cal P}_k(\Lambda)}
    \id\left(d(I,J)\le \delta
    \wedge
    \diam(J)\le  \delta
    \right)
    =
    \sum_{r,R\le \delta}    
    M(I,r,R)
    \,,
\end{align}
where $M$ is defined in Corollary \ref{corollary:M}. Then
\begin{align}
    |S_{I,\delta}|\le C
    \sum_{r,R\le \delta}
    r^{(k-1)D-1} R^{D-1}
    \le 
    C \delta^{kD}\,.
\end{align}
Note that compared to \cite{lewis2023improved}, we have $\delta^{kD}$ instead of $\delta^{D}$ because of the requirement that $\diam(J)\le \delta$ which is non-trivial for general $k$-body interactions.
Then
\begin{align}
    |X_{I,\delta}|\le 
    \left(2\sqrt{C'}  \delta^{kD/2}\epsilon^{-1}
     + 3\right)^{C \delta^{kD}}
     =
     2^{C\delta^{kD} \log_2(2\sqrt{C'} \delta^{kD/2}\epsilon^{-1}
     + 3)}\,.
\end{align}
To proceed, we assume exponentially decaying interactions first.
Then we know from Proposition \ref{prop:approx_exp_decay}
that $\delta = C_{\text{max}}\log^2(1/\epsilon)$ for small $\epsilon$, which implies
\begin{align}
    |X_{I,\delta}|\le 2^{{\cal O}(\text{polylog}(1/\epsilon))}\,.
\end{align}
So the number of features is
\begin{align}
    m_\phi
    \le 
    \sum_{I\in {\cal S}} 2^{{\cal O}(\text{polylog}(1/\epsilon))}
    = |{\cal S}| 2^{{\cal O}(\text{polylog}(1/\epsilon))}
    \,.
\end{align}
We can also bound the $\ell_1$ norm of the true weights:
\begin{align}
    \|w'\|_1
    =
    \sum_{I\in {\cal S}}\sum_{x'\in X_{I,\delta}}|f(O_I,x')|
    \le 
    \sum_{I\in {\cal S}}\sum_{x'\in X_{I,\delta}}\|O_I\|
    \le 
    2^{{\cal O}(\text{polylog}(1/\epsilon))}
    r(O)
    \,.
\end{align}
Then we consider power law decay with $\alpha>2D$. 
From Proposition \ref{prop:approx_alpha_gt_2d} we have $\delta = C_{\text{max}}\epsilon^{-1/(\nu-D)}$ for a fixed $\alpha$ and small $\epsilon$, which implies
\begin{align}
    |X_{I,\delta}|\le 
    2^{\mathcal{O}(\epsilon^{-\omega} \log(\epsilon^{-1-\omega/2}))}
    =
    2^{\mathcal{O}(\epsilon^{-\omega} \log(1/\epsilon))}
    \,,
    \quad 
    \omega = \frac{kD}{\nu-D}\,.
\end{align}
So the number of features is
\begin{align}
    m_\phi
    \le 
    \sum_{I\in {\cal S}} 2^{\mathcal{O}(\epsilon^{-\omega} \log(1/\epsilon))}
    = |{\cal S}| 2^{\mathcal{O}(\epsilon^{-\omega} \log(1/\epsilon))}
    \,.
\end{align}
and the $\ell_1$ norm of the true weights:
\begin{align}
    \|w'\|_1
    =
    \sum_{I\in {\cal S}}\sum_{x'\in X_{I,\delta}}|f(O_I,x')|
    \le 
    \sum_{I\in {\cal S}}
    \|O_I\|
    \sum_{x'\in X_{I,\delta}}1
    \le 
    r(O)
    2^{\mathcal{O}(\epsilon^{-\omega} \log(1/\epsilon))}
    \,.
\end{align}
\end{proof}

Note that if we used a polynomial approximation of $f(O_I,x)$ rather than discretisation, we would have a similar growth of the number of parameters with the error $\epsilon$, as can be derived from classical approximation theory, see e.g.~\cite[Theorem 4]{Newman1964}.

\subsection{Limit theorem for global observables}\label{sec:CLTGlobalObservables}

We now recall some facts about random fields \cite{gaetan2009spatial}.
A random field is a collection of random variables $X=(X_i)_{i\in S}$, $S\subseteq \mathbb{Z}^D$.
We call a random field $R$-dependent if for any pair $i,j\in S$ such that $d(i,j)>R$, $X_i$ and $X_j$ are independent: $\mathbb{E}(X_iX_j)=\mathbb{E}(X_i)\mathbb{E}(X_j)$.
We then have the following central limit theorem.
\begin{lemma}[Prop. B1 in \cite{gaetan2009spatial} for $R$-dependent random fields]
    Suppose $D_n$ is a strictly increasing sequence of finite subsets of $\mathbb{Z}^d$, 
    $X$ a random field with zero mean that is $R$-dependent, and define $S_n = \sum_{i\in D_n}X_i$, $\sigma_n^2=\text{Var}(S_n)$.
    Then $\sigma_n^{-1}S_n\to {\cal N}(0,1)$, the standard normal distribution, 
    provided that
    \begin{enumerate}
        \item There exists 
        $\delta>0$ such that 
        $\sup_i \mathbb{E}(|X_i|^{2+\delta})<\infty$  
    \item $\liminf_n |D_n|^{-1}\sigma_n^2>0$.
    \end{enumerate}    
\end{lemma}
We shall apply this result to 
\begin{align}
    X_i = \Tr(O_i \rho(\chi_{S_{i,\delta}}(x)))
    -
    \mu_i\,,\quad \mu_i=
    \mathbb{E}(\Tr(O_i \rho(\chi_{S_{i,\delta}}(x))))
    \,,
\end{align}
for $i\in\Lambda$.
We recall that the random fields $X_i$, whose randomness is derived from the random parameters $x\in [-1,1]^m$, 
are functions of $x_J$ with $d(J,i)\le \delta$ and $\diam(J)\le \delta$, so for any $j_a\in J$, $J\in S_{i,\delta}$:
\begin{align}
    d(j_a,i)\le d(j_b,i)+d(j_b,j_a)
    \le 2\delta\,,
\end{align}
where $j_b$ is the closest point to $i$.
Therefore if $d(i,j)>4\delta$, $X_i,X_j$ depend on different $x$ variables and thus are independent, if we assume that the $x_J$'s are independent. We thus have the following Proposition as a Corollary of the central limit theorem for random fields above.
\begin{prop}\label{prop:clt}
Assume the $x_J$'s are independent random variables. Then $$ X_i = \Tr(O_i \rho(\chi_{S_{i,\delta}}(x)))-
    \mu_i$$ are $4\delta$-dependent, and defining
\begin{align}
    S_n = \sum_{i=1}^n X_i
    =\sum_{i=1}^n \Tr(O_i \rho(\chi_{S_{i,\delta}}(x)))
    -\sum_{i=1}^n\mu_i
    \,,\quad 
    \sigma_n^2 = \text{Var}(S_n)\,,
\end{align}
we have 
$\sigma_n^{-1}S_n\to {\cal N}(0,1)$, 
    provided that
    \begin{enumerate}
        \item There exists 
        $\delta>0$ such that 
        $\sup_i \mathbb{E}(|X_i|^{2+\delta})<\infty$  
    \item $\liminf_n \frac{\sigma_n^2}{n}>0$.
    \end{enumerate}  
\end{prop}
Note that
\begin{align}
    |X_i|
    \le
    |\Tr(O_i \rho(\chi_{S_{i,\delta}}(x)))|+
    |\mu_i|
    \le 
    2\|O_i\|\,.
\end{align}
So the first condition is 
\begin{align}
    \mathbb{E}(|X_i|^{2+\delta})
    \le 
    (2\|O_i\|)^{2+\delta}
     \,,
\end{align}
and $\sup_i \mathbb{E}(|X_i|^{2+\delta})<\infty$ as long as $\sup_i \|O_i\|$ is finite, which we assume.
We can bound the variance of $S_n$ as
\begin{align}
    \sigma_n^2
    =
     \mathbb{E}\left( 
     \sum_{i=1}^n X_i
     \right)^2
     \le 
     \mathbb{E}\left( 
     \sum_{i=1}^n |X_i|
     \right)^2
     \le 
     4
     \left( 
     \sum_{i=1}^n \| O_i\|
     \right)^2\,.
\end{align}
The second condition means that  $\sigma_n^2 = \sigma^2 n+o(n)$. 
Even if we normalise $\|O_i\|$ so that  $\sum_{i=1}^n \| O_i\|={\cal O}(\sqrt{n})$, and we have a bound of $cn$, we cannot however conclude that the second condition is satisfied.
We shall discuss this further in Section \ref{sec:Experiments}, where we will check it numerically for specific examples.

We note that under the assumption of Proposition \ref{prop:clt} for the random variable $S_n/n$ we would get $\text{Var}(S_n/n)=\sigma_n^2/n^2\to 0$
so that $S_n/n$ tends to the average mean and fluctuations are suppressed for large $n$: 
\begin{align}
    \frac{1}{n}
    \sum_{i=1}^n \Tr(O_i \rho(\chi_{S_{\delta,i}}(x)))
    =
    \frac{1}{n}
    \sum_{i=1}^n\mu_i
    +
    \frac{\sigma}{\sqrt{n}}
    \xi
    +
    o(n^{-1/2})
    \,.
\end{align}
Here $\xi$ is a random variable that depends on $x$ and is distributed as a standard normal.
If $\|O_i\|={\cal O}(1)$ and want a bounded operator norm,
we shall consider $S_n/n$, since $\|S_n/n\|\le \sum_i \|O_i\|/n = {\cal O}(1)$.

Now we recall that if we choose $\delta$ as in \eqref{eq:delta_eps_exp} for exponential decay and \eqref{eq:delta_eps_alpha_gt_2d} for power law decay with $\alpha>2D$, we have
\begin{align}
    |
    \Tr(O_i \rho(x))
    -
    \Tr(O_i \rho(\chi_{S_{i,\delta}}(x)))
    |\le \epsilon \|O_i\|\,.
\end{align}
Thus, 
\begin{align}
    \left|
    \frac{1}{n}\sum_{i=1}^n
    \Tr( O_i \rho(x))
    -
    \frac{1}{n}
    \sum_i\mu_i
    -
    \frac{\sigma}{\sqrt{n}}
    \xi
    +
    o(n^{-1/2})
    \right|
    \le \epsilon 
    \frac{1}{n}
    \sum_{i=1}^n
    \|O_i\|
    \equiv \epsilon'
    \,.
\end{align}
So if $\|O_i\|={\cal O}(1)$, 
we can approximate $n^{-1}\sum_{i=1}^n
    \Tr( O_i \rho(x))$ by the $x$-independent number 
$n^{-1} \sum_i\mu_i$ with constant error $\epsilon'$, meaning that for large $n$ the observable concentrates around its mean.

So far we have considered observables of the type $O_I$ with $|I|=1$. We could not find a central limit theorem for the fields supported on sets rather than single sites in the literature, but we conjecture that the same behavior holds for the more general class of observables of the form $\sum_{I\in {\cal P}_p(\Lambda)} O_I$ with $p={\cal O}(1)$.

\section{Machine learning model for classical shadows}
\label{sec:Machine learning model and sample complexity bounds}

\subsection{Classical shadows and prediction of many observables}
\label{sec:Classical shadows and prediction of many observables}

This section extends \cite[Corollary 5]{lewis2023improved} to our setting, and shows how to predict many observables from classical shadows data.
We start by recalling that a classical shadow is an approximation of the density matrix $\rho$ obtained by repeated random measurements \cite{Huang_2020}. For each copy of $\rho$ we select  uniformly at random whether to measure $X,Y,Z$ for each qubit and store the associated measurement results as classical data $s_i^t\in \{ 0_X,1_X,0_Y,1_Y,0_Z,1_Z\}$ for the measurement outcome of qubit $i$ at time $t$. 
Here we denote by $0_A,1_A$ the possible states after measurement of the Pauli $A$.
After $T$ measurements, the classical shadow is 
\begin{align}
    \sigma_T(\rho)
    =
    \frac{1}{T}
    \sum_{t=1}^T
    \left(3\ket{s_1^t}\bra{s_1^t}-\id_2 \right)\otimes \cdots \otimes
    \left(3\ket{s_n^t}\bra{s_n^t}-\id_2 \right)\,.
\end{align}
Note that $\Tr\left(3\ket{s_i^t}\bra{s_i^t}-\id_2 \right)=1$, so the reduced density matrix for subsystem $I=\{i_1,\dots,i_\ell\}$ is
simply
\begin{align}
    \Tr_{I^c}(\sigma_T(\rho))
    =
    \frac{1}{T}
    \sum_{t=1}^T
    \left(3\ket{s_{i_1}^t}\bra{s_{i_1}^t}-\id_2 \right)\otimes \cdots \otimes
    \left(3\ket{s_{i_\ell}^t}\bra{s_{i_\ell}^t}-\id_2 \right)\,.
\end{align}
We are going to use the following result.
\begin{lemma}[Lemma 1 in \cite{Huang_2022}]
    Given $\epsilon,\gamma\in (0,1)$, we have with probability at least $1-\delta$,
    \begin{align}
        \|\Tr_{I^c}(\rho)-\Tr_{I^c}(\sigma_T(\rho))\|_1
        \le 
        \epsilon
    \end{align}
    for all $I\in {\cal P}_k(\Lambda)$, if
    \begin{align}
        T = \mathcal{O}(k 12^k \log(n/\gamma)/\epsilon^2)\,.
    \end{align}
\end{lemma}
\begin{proof}
We here sketch some steps of the proof that are going to be useful below.
The Lemma follows from the Bernstein matrix inequality, which gives \cite{Huang_2022}:
\begin{align}
    \text{Pr}(\|\Tr_{I^c}(\rho)-\Tr_{I^c}(\sigma_T(\rho))\|_1\ge \epsilon)
    \le 
    2^{k+1}\exp\left(-\frac{3T\epsilon^2}{8 \times 12^k}\right)
    \,.
\end{align}
Now we want to bound 
\begin{align}
    \text{Pr}(\|\Tr_{I^c}(\rho)-\Tr_{I^c}(\sigma_T(\rho))\|_1
        \le 
        \epsilon\,, \forall I
        \in {\cal P}_k(\Lambda)
    )
    =
    \text{Pr}(
    \bigcap_{I
        \in {\cal P}_k(\Lambda)}
    \|\Tr_{I^c}(\rho)-\Tr_{I^c}(\sigma_T(\rho))\|_1
        \le 
        \epsilon
    )\,,
\end{align}
and we are going to use that
\begin{align}
   \text{Pr}(A\cap B)
   =
   1-\text{Pr}(\overline{A\cap B})
   =
   1-\text{Pr}(\overline{A}\cup \overline{B})
   \,,
\end{align}
and then compute the r.h.s.~with the union bound
$\text{Pr}(\cup_{i}A_i)\le\sum_i \text{Pr}(A_i)$.
We have
\begin{align}
    \text{Pr}\left(\bigcup_{I\in {\cal P}_k(\Lambda)}
    \|\Tr_{I^c}(\rho)-\Tr_{I^c}(\sigma_T(\rho))\|_1\ge \epsilon\right)
    &\le
    \sum_{I\in {\cal P}_k(\Lambda)}
    \text{Pr}\left(
    \|\Tr_{I^c}(\rho)-\Tr_{I^c}(\sigma_T(\rho))\|_1\ge \epsilon\right)
    \\
    &\le 
    C_k
    n^k 2^{k+1}
    \exp\left(-\frac{3T\epsilon^2}{8 \times 12^k}\right)
    \,.
\end{align}
Setting this equal to $\gamma$, gives the value of $T$ in the Lemma.
\end{proof}

In our case, $k=\mathcal{O}(1)$ and so we can accurately predict the reduced density matrix on any subsystems of size $k$ with only $\mathcal{O}(\log(n/\gamma)/\epsilon^2)$ randomized measurements.

The following result shows that if we have access to classical shadow data we can predict all observables of the form    $\sum_{I\in {\cal P}_k}O_I$ 
with the same sample complexity of predicting a single one of Theorem \ref{thm:sample_complexity_single_obs}.
\begin{theorem}
    Suppose we have data $\{x^{(i)}, \sigma_T(\rho(x^{(i)}) \}_{i=1}^N$ with
\begin{align}
    T &= \mathcal{O}(\log(nN/(\gamma/2))/(\epsilon_2')^2)
    \,,\quad
    \epsilon_2'=\epsilon_2/C\,,\\
    N &=
    (\epsilon_3/C^2)^{-2}
    {\cal N}(\epsilon_1/C)
    \mathcal{O}(\log(n/\gamma)) 
    \,,
\end{align}
with $\epsilon_1,\epsilon_2,\epsilon_2,C>0$ and $\gamma\in(0,1)$.
    Then we can learn a predictor $\hat{\rho}(x)$ that achieves
    \begin{align}
        \mathbb{E}_{x\sim {\cal D}}|\Tr(O\rho(x))-\Tr(O\hat{\rho}(x))|^2
    \le 
    (\epsilon_1+\epsilon_2)^2+\epsilon_3\,,
    \end{align}
    with probability at least $1-\gamma$ for all observables $O$ such that 
    \begin{align}
    O = \sum_{P\in P_k} \alpha_P P
    \,,\quad 
    \sum_{P\in P_k}|\alpha_P|\le C\,,
\end{align}
where $P_k$ is the set of Pauli strings with weight at most $k$. 
\end{theorem}
\begin{proof}
We start to note that 
\begin{align}
    |P_k| = \sum_{\ell=0}^k
    \binom{n}{\ell} 3^\ell
    \le 
    (k+1)\binom{n}{k} 3^k 
    \le 
    (k+1)(3 \e n/ k)^k
    =
    \mathcal{O}(n^k)\,.
\end{align}
Then by the union bound
\begin{align}
    &\text{Pr}\left(
    \bigcup_{i=1}^N
    \bigcup_{I\in {\cal P}_k(\Lambda)}
    \|\Tr_{I^c}(\rho(x^{(i)}))-\Tr_{I^c}(\sigma_T(\rho(x^{(i)})))\|_1\ge \epsilon\right)\\
    &\le
    \sum_{i=1}^N
    \sum_{I\in {\cal P}_k(\Lambda)}
    \text{Pr}\left(
    \|\Tr_{I^c}(\rho(x^{(i)}))-\Tr_{I^c}(\sigma_T(\rho(x^{(i)})))\|_1\ge \epsilon\right)
    \\
    &\le 
    N
    C_k
    n^k 2^{k+1}
    \exp\left(-\frac{3T\epsilon^2}{8 \times 12^k}\right)\,.
\end{align}
If we set this equal to $\gamma/2$ we get that with $T$ as in the assumptions
we can achieve
with probability at least $1-\gamma/2$,
\begin{align}
    \|\Tr_{I^c}(\rho(x^{(i)})-\Tr_{I^c} (\sigma_T(\rho(x^{(i)})))\|_1
    \le 
    \frac{
    \epsilon_2}{C}\,,
\end{align}
for all $i=1,\dots,N$ and subsets $I\in {\cal P}_k$.
Then we can compute for each $P\in P_k$, 
\begin{align}
\label{eq:yp_x}
y^{(P)}(x^{(i)})=\Tr(P\sigma_T(\rho(x^{(i)})))\,,
\end{align}
which, with probability at least $1-\gamma/2$, satisfies, since $\|P\|=1$,
\begin{align}
    \label{eq:eps_2_shadows}
    |y^{(P)}(x^{(i)}) - \Tr(P\rho(x^{(i)}))|
    \le
    \|\Tr_{I^c}(\rho)-\Tr_{I^c} (\sigma_T(\rho))\|_1
    \le
    \frac{\epsilon_2}{C}\,,
\end{align}
thus producing a dataset $\{ 
x^{(i)}, y^{(P)}(x^{(i)}) \}_{i=1}^N$ for each $P\in P_k$.
Note that we have $N$ data points for each $P$, but this counts as $N$ samples only since they can all be produced by preparing $N$ different quantum states $\rho(x_\ell)$.

Now we learn a model $h^P_*$ for each of those $P$  using this dataset.
Called ${\cal D}$ the distribution over $x$, we know from Theorem \ref{thm:sample_complexity_single_obs} that if we choose
\begin{align}
    N 
    = 
    (\epsilon_3/C^2)^{-2}
    {\cal N}(\epsilon_1/C)
    \mathcal{O}(\log(1/\gamma')) 
    \,,\quad \gamma'=\gamma/(2 |P_k|)
    \,,
\end{align}
we have, denoting  $\epsilon=\tfrac{1}{C^2}((\epsilon_1+\epsilon_2)^2+\epsilon_3)$, that
\begin{align}
    \text{Pr}\left(
    \mathbb{E}_{x\sim {\cal D}}
    |
    h^P_*(x) - y^{(P)}(x)
    |\ge \epsilon
    \right)
    \le 
    \gamma'\,,
\end{align}
and this holds for any $P\in P_k$, since we have used the generalisation bound for each dataset $\{ 
x_\ell$, $y^{(P)}(x_\ell) \}_{\ell=1}^N$,  $P\in P_k$. 
Then, again by using the union bound:
\begin{align}
    \label{eq:union_bound_gen_error}
    \text{Pr}\left(
    \bigcup_{P\in P_k}
    \mathbb{E}_{x\sim {\cal D}}
    |
    h^P_*(x) - y^{(P)}(x)
    |\ge \epsilon
    \right)
    \le 
    \sum_{P\in P_k}
    \text{Pr}\left(
    \mathbb{E}_{x\sim {\cal D}}
    |
    h^P_*(x) - y^{(P)}(x)
    |\ge \epsilon
    \right)
    \le 
    |P_k|
    \gamma'
    =
    \frac{\gamma}{2}
\end{align}
conditioned on \eqref{eq:eps_2_shadows} to occur. The probability for $\mathbb{E}_{x\sim {\cal D}}
    |
    h^P_*(x) - y^{(P)}(x)
    |\le \epsilon$ for all $P\in P_k$
    and \eqref{eq:eps_2_shadows} is thus at least $1-\gamma$ again by the union bound.

We can then construct a ground state representation for all observables
\begin{align}
    \hat{\rho}(x)
    =
    \frac{1}{2^n}
    \sum_{P\in {P_k}}
    h_*^P(x) P
    \,,
\end{align}
which attains for any $O$ as in \eqref{eq:O_alpha_P} the bound:
\begin{align}
    &\mathbb{E}_{x\sim {\cal D}}|\Tr(O\rho(x))-\Tr(O\hat{\rho}(x))|^2
    = 
    \sum_{P,P'\in P_k}|\alpha_P\|\alpha_{P'}|
    \mathbb{E}_{x\sim {\cal D}}
    |
    h^P_*(x) - \Tr(P\rho(x))
    |
    |
    h^{P'}_*(x) - \Tr(P'\rho(x))
    |\\
    &\le 
    \sum_{P,P'\in P_k}|\alpha_P\|\alpha_{P'}|
    \sqrt{\mathbb{E}
    |
    h^P_*(x) - \Tr(P\rho(x))
    |^2}
    \sqrt{\mathbb{E}
    |
    h^{P'}_*(x) - \Tr(P'\rho(x))
    |^2}
    \\
    &\le 
    \left(\sum_{P\in P_k}|\alpha_P|\right)^2
    \tfrac{1}{C^2}((\epsilon_1+\epsilon_2)^2+\epsilon_3)
    =
    (\epsilon_1+\epsilon_2)^2+\epsilon_3\,.
\end{align}
In the first equality we used 
$\Tr(PP') = 2^n \delta_{P,P'}$, and in the second inequality we used H\"{o}lder's inequality.
\end{proof}

\subsection{Equivariant classical shadows}
\label{sec:Equivariant classical shadows}

In this section, we'll show how equivariance reduces the number of samples needed for predicting any observable of the form \eqref{eq:O_alpha_P}.
We will repeat the derivation of Theorem \ref{thm:sample_complexity_classical_shadows} in the case we have a non-trivial automorphism group of the interaction hypergraph ${\cal I}$.
Using the notation of Appendix \ref{sec:Classical shadows and prediction of many observables}, we first show the following result.
\begin{lemma}\label{lemma:equi_gen_error}
    If ${\cal D}$ is a probability distribution over $x$ that is invariant under $\text{Aut}({\cal I})$, and $h^P_*(x)$ is an equivariant model as in Proposition \ref{prop:equi_ml_model}, then
    \begin{align}
        \mathbb{E}_{x\sim {\cal D}}
        \left( |h^P_*(x)-y^P(x)|^2 \right)
        =
        \mathbb{E}_{x\sim {\cal D}}
        \left( |h^{P'}_*(x)-y^{P'}(x)|^2 \right)
        \,,\quad P'=\hat{g}P\hat{g}^{-1}\,,
    \end{align}
where $y^{(P)}(x)$ is the expectation value of $P$ in the classical shadow as in \eqref{eq:yp_x}.
\end{lemma}
\begin{proof}
Explicitly, we have
\begin{align}
\mathbb{E}_{x\sim {\cal D}}
\left( |h^P_*(x)-y^P(x)|^2 \right)
=    
\mathbb{E}_{x\sim {\cal D}}
\mathbb{E}_{y\sim \mu^P(\cdot|x)}
\left( |h^P_*(x)-y|^2 \right)\,,
\end{align}
with
\begin{align}
    \mu^P
    =
    F^P\# q
\end{align}
namely the push forward of the measure of the randomized measurements 
\begin{align}
    q(s|x)
    =
    \prod_{t=1}^T
    q_1(s^t|\rho(x))
    \,,\quad 
    q_1(s|\rho)=
    \frac{1}{3^n}
    \Tr(\ket{s}\bra{s} \rho)
\end{align}
under the map
\begin{align}
    F^P(s)
    =
    \frac{1}{T}
    \sum_{t=1}^T
    \Tr(P 
    \bigotimes_{i=1}^n (3\ket{s_i^t}\bra{s_i^t}-\id_2)
    )
    \,.
\end{align}
Note that the factor $3^n$ in $q_1(s|\rho)$ ensures the correct normalisation
\begin{align}    
    \sum_{s\in {\cal S}^n}
    q_1(s|\rho)
    =
    \Tr(
    \left(\frac{1}{3}\sum_{s\in {\cal S}} \ket{s}\bra{s}\right)^{\otimes n}
    \rho)
    =\Tr(\rho)=1
    \,,\quad
    {\cal S} = \{ 0_X,1_X,0_Y,1_Y,0_Z,1_Z\}\,,
\end{align}
which in turns ensures the normalisation of $q(s|x)$ which is a product distribution over $t$.
We can write $\mu^P$ explicitly as
\begin{align}
    \mu^P(y|x)
    =
    \sum_{s\in {\cal S}^{nT}}
    \delta(y-F^P(s))q(s|x)\,.
\end{align}
Now we show how these pieces transform under $P\to \hat{g}P\hat{g}^{-1}$. 
Since $\hat{g}\bigotimes_{i=1}^n \ket{s_i}=\bigotimes_{i=1}^n \ket{s_{gi}}$,
we have
\begin{align}
    F^{\hat{g}P\hat{g}^{-1}}(s)
    &=
    \frac{1}{T}
    \sum_{t=1}^T
    \Tr(\hat{g}P\hat{g}^{-1} 
    \bigotimes_{i=1}^n (3\ket{s_i^t}\bra{s_i^t}-\id_2)
    )
    =
    \frac{1}{T}
    \sum_{t=1}^T
    \Tr(P
    \bigotimes_{i=1}^n (3\ket{r_{i}^t}\bra{r_{i}^t}-\id_2)
    )\\
    &=
    F^{P}(r)
    \,,\quad
    r_i^t=
    (g^{-1}\cdot s)_i^t = s_{g^{-1}i}^t
    \,.
\end{align}
Also, from \eqref{eq:g_rho_g}
\begin{align}
    q_1(g\cdot s|\rho(x))
    =
    \frac{1}{3^n}
    \Tr(\hat{g}\ket{s}\bra{s}\hat{g}^{-1} \rho(x))
    =
    q_1(s|\rho(g\cdot x))
    \,,
\end{align}
and so
\begin{align}
    q(g\cdot s|x)
    =
    \prod_{t=1}^T
    q_1(g\cdot s^t|\rho(x))
    =
    \prod_{t=1}^T
    q_1(s^t|\rho(g\cdot x))
    =
    q(s|g\cdot x)\,.
\end{align}
Putting these together, with $s'=g^{-1}\cdot s$ and since $g \cdot {\cal S}^{nT}={\cal S}^{nT}$,
\begin{align}
    \mu^{\hat{g}P\hat{g}^{-1}}(y|x)
    &=
    \sum_{s\in {\cal S}^{nT}}
    \delta(y-F^{\hat{g}P\hat{g}^{-1}}(s))q(s|x)
    =
    \sum_{s\in {\cal S}^{nT}}
    \delta(y-F^{P}(g^{-1}\cdot s))q(s|x)
    \\
    &=
    \sum_{s'\in {\cal S}^{nT}}
    \delta(y-F^{P}(s'))q(g\cdot s'|x)
    =
    \sum_{s'\in {\cal S}^{nT}}
    \delta(y-F^{P}(s'))q(s'|g\cdot x)
    =\mu^P(y|g\cdot x)
    \,.
\end{align}
We also know from Proposition \ref{prop:equi_ml_model} that an equivariant model has weights such that:
\begin{align}
    h^{\hat{g}P\hat{g}^{-1}}_*(x)=h^P_*(g\cdot x)\,.
\end{align}
Then, with $x'=g\cdot x$,
\begin{align}
&\mathbb{E}_{x\sim {\cal D}}
\mathbb{E}_{y\sim \mu^{\hat{g}P\hat{g}^{-1}}(\cdot|x)}
\left( |h^{\hat{g}P\hat{g}^{-1}}_*(x)-y|^2 \right)
=
\mathbb{E}_{x\sim {\cal D}}
\mathbb{E}_{y\sim \mu^{P}(\cdot|g\cdot x)}
\left( |h^P_*(g\cdot x)-y|^2 \right)
\\
&=
\mathbb{E}_{x'\sim {\cal D}'}
\mathbb{E}_{y\sim \mu^{P}(\cdot|x')}
\left( |h^P_*(x')-y|^2 \right)
\,,\quad 
{\cal D}'=g\# {\cal D}
\,.
\end{align}
The result follows since we assume ${\cal D}'={\cal D}$.
\end{proof}
This gives the following corollary:
\begin{corollary}\label{cor:equiv_classical_shadows}
    Under the hypothesis of Theorem \ref{thm:sample_complexity_classical_shadows} and Lemma \ref{lemma:equi_gen_error}, we can predict all observables of the form \eqref{eq:O_alpha_P} with
    \begin{align}
    N 
    =
    (\epsilon_3/C^2)^{-2}
    {\cal N}(\epsilon_1/C)
    \mathcal{O}(\log(|P_k/\text{Aut}({\cal I})|/\gamma)) 
    \,.
\end{align}
\end{corollary}
\begin{proof}
    If we denote $P'_k = P_k/\text{Aut}({\cal I})$,
    the proof follows that of Theorem \ref{thm:sample_complexity_classical_shadows} and by noting that we can replace \eqref{eq:union_bound_gen_error}  with
    \begin{align}
    \text{Pr}\left(
    \bigcup_{P\in P_k}
    \mathbb{E}_{x\sim {\cal D}}
    |
    h^P_*(x) - y^{(P)}(x)
    |\ge \epsilon
    \right)
    =
    \text{Pr}\left(
    \bigcup_{P\in P'_k}
    \mathbb{E}_{x\sim {\cal D}}
    |
    h^P_*(x) - y^{(P)}(x)
    |\ge \epsilon
    \right)
    \le 
    |P'_k|
    \gamma'\,,
\end{align}
which equals $\gamma/2$ 
if we choose a $\gamma'=\gamma/(2|P_k'|)$ that leads to the result of the corollary.
\end{proof}

If $\mathcal{O}(|P_k/\text{Aut}({\cal I})|)=\mathcal{O}(1)$, then this reduces the complexity from $\mathcal{O}(\log(n))$ to $\mathcal{O}(1)$.
This is the case for example if we consider observables that are sum of geometrically local terms (which is the setting of \cite{lewis2023improved,onorati2023efficient,onorati2023provably}) on a lattice with periodic boundary conditions, which is a typical setting for numerical experiments, so that $\text{Aut}({\cal I})$ constrains the $D$-dimensional translation group.

\section{Technical Lemmas}
\label{sec:Technical Lemmas}

\subsection{Counting Lemmas}

\begin{lemma}\label{lemma:vol_area_sphere_latt}
    Let $\Lambda=\mathbb{Z}^D$ be the $D$-dimensional lattice. Then if $R>D$, there exists $C>0$ such that
    \begin{align}
        \sum_{i\in \Lambda} \delta(d(i,j), R)\le C R^{D-1}\,,\quad
        \sum_{i\in \Lambda} \id(d(i,j)\le R)\le C R^{D}\,,\quad \forall j\in \Lambda \,.
    \end{align}
\end{lemma}
\begin{proof}
    By translation invariance of the distance and $\Lambda$, we can evaluate each term at $j=0$, and then use the exact formulas for both quantities, see e.g.~\cite[Prop. 31]{janjic2013counting}:
    \begin{align}
        \sum_{i\in \Lambda} \delta(d(i,0), R)=
        \binom{R-1,D}{D-1,2}\,,\quad 
        \sum_{i\in \Lambda} \id(d(i,0)\le R)=
        \binom{R,D}{D,2}\,,\quad 
    \end{align}
    with
    \begin{align}
        \binom{m,n}{k,2}
        =
        2^{n-k}
        \sum_{i=0}^m 2^i \binom{m}{i}\binom{n}{k-i}\,.
    \end{align}
    So, assuming $R>D$ and recalling that $\binom{n}{k}=0$ if $k<0$ or $n-k<0$, there exist $C_1,C_2>0$ such that:
    \begin{align}
        \sum_{i\in \Lambda} \delta(d(i,j), R)&\le
        2
        \sum_{i=0}^{R-1}2^i \binom{R-1}{i}\binom{D}{D-1-i}
        =
        2
        \sum_{i=0}^{D-1}2^i \binom{R-1}{i}\binom{D}{D-1-i}
        \le C_1 (R-1)^{D-1}
        \,,
        \\ 
        \sum_{i\in \Lambda} \delta(d(i,j), R)&\le
        2
        \sum_{i=0}^{R}2^i \binom{R}{i}\binom{D}{D-i}
        =
        2
        \sum_{i=0}^{D}2^i \binom{R}{i}\binom{D}{D-i}
        \le C_2 (R-1)^{D}\,.
    \end{align}
    The Lemma follows by taking $C=\max(C_1,C_2)$ and noting that $(R-1)^d< R^d$ for $R,d>0$.
\end{proof}

\begin{lemma}\label{lemma:Mtilde}
    We have
    \begin{align}
        \widetilde{M}_\ell(ij,R)
        &=
        \sum_{i_1,\dots,i_\ell} \delta\left(\max_{1\le a<b\le \ell}(d_{ab}),R-1\right) \id(ij\in (i_1,\dots,i_\ell))\\
        &\le
        C_\ell
        \left( 
        \delta(d(i,j), R-1)
        +
        [\tfrac{1}{2}\ell(\ell-1)-1]
        \id(d(i,j)\le R-1)
        R^{- 1}  \right)
        R^{(\ell-2) D}  
        \,,
    \end{align}
    where we denoted $d_{ab}=d(i_a,i_b)$.
\end{lemma}
\begin{proof}
    Note that the summand in $\widetilde{M}_\ell(ij,R)$ is symmetric in $i_1,\dots,i_\ell$. Writing $\id(ij\in (i_1,\dots,i_\ell))$ as a sum where $ij$ can be any of the $\ell(\ell-1)/2$ pairs, we then have that each summand is the same, and we can write:
    \begin{align}
        \widetilde{M}_\ell(ij,R)=
        \tfrac{1}{2}\ell(\ell-1)
        \sum_{i_3,\dots,i_\ell} \delta\left(\max_{1\le a<b\le \ell}(d_{ab}),R-1\right) \,,\quad i_1=i,i_2=j\,.
    \end{align}
    We are further going to write this as a sum over the $\ell(\ell-1)/2$ pairs where the $a,b$ term corresponds to the case when ther maximum is attained at $i_a,i_b$:
    \begin{align}
        \widetilde{M}_\ell(ij,R)&=
        \tfrac{1}{2}\ell(\ell-1)
        \sum_{1\le a<b\le \ell} M_{\ell,ab}(ij,R)
        \,,\\ 
        \widetilde{M}_{\ell,ab}(ij,R)&=
        \sum_{i_3,\dots,i_\ell} \delta(d_{ab},R-1) 
        \prod_{\substack{1\le c<d\le \ell \\ cd\neq ab}} \id(d_{cd}\le R-1)
        \,.
    \end{align}
    Let us start by looking at $a,b=1,2$:
    \begin{align}
        \widetilde{M}_{\ell,12}(ij,R)
        &\le 
        \delta(d_{12},R-1) 
        \sum_{i_3,\dots,i_\ell} 
        \prod_{2\le c<d\le \ell } \id(d_{cd}\le R-1)\\
        &=
        \delta(d_{12},R-1) 
        \sum_{i_3,\dots,i_\ell} 
        \prod_{b=3}^\ell \id(d_{2,b}\le R-1)
        \prod_{b=4}^\ell \id(d_{3,b}\le R-1)
        \cdots
        \id(d_{\ell-1,\ell}\le R-1)
        \,.
    \end{align}
    In the first inequality we have used $\id(d_{1b}\le R-1)\le 1$.
    Then we can perform each summation in the order $i_\ell, i_{\ell-1}, i_{\ell-2}, \dots, i_3$. At step $i_\ell$ 
    \begin{align}
        \sum_{i_\ell}
        \prod_{a=2}^{\ell-1} \id(d_{a,\ell}\le R-1)
        \le 
        \sum_{i_\ell}
        \id(d_{\ell-1,\ell}\le R-1)
        \le
        C R^D\,,
    \end{align}
    since we can use the bound $\id(d_{a-1,b}\le R-1)\le 1$ and we used the result from Lemma \ref{lemma:vol_area_sphere_latt}. 
    Then, similarly, for $b=\ell-1, \dots, 3$
    \begin{align}
        \label{eq:sum_i_b}
        \sum_{i_b}
        \prod_{a=2}^{b-1} \id(d_{a,b}\le R-1)
        \le 
        \sum_{i_b}
        \id(d_{b-1,b}\le R-1)
        \le
        C R^D\,.
    \end{align}
    Thus we have shown
    \begin{align}
        \widetilde{M}_{\ell,12}(ij,R)
        &\le 
        \delta(d_{12},R-1) 
        (C R^D)^{\ell-2}
        \,.
    \end{align}
    Now we note that the sum in $\widetilde{M}_\ell(ij,R)$ is symmetric under permutations of $i_3,\dots,i_\ell$. So if $b>a>2,d>c>2$, then $\widetilde{M}_{\ell,ab}=\widetilde{M}_{\ell,cd}$.
    In fact, we can simply relabel the dummy indices to show this equivalence.  So to deal with the cases $b>a>2$, we can set $a=\ell-1,b=\ell$:
    \begin{align}
        \widetilde{M}_{\ell,ab}(ij,R)
        \le 
        \sum_{i_3,\dots,i_\ell} \delta(d_{\ell-1,\ell},R-1) 
        \id(d_{12}\le R-1)
        \prod_{d=2}^{\ell} \id(d_{1d}\le R-1)                
        \prod_{c=2}^{\ell-2}
        \prod_{d=c+1}^{\ell} \id(d_{cd}\le R-1)
    \end{align}
    Now we perform the sum over $i_\ell$:
    \begin{align}
        \sum_{i_\ell}
        \prod_{a=1}^{\ell-2} \id(d_{a,\ell}\le R-1)
        \delta(d_{\ell-1,\ell},R-1) 
        \le 
        \sum_{i_\ell}
        \delta(d_{\ell-1,\ell},R-1) 
        \le
        C R^{D-1}\,,
    \end{align}
    and for $i_{\ell-1}, i_{\ell-2}, \dots, i_3$ we have like in \eqref{eq:sum_i_b}.
    Thus if $b>a>2$,
    \begin{align}
        \widetilde{M}_{\ell,ab}(ij,R)
        \le 
        \id(d_{12}\le R-1)
        C^{\ell-2} R^{(\ell-2) D - 1}\,.
    \end{align}
    Then, assume that $a=1,b>2$. Again, any $b>2$ will give the same result due to symmetry and we choose $b=\ell$. Then
    \begin{align}
        \widetilde{M}_{\ell,1b}(ij,R)
        \le 
        \sum_{i_3,\dots,i_\ell} \delta(d_{1,\ell},R-1) 
        \id(d_{12}\le R-1)
        \prod_{d=2}^{\ell-1} \id(d_{1d}\le R-1)                
        \prod_{2\le c<d\le \ell}\id(d_{cd}\le R-1)\,.
    \end{align}
    Again, we perform the sum over $i_\ell$ first. We have
    \begin{align}
        \sum_{i_\ell}
        \delta(d_{1,\ell},R-1) 
        \prod_{a=2}^{\ell-1} \id(d_{a,\ell}\le R-1)
        \le 
        \sum_{i_\ell}
        \delta(d_{1,\ell},R-1) 
        \le
        C R^{D-1}\,,
    \end{align}
    and the other sums follow like in \eqref{eq:sum_i_b}.
    Then for $b>2$:
    \begin{align}
        \widetilde{M}_{\ell,1b}(ij,R)
        \le 
        \id(d_{12}\le R-1)
        C^{\ell-2} R^{(\ell-2) D - 1}\,.
    \end{align}
    Finally, since $i_1,i_2$ are symmetric, the same result holds for $a=2, b>a$. Putting everything together, we have
    \begin{align}
        \widetilde{M}_\ell(ij,R)&=
        C_\ell
        \left( 
        \delta(d_{12}, R-1)
        R^{(\ell-2) D}  
        +
        [\tfrac{1}{2}\ell(\ell-1)-1]
        \id(d_{12}\le R-1)
        R^{(\ell-2) D - 1}  \right)
        \,,
    \end{align}
    where we denoted $C_\ell = \tfrac{1}{2}\ell(\ell-1)C^{\ell-2}$.
\end{proof}

\begin{lemma}\label{lemma:Mell}
    For $\ell>1$, there exists a positive constant $C_\ell$ such that
    \begin{align}
    M_\ell(I, r, R)
    =
    \sum_{J, |J|=\ell}
    \delta(r, \diam(J))
    \delta(R, d_{IJ})
    \le 
    C_{\ell} r^{(\ell-1)D-1} R^{D-1}\,.
    \end{align}
\end{lemma}
\begin{proof}
We want to compute, defining $r'=r-1$ and $d_{ab}\equiv d(j_a,j_b)$, the value of
\begin{align}
    M_\ell(I, r, R)=
    \sum_{j_1,\dots,j_\ell}
    \delta\left(r', \max_{1\le a<b\le \ell}(d_{ab})\right)
    \delta\left(R, \min_{i\in I, j\in J}d(i,j)\right)\,.
\end{align}
First of all, we note that
\begin{align}
    \delta\left(R, \min_{i\in I, j\in J}d(i,j)\right)
    \le 
    \sum_{a=1}^{|I|}
    \sum_{b=1}^{|J|}
    \delta(R, d(i_a,j_b))
    \,.
\end{align}
This is because if the $\delta$ on the l.h.s.~is $1$, then the sum on the r.h.s.~must be $\ge 1$. Then
\begin{align}
    M_\ell(I, r, R)
    \le 
    \sum_{j_1,\dots,j_\ell}
    \sum_{c=1}^{|I|}
    \sum_{d=1}^{|J|}
    \delta\left(r', \max_{1\le a<b\le \ell}(d_{ab})\right)
    \delta(R, d(i_c,j_d))
    \,.
\end{align}
Now note that the function we are summing over is invariant under permutation of the $j$'s.
Therefore, each choice of $j_d$ will give the same result and we can set $d=1$:
\begin{align}
    M_\ell(I, r, R)
    \le 
    \ell
    \sum_{c=1}^{|I|}
    \sum_{j_1,\dots,j_\ell}
    \delta\left(r', \max_{1\le a<b\le \ell}(d_{ab})\right)
    \delta(R, d(i_c,j_1))
    \,.
\end{align}
Then we consider each of the cases where the  maximum is attained at each of  the $\tfrac{1}{2}\ell(\ell-1)$ pairs $j_a,j_b$:
\begin{align}
    M_\ell(I, r, R)
    &\le 
    \ell
    \sum_{c=1}^{|I|}
    \sum_{1\le a<b\le \ell}
    M_{\ell}(a,b,c,r',R)
    \\
    M_{\ell}(a,b,c,r',R)&=
    \sum_{j_1,\dots,j_\ell}
    \delta(r', d_{ab})
    \prod_{\substack{1\le e<f\le \ell \\ ef\neq ab}} \id(d_{ef}\le r')    
    \delta(R, d(i_c,j_1))
    \,.
\end{align}
We consider first $a,b=1,2$:
\begin{align}
    M_{\ell}(1,2,c,r',R)
    =
    \sum_{j_1,\dots,j_\ell}
    \delta(r', d_{12})
    \prod_{f=3}^\ell \id(d_{1f}\le r')    
    \prod_{2\le e<f\le \ell} 
    \id(d_{ef}\le r')    
    \delta(R, d(i_c,j_1))
    \,.
\end{align}
We now perform the sum over $j_2$. Using $\id(\cdot )\le 1$ and Lemma \ref{lemma:vol_area_sphere_latt}:
\begin{align}
    \sum_{j_2} 
    \delta(r', d_{12})
    \prod_{f=3}^\ell
    \id(d_{2f}\le r')   
    \le 
    \sum_{j_2} 
    \delta(r', d_{12})
    \le 
    C (r')^{D-1}\,.
\end{align}
Now summing sequentially over $j_3,j_4,\dots,j_\ell$, gives at $a$-th step:
\begin{align}
    \sum_{j_a} 
    \prod_{f=a+1}^\ell
    \id(d_{af}\le r')   
    \le 
    \sum_{j_a} 
    \id(d_{a,a+1}\le r')   
    \le 
    C (r')^{D}\,.
\end{align}
Finally, we sum over $j_1$:
\begin{align}
    \sum_{j_1} 
    \prod_{f=3}^\ell \id(d_{1f}\le r')    
    \delta(R, d(i_c,j_1))    
    \le 
    \sum_{j_1} 
    \delta(R, d(i_c,j_1))    
    \le 
    C R^{D-1}\,,
\end{align}
which gives
\begin{align}
    M_{\ell}(1,2,c,r',R)
    \le 
    C^\ell 
    R^{D-1}
    (r')^{(\ell-1)D-1}
    \,.
\end{align}
Now note that $M_{\ell}(1,2,c,r',R)=M_{\ell}(1,b,c,r',R)$ for any $b\ge 2$ since the summand is symmetric under exchanging $j_b$ and $j_2$.
Then consider $M_{\ell}(a,b,c,r',R)$, $b>a>1$. Again by symmetry these are all equal since we can relabel any $a,b$ with $b>a>1$ as $\ell-1,\ell$, and we can consider for definiteness 
\begin{align}
    M_{\ell}(\ell-1,\ell,c,r',R)
    =
    \sum_{j_1,\dots,j_\ell}
    \delta(r', d_{\ell-1,\ell})
    \prod_{e=1}^{\ell-2}
    \prod_{f=e+1}^\ell
    \id(d_{ef}\le r')    
    \delta(R, d(i_c,j_1))\,.
\end{align}
We start with summing over $j_{\ell-1}$:
\begin{align}
    \sum_{j_{\ell-1}}
    \delta(r', d_{\ell-1,\ell})
    \prod_{e=1}^{\ell-2} \id(d_{e,\ell-1}\le r')
    \le 
    \sum_{j_{\ell-1}}
    \delta(r', d_{\ell-1,\ell})
    \le 
    C (r')^{D-1}\,.
\end{align}
Now, any other sum apart from $j_1$ can be bounded by $C (r')^D$ as done before. The sum over $j_1$ gives:
\begin{align}
    \sum_{j_1}
    \prod_{f=2}^\ell
    \id(d_{1f}\le r')    
    \delta(R, d(i_c,j_1))
    \le 
    \sum_{j_1}
    \delta(R, d(i_c,j_1))
    \le 
    C R^{D-1}\,, 
\end{align}
so that
\begin{align}
    M_{\ell}(\ell-1,\ell,c,r',R)
    \le C^\ell R^{D-1} (r')^{(\ell-1)D-1}\,,
\end{align}
and thus
\begin{align}
    M_\ell(I, r, R)
    \le 
    \ell C^\ell |I| \tfrac{1}{2}\ell(\ell-1)
        R^{D-1} (r')^{(\ell-1)D-1}
\end{align}
\end{proof}

We have the following immediate consequence:
\begin{corollary}\label{corollary:M}
We have
    \begin{align}
    M(I, r, R)
    =
    \sum_{J\in {\cal P}_k(\Lambda)}
    \delta(r, \diam(J))
    \delta(R, d_{IJ})
    =
    \sum_{\ell=1}^k
    M_\ell(I, r, R)
    \le 
    C_{\text{sum}}
    r^{(k-1)D-1}
    R^{D-1}
    \,.    
\end{align}
\end{corollary}

\subsection{Solutions to inequalities}

\begin{lemma}[Lemma 6 in \cite{lewis2023improved}]\label{lemma:6}
    Given $a>0,s>1$, $0<\epsilon < \e^{-1}$, there exists a constant $c>0$ such that for all $y\ge c\log^2(1/\epsilon)$,
    \begin{align}
        y^s \e^{-a y} \le \epsilon\,.
    \end{align}
\end{lemma}
\begin{proof}
    We can rewrite the condition we want to prove taking $-\log$ of  both sides as
    \begin{align}
        ay - s \log(y)\ge \log(1/\epsilon)\,.
    \end{align}
    We can then apply \cite[Lemma 6]{lewis2023improved} to show that the Lemma holds with 
    \begin{align}
        c = \frac{(2s+\sqrt{4s^2+a})^2}{4a^2}\,.
    \end{align}
\end{proof}

\begin{lemma}[Adapting Lemma 7 in \cite{lewis2023improved}]\label{lemma:7}
    Given $a,s>0$, $0<\epsilon \le \e^{-1}$, there exists a constant $c>0$ such that for all $y\ge c\log^2(1/\epsilon)$,
    \begin{align}
        y^s u_a(y) \le \epsilon\,.
    \end{align}
\end{lemma}
\begin{proof}
    We assume $y\ge 0$. We can rewrite the condition to prove as
    \begin{align}
        f(y)\coloneqq a\frac{y}{\log^2(y)}-s\log(y)\ge x\,,\quad x = \log(1/\epsilon)\,.
    \end{align}
    We compute 
    \begin{align}
    f'(y) = \frac{a}{\log^3(y)}\left( 
    \log(y)- 2\right)  
    -
    \frac{s}{y}\,,
    \end{align}
    and find $y_0$ such that $f'(y)\ge 0$ for $y\ge y_0$.
    We use $\log(z)\le 3 z^{1/3}$, or $1/\log^3(z)\ge 1/(3^3 z)$ for $z>0$ to get
    \begin{align}
    f'(y)\ge     
    \frac{1}{y}(a3^{-3}(\log(y)-2)-s)
    \end{align}
    so we need $a3^{-3}(\log(y)-2)-s\ge 0$ or $\log(y)-2-\frac{s}{a} 3^3\ge 0$, namely
    \begin{align}
        y\ge y_0 \equiv \exp(2+\frac{s}{a}3^3)\,.
    \end{align}
    We note that numerically we can see that this bound is quite loose but it will serve our purposes here.
    So for $y\ge y_0$, $f$ is monotonically increasing.
    Next we will show that for $x>1$ ($\epsilon<1/\e$), there
    is a $c\ge y_0$ such that $f(cx^2)\ge x$. This will imply that $f(y)\ge x$ for all $y\ge cx^2$ due to the monotonicity of $f$ for $cx^2\ge y_0$ that is ensured for $x\ge 1$ as long as $c\ge y_0$. 
    To show 
    we can get the desired bound:
    \begin{align}
        f(cx^2) = a\frac{cx^2}{\log^2(cx^2)}-s\log(cx^2)
        \ge x\,,
    \end{align}
    we will show the two inequalities
    \begin{align}
        \frac{a}{2}\frac{cx^2}{\log^2(cx^2)}\ge x
        \,,\quad 
        \frac{a}{2}\frac{cx^2}{\log^2(cx^2)}
        \ge s\log(cx^2)
        \,.
    \end{align}
    For the first, 
    we have, using $\log(z)\le 4z^{1/4}$ for $z\ge 0$,
    \begin{align}
        \frac{a}{2}\frac{cx^2}{\log^2(cx^2)}
        \ge 
        \frac{a}{32}
        \frac{cx^2}{\sqrt{cx^2}}
        =
        \frac{a}{32}\sqrt{c}x
        \ge x\,\quad 
        \text{if }c\ge \frac{32^2}{a^2}\,.
    \end{align}
    For the second, since $x\ge 1$, $a cx^2\ge a c x^{3/2}$. Using again $\log(z)\le 4z^{1/4}$ with $z = cx^2$, we have
    \begin{align}
        a cx^2
        \ge a c^{1/4} ((cx^2)^{1/4})^3 
        \ge \frac{a c^{1/4}}{4^3} \log^3(cx^2)\,,
    \end{align}
    so
    \begin{align}
        \frac{a}{2}\frac{cx^2}{\log^2(cx^2)}
        \ge 
        \frac{a c^{1/4}}{2\times 4^3} \log(cx^2)
        \,,
    \end{align}
    and this is greater or equal to $s\log(cx^2)$ if
    \begin{align}
        \frac{a c^{1/4}}{2\times 4^3} \ge s
        \Rightarrow
        c\ge 2^4 4^{12}\frac{s^4}{a^4} \,.
    \end{align}
    Finally, under the assumption
    \begin{align}
        c \ge \max\left(y_0, \frac{32^2}{a^2}, 2^44^{12}\frac{s^4}{a^4} \right)\,,
    \end{align}
    we can get the desired bound.
\end{proof}

\subsection{Integral bounds}

We use 
\begin{lemma}[\cite{pinelis2020exact}]
For any $a\ge  2, x>0$:
\begin{align}
    \Gamma(a,x)
    \le 
    \frac{(x+b_a)^a-x^a}{ab_a}\e^{-x}
    \,,\quad
    b_a
    =
    \Gamma(a+1)^{1/(a-1)}\,.
\end{align}
\end{lemma}
We can simplify this bound for $x\ge 1$, so that $(1+b_a)x = x + b_ax \ge x+b_a$, as follows
\begin{align}
\label{eq:gamma_bound}
    \Gamma(a,x)
    \le 
    \frac{(1+b_a)^a-1}{ab_a}
    x^a\e^{-x}
    \,.
\end{align}

We are going to use this bound on the incomplete Gamma function to prove an extension of Lemma 2.5 \cite{Bachmann_2011} to non-integral powers $s$.
\begin{lemma}\label{lemma:extension_2.5}
For all $t>1,s>3$ such that $\log^4(t)\le t$ and
    \begin{align}
    \tau(t)\coloneqq a\frac{t}{\log^2(t)}\ge b_s = \Gamma(s+1)^{1/(s-1)}\,,
    \end{align}
there exists a $t$-independent positive number $C$ such that:
    \begin{align}
        F(s,t)=
        \int_t^\infty \dd \eta\, \eta^s u_a(\eta)
        \le 
        C t^{2s+3}u_a(t)\,.
    \end{align}
\end{lemma}
\begin{proof}
We have
\begin{align}
    \frac{\dd \eta}{\dd \tau(\eta)}
    =
    \left( a\frac{1}{\log^2(\eta)}-2a\frac{1}{\log^3(\eta)} \right)^{-1}
    =
    \frac{1}{a}\frac{\log^2(\eta)}{1-\frac{2}{\log(\eta)}}\le \frac{\eta}{a}
    \,,
\end{align}
and from assumption, $\log^4(\eta)\le \eta$, so
\begin{align}
    \eta\le \frac{\eta^2}{\log^4(\eta)} = \frac{\tau^2}{a^2}\,.
\end{align}
Then we change the variable from $\eta$ to $\tau$:
\begin{align}
    F(s,t)
    \le
    \frac{1}{a^{2s+3}}
        \int_{\tau(t)}^\infty 
        \dd\tau \tau^{2(s+1)} \e^{-\tau}
        =
\frac{1}{a^{2s+3}}
    \Gamma(2s+3,\tau(t))\,.
    \end{align}
    Now we use the bound \eqref{eq:gamma_bound} to get
    \begin{align}
    F(s,t)
    \le
    C'
    \tau(t)^{2s+3}u_a(t)
    \le 
a^{2s+3}    C' t^{2s+3}u_a(t)
    \,,
    \quad
    C'=
    \frac{(1+b_{2s+3})^{2s+3}-1}{(2s+3)b_{2s+3}}
    \,.
    \end{align}
\end{proof}

\section{Lieb-Robinson bounds for power law interactions}
\label{sec:Lieb Robinson bounds for power law interactions}

The following result covers the case of two body interactions.
\begin{theorem}[Thm 1, S1 \cite{Tran_2021}]\label{thm:S1Tran2021}
    Assume that $H=\sum_{ij}H_{ij}$ with $\|H_{ij}\|\le 1/d(i,j)^\alpha$.
    Then let $O$ be a norm $1$ operator supported at the origin, and define $\mathbb{P}_rO$ the projection of $O$ onto sites that are at least of distance $r$ from the origin.
    For any $\alpha\in (2D,2D+1)$, $\epsilon\in (0, \epsilon_*)$, $\epsilon_* = \frac{(\alpha-2D)^2}{(\alpha-2D)^2+\alpha-D}$, there exist constants $c,C_1,C_2\ge 0$ s.t.
    \begin{align}
    \|\mathbb{P}_rO(t)\|
    \le
    f(t,r)\,,\quad
f(t,r)=
    C_1 \left( 
    \frac{t}{r^{\alpha-2D-\epsilon}}
    \right)^{\frac{\alpha-D}{\alpha-2D}-\frac{\epsilon}{2}}
    +
    C_2
    \frac{t}{r^{\alpha-D}}
    \,,
    \end{align}
    for all $0\le t\le c r^{\alpha-2D-\epsilon}$.
\end{theorem}
In \cite{Tran_2021} this result was also conjectured to be valid for the more general class of Hamiltonians 
\begin{align}
    H=\sum_{I\in {\cal P}(\Lambda)} h_I\,,\quad 
    \sum_{I\ni i,j}\|h_I\|\le g d(i,j)^{-\alpha}\,,
\end{align}
for a constant $g$, and this is the form that we are going to use.
We derive the following corollary.
\begin{corollary}\label{corollary:lr_alpha_ge_2d}
    For any $\alpha\in (2D,2D+1)$, $\epsilon\in (0, \epsilon_*)$, $\epsilon_* = \frac{(\alpha-2D)^2}{(\alpha-2D)^2+\alpha-D}$, there exist constants $c,C_1,C_2\ge 0$ s.t. for all $0\le t\le c r^{\alpha-2D-\epsilon}$ and $|I|, |J|=\mathcal{O}(1)$
    \begin{align}
    \|[O_I(t),O_J]\|
\le
\|O_I\|\,\|O_J\|\,
f(t, d(I,J)) 
    \,.
\end{align}
\end{corollary}
\begin{proof}
We assume unit norm operators since for general operators we can apply the result to the normalized operator.
For a set $I$, let $P_I$ be the projector onto the non-trivial part of a Pauli string on $I$, i.e.~it is zero unless there is a non-trivial operator in $I$:
\begin{align}
    P_I
    \sigma_1^{a_1}\cdots
    \sigma_n^{a_n}
    = 
    \id(
    \exists j \in I\,|\,
    a_j\neq 0)
    \sigma_1^{a_1}\cdots
    \sigma_n^{a_n}
    =
    \id\Big(
    \bigvee_{j\in I}(a_j\neq 0)
    \Big)
    \sigma_1^{a_1}\cdots
    \sigma_n^{a_n}
    \,.
\end{align}
Then define $\mathbb{P}^{I}_r$ 
to be the projector onto the set $B_{r}^I=\{ j\,|\, d(I,j)\ge r\}$. 
$\mathbb{P}_r$ of the Theorem \ref{thm:S1Tran2021} coincides with $\mathbb{P}^{0}_r$. 
Since $B_{r}^I = \bigcap_{i\in I}B_r^i$ and $I(\vee_{i\in {I\cap J}}(a_i\neq 0))\le I(\vee_{i\in {I}}(a_i\neq 0))+I(\vee_{i\in {J}}(a_i\neq 0))$  because there are more configurations of the $a_i$'s that make the r.h.s.~non-zero, we have
\begin{align}
    \mathbb{P}^{I}_r\le
    \sum_{i\in I}\mathbb{P}^{i}_r\,.
\end{align}
Then note that 
$(\id - \mathbb{P}^{I}_r)O$ for any operator $O$ is supported on sites of distance smaller than $r$ to $I$.
Then if $r=d(I,J)$, we have that $[(\id - \mathbb{P}^{I}_r)O,O_J]=0$ for any $O$ since the support of the two operators in the commutator has no overlap.
Then, with $r=d(I,J)$, we have
\begin{align}
    \|[O_I(t),O_J]\|
    &=
    \|[\mathbb{P}^{I}_rO_I(t)+(\id-\mathbb{P}^{I}_r)O_I(t),O_J]\|
    =
    \|[\mathbb{P}^{I}_rO_I(t),O_J]\|
    \le 
    \sum_{i\in I}
    \|[\mathbb{P}^{i}_rO_I(t),O_J]\|
    \\    
    &\le 2 
    \sum_{i\in I}    
    \|\mathbb{P}^{i}_rO_I(t)\|
    \,.    
\end{align}
Then we use formula \cite[Lemma 4]{Tran_2020}: if $\|P_J O_i(t)\|\le f(t, d(i,J))$ for all unit norm operators $O_i$ with $i\in I$, then there exists a positive constant $K$ such that
\begin{align}
\|P_J O_I(t)\|\le 
K \sum_{i\in I}f(t, d(i,J))\,.
\end{align}
From Theorem \ref{thm:S1Tran2021} 
we know
\begin{align}
    \|\mathbb{P}^{i}_rO_i(t)\|
    \le f(t,r)\,,
\end{align}
so we have that
\begin{align}
\|\mathbb{P}^{i}_rO_I(t)\|
\le
K
|I| f(t,r)\,,
\end{align}
and putting things together:
\begin{align}
    \|[O_I(t),O_J]\|
    \le 2 K 
    \, |I|^2     f(t,r)
    \,.    
\end{align}
\end{proof}

\section{Implementation of Matrix Product Operators}
\label{sec: Implementation MPO}

\subsection{Heisenberg model}
\label{sec: Implementing Heisenberg}

Implementing the nearest neighbour Heisenberg model \begin{equation}H = \sum\limits_{\langle i j \rangle} J_{ij} (X_i X_j + Y_i Y_j + Z_i Z_j)\end{equation} with open boundary conditions is straightforward. Following the discussion in Appendix \ref{sec:PeriodicBC}, the Hamiltonian on a periodic chain may be represented as an MPO in the following way: 
\begin{equation}H = W_L \cdot \prod_{i=1}^{n-2} W[i] \cdot  W_R,\end{equation}
where the grid matrices are given by
$$W[i] = \left(\begin{matrix}
I & 0 & 0 & 0 & J_{i,i+1} X_i & J_{i,i+1} Y_i & J_{i,i+1} Z_i & 0\\
0 & I & 0 & 0 &0 &0 & 0& 0\\
0 & 0 & I & 0 &0 &0 & 0&0\\
0 & 0& 0 & I &0 &0 & 0& 0\\
0 & 0& 0 & 0 &0 & 0& 0& X_i\\
0 & 0& 0 & 0 &0 & 0& 0& Y_i\\
0 & 0& 0 & 0 &0 & 0& 0& Z_i\\
0 & 0& 0 & 0 &0 & 0& 0& I\\
\end{matrix} \right)$$
together with
$$W_L = (\begin{matrix} I & J_{0,n-1} X_0& J_{0,n-1} Y_0& J_{0,n-1} Z_0 & J_{01} X_0 & J_{01} Y_0 & J_{01} Z_0 & 0 \end{matrix}),$$
$$W_R = (\begin{matrix} 0 & X_{n-1}& Y_{n-1}& Z_{n-1} & X_{n-1} & Y_{n-1} & Z_{n-1} & I \end{matrix})^T.$$

\subsection{Ising model}
\label{sec: Implementing Ising}

The  Hamiltonian 
\begin{equation}H = \sum\limits_{i<j} \frac{1+J_i J_j}{d(i,j)^\alpha} Z_i  Z_j + \sum\limits_i h_i  X_i,\end{equation} may be represented as a matrix product operator using matrices of the form 
$$W[i] = \left(\begin{matrix}
I & \lambda_1 J_i Z_i & \lambda_2 J_i Z_i &\dots & \lambda_K J_i Z_i & \lambda_1 Z_i & \dots & \lambda_K Z_i & h_i X_i\\
0 & \lambda_1 I & 0 & \dots &0 &0 &\dots & 0 & a_1 J_i Z_i\\
0 & 0 & \lambda_2 I &  \dots & 0 & 0 & \dots & 0 & a_2 J_i Z_i\\
\vdots & \vdots & \vdots & \ddots & \vdots & \vdots & & \vdots & \vdots\\
0 & 0 & 0& \dots & \lambda_K I & 0&\dots & 0 & a_k J_i Z_i\\
0 & 0 & 0 & \dots & 0 & \lambda_1 I &\dots & 0 & a_1 Z_i\\
\vdots & \vdots & \vdots &  & \vdots &\vdots  &\ddots & \vdots & \vdots\\
0 & 0 & 0 & \dots & 0 &0 &\dots & \lambda_K I & a_K Z_i\\
0 & 0 & 0 & \dots & 0& 0& \dots & 0 & I\\
\end{matrix}\right)$$
together with
$$W_L  = \left(\begin{matrix}
I & \lambda_1 J_0 Z_0 & \lambda_2 J_0 Z_0 & \dots & \lambda_K J_0 Z_0 & \lambda_1 Z_0 & \dots & \lambda_K Z_0 & h_0 X_0
\end{matrix} \right),$$
$$W_R  = \left(\begin{matrix}
h_{n-1} X_{n-1} & a_1 J_{n-1} Z_{n-1} & a_2 J_{n-1} Z_{n-1} & \dots & a_K J_{n-1} Z_{n-1} & a_1 Z_{n-1} & \dots & a_K Z_{n-1} & I
\end{matrix} \right)^T$$
like 
\begin{equation}H = W_L \cdot \prod_{i=1}^{n-2} W[i] \cdot  W_R,\end{equation}
where we used exponential-sum fitting to approximate the decay of interactions as
$$\frac{1}{d(i,j)^\alpha} \approx \sum_{l=1}^K a_l \lambda_l ^{|i-j|},$$
using an algorithm described in the Appendix of \cite{Pirvu_2010}. For periodic boundary conditions, see the approach described in Appendix \ref{sec:PeriodicBC}.

\subsection{Rydberg model}

The Hamiltonian \begin{equation}
   H =  \sum\limits_{i<j} \frac{V}{|i+\delta_i - j - \delta_j|^6} N_i N_j + \sum\limits_i \left( \frac{\Omega}{2} X_i + \Delta N_i \right),
\end{equation} 
 may be represented as a matrix product operator using matrices of the form 
$$W[i] = \left(\begin{matrix}
I & V \lambda_1^{1-\delta_i} N_i & V\lambda_2^{1-\delta_i} N_i &\dots & V \lambda_K^{1-\delta_i} N_i & \frac{\Omega}{2} X_i + \Delta N_i\\
0 & \lambda_1 I & 0 & \dots &0  & a_1 \lambda_1^{\delta_i} N_i\\
0 & 0 & \lambda_2 I &  \dots & 0 & a_2 \lambda_2^{\delta_i} N_i\\
\vdots & \vdots & \vdots & \ddots & \vdots & \vdots\\
0 & 0 & 0& \dots & \lambda_K I  & a_k \lambda_K^{\delta_i} N_i\\
0 & 0 & 0 & \dots & 0 & I\\
\end{matrix}\right)$$
together with
$$W_L = \left(\begin{matrix}
    I & V \lambda_1^{1-\delta_0} N_0 & V\lambda_2^{1-\delta_0} N_0 &\dots & V \lambda_K^{1-\delta_0} N_0 & \frac{\Omega}{2} X_0 + \Delta N_0
\end{matrix} \right),$$
$$W_R  = \left(\begin{matrix}
\frac{\Omega}{2} X_{n-1} + \Delta N_{n-1} & a_1 \lambda_1^{\delta_{n-1}} N_{n-1} & a_2 \lambda_2^{\delta_{n-1}} N_{n-1} & \dots & a_k \lambda_K^{\delta_{n-1}} N_{n-1} & I
\end{matrix} \right)^T$$
like 
\begin{equation}H = W_L \cdot \prod_{i=1}^{n-2} W[i] \cdot  W_R,\end{equation}
where we again used exponential-sum fitting to approximate the decay of interactions as
$$\frac{1}{|i-j|^6} \approx \sum_{l=1}^K a_l \lambda_l ^{|i-j|}$$ like previously. This exponential fit further allowed us to separate the site displacements from the overall decay of interactions and hence implement this nontrivial modified geometry using tensor networks.

\subsection{Implementation of periodic boundary conditions}
\label{sec:PeriodicBC}

All particular details regarding specific implementations for both open and periodic boundary conditions for specific models are discussed in their corresponding sections. In this appendix, we look at some general ideas behind implementing periodic boundary conditions. For practical reasons, the simulations of periodic boundary conditions in larger systems were implemented using open boundary tensor networks. 

When considering only a nearest-neighbour coupling on a ring, it is a common practice to use a so-called folded ordering of the MPO, which transforms the model from having only nearest-neighbour interactions with one additional interaction over the length of the whole chain, into a model with nearest-neighbour and next-nearest-neighbour interactions. However, when carrying out the simulations for the periodic Heisenberg model, we found it to be faster to implement this long-range interaction using a non-uniform exponentially decaying coupling, where the base would be equal to 1 (so that there wouldn't actually be any decay), and the only non-zero interaction would be between the first and last site. The implementation for this case is shown in Appendix \ref{sec: Implementing Heisenberg}, while the general non-uniform exponential coupling is discussed in more detail in Appendix \ref{sec: Implementing Ising}.

When considering general long-range interactions, the periodic boundary conditions were implemented by appropriately changing the metric in the interaction terms. The metric specifying one-dimensional chain with open boundary conditions is simply $d(i,j)=|i-j|$, while the corresponding metric for periodic boundary conditions is $d(i,j) = \min\{|i-j|,n-|i-j|\}$. The newly obtained symmetric decay rate of interactions was then  fitted using correspondingly symmetric sum of exponentials $$\frac{1}{d(i,j)^\alpha} \approx \sum_{l=1}^K c_l \cdot \left(b_l ^ {|i-j|} + b_l ^{n-|i-j|}\right).$$

As this task proved to be numerically problematic, we've devised the following approach: denote $\sum_{l=1}^k c_l \cdot b_l ^ x$ by $g(x)$, so that the desired polynomial decay rate is given by $\frac{1}{d(x)^\alpha} \approx g(x)+g(n-x)$, where $x=|i-j|$. If we can find a function $g(x)$ obeying this relation, which is decaying and smooth enough, we could use a robust non-iterative algorithm to fit it with a sum of exponentials \cite{Pirvu_2010}. A suitable function for this task was found to be
$$g(x) = \begin{cases}\frac{1}{x^\alpha} - \frac{1/2}{(n-x)^\alpha} & \text{for }x \leq \frac{n}{2},\\ \frac{1/2}{x^\alpha} & \text{for }x \geq \frac{n}{2}. \end{cases}$$
Hence, after obtaining the fitting parameters $\{c_l, b_l\}_{l=1}^K$, we would create the MPO using $\{a_l, \lambda_l\}_{l=1}^{2K} = \{c_l, b_l\}_{l=1}^K \cup \{c_l \cdot b_l ^n, \frac{1}{b_l}\}_{l=1}^K$ as the parameters for the exponential fit to correctly account for the symmetry.

\pagebreak
\nocite{*}

\bibliographystyle{quantum}
\bibliography{bibliography}

\begin{thebibliography}{10}

\bibitem{parr1994density}
R.G. Parr and Y.~Weitao.
\newblock ``Density-functional theory of atoms and molecules''.
\newblock \href{https://dx.doi.org/10.1093/oso/9780195092769.001.0001}{International Series of Monographs on Chemistry}. Oxford University Press. ~(1994).

\bibitem{RevModPhys.73.33}
W.~M.~C. Foulkes, L.~Mitas, R.~J. Needs, and G.~Rajagopal.
\newblock ``Quantum {M}onte {C}arlo simulations of solids''.
\newblock \href{https://dx.doi.org/10.1103/RevModPhys.73.33}{Rev. Mod. Phys. {\bf 73}, 33--83}~(2001).

\bibitem{RevModPhys.93.045003}
J.~Ignacio Cirac, David P\'erez-Garc\'{\i}a, Norbert Schuch, and Frank Verstraete.
\newblock ``Matrix product states and projected entangled pair states: Concepts, symmetries, theorems''.
\newblock \href{https://dx.doi.org/10.1103/RevModPhys.93.045003}{Rev. Mod. Phys. {\bf 93}, 045003}~(2021).

\bibitem{zitnick2020introduction}
C.~Lawrence Zitnick, Lowik Chanussot, Abhishek Das, Siddharth Goyal, Javier Heras-Domingo, Caleb Ho, Weihua Hu, Thibaut Lavril, Aini Palizhati, Morgane Riviere, Muhammed Shuaibi, Anuroop Sriram, Kevin Tran, Brandon Wood, Junwoong Yoon, Devi Parikh, and Zachary Ulissi.
\newblock ``An introduction to electrocatalyst design using machine learning for renewable energy storage''~(2020).
\newblock  \href{http://arxiv.org/abs/2010.09435}{arXiv:2010.09435}.

\bibitem{dalzell2023quantum}
Alexander~M. Dalzell, Sam McArdle, Mario Berta, Przemyslaw Bienias, Chi-Fang Chen, András Gilyén, Connor~T. Hann, Michael~J. Kastoryano, Emil~T. Khabiboulline, Aleksander Kubica, Grant Salton, Samson Wang, and Fernando G. S.~L. Brandão.
\newblock ``Quantum algorithms: A survey of applications and end-to-end complexities''~(2023).
\newblock  \href{http://arxiv.org/abs/2310.03011}{arXiv:2310.03011}.

\bibitem{tindall2023efficient}
Joseph Tindall, Matt Fishman, Miles Stoudenmire, and Dries Sels.
\newblock ``Efficient tensor network simulation of {IBM}'s {E}agle kicked {I}sing experiment''~(2023).
\newblock  \href{http://arxiv.org/abs/2306.14887}{arXiv:2306.14887}.

\bibitem{jumper2021highly}
John Jumper, Richard Evans, Alexander Pritzel, Tim Green, Michael Figurnov, Olaf Ronneberger, Kathryn Tunyasuvunakool, Russ Bates, Augustin {\v{Z}}{\'\i}dek, Anna Potapenko, et~al.
\newblock ``Highly accurate protein structure prediction with {A}lphafold''.
\newblock \href{https://dx.doi.org/10.1038/s41586-021-03819-2}{Nature {\bf 596}, 583--589}~(2021).

\bibitem{gilmer2017neural}
Justin Gilmer, Samuel~S. Schoenholz, Patrick~F. Riley, Oriol Vinyals, and George~E. Dahl.
\newblock ``Neural message passing for quantum chemistry''~(2017).
\newblock  \href{http://arxiv.org/abs/1704.01212}{arXiv:1704.01212}.

\bibitem{Carleo_2017}
Giuseppe Carleo and Matthias Troyer.
\newblock ``Solving the quantum many-body problem with artificial neural networks''.
\newblock \href{https://dx.doi.org/10.1126/science.aag2302}{Science {\bf 355}, 602–606}~(2017).

\bibitem{torlai2018neural}
Giacomo Torlai, Guglielmo Mazzola, Juan Carrasquilla, Matthias Troyer, Roger Melko, and Giuseppe Carleo.
\newblock ``Neural-network quantum state tomography''.
\newblock \href{https://dx.doi.org/10.1038/s41567-018-0048-5}{Nature Physics {\bf 14}, 447--450}~(2018).

\bibitem{Carrasquilla_2020}
Juan Carrasquilla.
\newblock ``Machine learning for quantum matter''.
\newblock \href{https://dx.doi.org/10.1080/23746149.2020.1797528}{Advances in Physics: X {\bf 5}, 1797528}~(2020).

\bibitem{Pfau_2020}
David Pfau, James~S. Spencer, Alexander G. D.~G. Matthews, and W.~M.~C. Foulkes.
\newblock ``Ab initio solution of the many-electron {S}chrödinger equation with deep neural networks''.
\newblock \href{https://dx.doi.org/10.1103/physrevresearch.2.033429}{Physical Review Research\,{\bf 2}}~(2020).

\bibitem{anshu2023survey}
Anurag Anshu and Srinivasan Arunachalam.
\newblock ``A survey on the complexity of learning quantum states''~(2023).
\newblock  \href{http://arxiv.org/abs/2305.20069}{arXiv:2305.20069}.

\bibitem{nicoli2023physicsinformed}
Kim~Andrea Nicoli, Christopher~J. Anders, Lena Funcke, Tobias Hartung, Karl Jansen, Stefan Kuhn, Klaus~Robert Muller, Paolo Stornati, Pan Kessel, and Shinichi Nakajima.
\newblock ``Physics-informed {B}ayesian optimization of variational quantum circuits''~(2024).
\newblock  \href{http://arxiv.org/abs/2406.06150}{arXiv:2406.06150}.

\bibitem{lewis2023improved}
Laura Lewis, Hsin-Yuan Huang, Viet~T Tran, Sebastian Lehner, Richard Kueng, and John Preskill.
\newblock ``Improved machine learning algorithm for predicting ground state properties''.
\newblock \href{https://dx.doi.org/10.1038/s41467-024-45014-7}{Nature communications {\bf 15}, 895}~(2024).

\bibitem{onorati2023efficient}
Emilio Onorati, Cambyse Rouzé, Daniel~Stilck França, and James~D. Watson.
\newblock ``Efficient learning of ground \& thermal states within phases of matter''~(2023).
\newblock  \href{http://arxiv.org/abs/2301.12946}{arXiv:2301.12946}.

\bibitem{onorati2023provably}
Emilio Onorati, Cambyse Rouzé, Daniel~Stilck França, and James~D. Watson.
\newblock ``Provably efficient learning of phases of matter via dissipative evolutions''~(2023).
\newblock  \href{http://arxiv.org/abs/2311.07506}{arXiv:2311.07506}.

\bibitem{Huang_2022}
Hsin-Yuan Huang, Richard Kueng, Giacomo Torlai, Victor~V. Albert, and John Preskill.
\newblock ``Provably efficient machine learning for quantum many-body problems''.
\newblock \href{https://dx.doi.org/10.1126/science.abk3333}{Science\,{\bf 377}}~(2022).

\bibitem{che2024exponentially}
Yanming Che, Clemens Gneiting, and Franco Nori.
\newblock ``Exponentially improved efficient machine learning for quantum many-body states with provable guarantees''.
\newblock \href{https://dx.doi.org/10.1103/PhysRevResearch.6.033035}{Phys. Rev. Res. {\bf 6}, 033035}~(2024).

\bibitem{defenu2021longrange}
Nicol\`o Defenu, Tobias Donner, Tommaso Macr\`{\i}, Guido Pagano, Stefano Ruffo, and Andrea Trombettoni.
\newblock ``Long-range interacting quantum systems''.
\newblock \href{https://dx.doi.org/10.1103/RevModPhys.95.035002}{Rev. Mod. Phys. {\bf 95}, 035002}~(2023).

\bibitem{yan2013observation}
Bo~Yan, Steven~A Moses, Bryce Gadway, Jacob~P Covey, Kaden~RA Hazzard, Ana~Maria Rey, Deborah~S Jin, and Jun Ye.
\newblock ``Observation of dipolar spin-exchange interactions with lattice-confined polar molecules''.
\newblock \href{https://dx.doi.org/10.1038/nature12483}{Nature {\bf 501}, 521--525}~(2013).

\bibitem{Saffman_2010}
M.~Saffman, T.~G. Walker, and K.~Mølmer.
\newblock ``Quantum information with {R}ydberg atoms''.
\newblock \href{https://dx.doi.org/10.1103/revmodphys.82.2313}{Reviews of Modern Physics {\bf 82}, 2313–2363}~(2010).

\bibitem{britton2012engineered}
Joseph~W Britton, Brian~C Sawyer, Adam~C Keith, C-C~Joseph Wang, James~K Freericks, Hermann Uys, Michael~J Biercuk, and John~J Bollinger.
\newblock ``Engineered two-dimensional {I}sing interactions in a trapped-ion quantum simulator with hundreds of spins''.
\newblock \href{https://dx.doi.org/10.1038/nature10981}{Nature {\bf 484}, 489--492}~(2012).

\bibitem{RevModPhys.58.801}
K.~Binder and A.~P. Young.
\newblock ``Spin glasses: Experimental facts, theoretical concepts, and open questions''.
\newblock \href{https://dx.doi.org/10.1103/RevModPhys.58.801}{Rev. Mod. Phys. {\bf 58}, 801--976}~(1986).

\bibitem{Santagati_2024}
Raffaele Santagati, Alan Aspuru-Guzik, Ryan Babbush, Matthias Degroote, Leticia González, Elica Kyoseva, Nikolaj Moll, Markus Oppel, Robert~M. Parrish, Nicholas~C. Rubin, Michael Streif, Christofer~S. Tautermann, Horst Weiss, Nathan Wiebe, and Clemens Utschig-Utschig.
\newblock ``Drug design on quantum computers''.
\newblock \href{https://dx.doi.org/10.1038/s41567-024-02411-5}{Nature Physics {\bf 20}, 549–557}~(2024).

\bibitem{Vojta_2019}
Thomas Vojta.
\newblock ``Disorder in quantum many-body systems''.
\newblock \href{https://dx.doi.org/10.1146/annurev-conmatphys-031218-013433}{Annual Review of Condensed Matter Physics {\bf 10}, 233–252}~(2019).

\bibitem{vsmid2024accurate}
{\v{S}}t{\v{e}}p{\'a}n {\v{S}}m{\'\i}d and Roberto Bondesan.
\newblock ``Accurate learning of equivariant quantum systems from a single ground state''~(2024).
\newblock  \href{http://arxiv.org/abs/2405.12309}{arXiv:2405.12309}.

\bibitem{satorras2022en}
Victor~Garcia Satorras, Emiel Hoogeboom, and Max Welling.
\newblock ``E(n) equivariant graph neural networks''~(2022).
\newblock  \href{http://arxiv.org/abs/2102.09844}{arXiv:2102.09844}.

\bibitem{Batzner_2022}
Simon Batzner, Albert Musaelian, Lixin Sun, Mario Geiger, Jonathan~P. Mailoa, Mordechai Kornbluth, Nicola Molinari, Tess~E. Smidt, and Boris Kozinsky.
\newblock ``E(3)-equivariant graph neural networks for data-efficient and accurate interatomic potentials''.
\newblock \href{https://dx.doi.org/10.1038/s41467-022-29939-5}{Nature Communications\,{\bf 13}}~(2022).

\bibitem{fu2023forces}
Xiang Fu, Zhenghao Wu, Wujie Wang, Tian Xie, Sinan Keten, Rafael Gomez-Bombarelli, and Tommi Jaakkola.
\newblock ``Forces are not enough: Benchmark and critical evaluation for machine learning force fields with molecular simulations''~(2023).
\newblock  \href{http://arxiv.org/abs/2210.07237}{arXiv:2210.07237}.

\bibitem{wanner2024predictinggroundstateproperties}
Marc Wanner, Laura Lewis, Chiranjib Bhattacharyya, Devdatt Dubhashi, and Alexandru Gheorghiu.
\newblock ``Predicting ground state properties: Constant sample complexity and deep learning algorithms''~(2024).
\newblock  \href{http://arxiv.org/abs/2405.18489}{arXiv:2405.18489}.

\bibitem{Smid_Efficient_Learning_of_2024}
{\v{S}}t{\v{e}}p{\'a}n {\v{S}}m{\'\i}d and Roberto Bondesan~(2024).
\newblock  code:~\href{https://github.com/Quantum-AI-Lab-ICL/Efficient-Learning-Long-Range-QS}{Quantum-AI-Lab-ICL/Efficient-Learning-Long-Range-QS}.

\bibitem{Huang_2020}
Hsin-Yuan Huang, Richard Kueng, and John Preskill.
\newblock ``Predicting many properties of a quantum system from very few measurements''.
\newblock \href{https://dx.doi.org/10.1038/s41567-020-0932-7}{Nature Physics {\bf 16}, 1050–1057}~(2020).

\bibitem{hastings2010locality}
M.~B. Hastings.
\newblock ``Locality in quantum systems''~(2010).
\newblock  \href{http://arxiv.org/abs/1008.5137}{arXiv:1008.5137}.

\bibitem{Tran_2021}
Minh~C. Tran, Andrew~Y. Guo, Christopher~L. Baldwin, Adam Ehrenberg, Alexey~V. Gorshkov, and Andrew Lucas.
\newblock ``Lieb-{R}obinson light cone for power-law interactions''.
\newblock \href{https://dx.doi.org/10.1103/physrevlett.127.160401}{Physical Review Letters\,{\bf 127}}~(2021).

\bibitem{mohri2018foundations}
M.~Mohri, A.~Rostamizadeh, and A.~Talwalkar.
\newblock ``Foundations of machine learning, second edition''.
\newblock Adaptive Computation and Machine Learning series. MIT Press. ~(2018).
\newblock  url:~\href{https://mitpress.mit.edu/9780262039406/foundations-of-machine-learning/}{MIT Press}.

\bibitem{Luitz_2015}
David~J. Luitz, Nicolas Laflorencie, and Fabien Alet.
\newblock ``Many-body localization edge in the random-field {H}eisenberg chain''.
\newblock \href{https://dx.doi.org/10.1103/physrevb.91.081103}{Physical Review B\,{\bf 91}}~(2015).

\bibitem{nishimori2001statistical}
Hidetoshi Nishimori.
\newblock ``Statistical physics of spin glasses and information processing: An introduction''.
\newblock \href{https://dx.doi.org/10.1093/acprof:oso/9780198509417.001.0001}{Oxford University Press}. ~(2001).

\bibitem{Marcuzzi2017}
Matteo Marcuzzi, Ji\v{r}\'{\i} Min\'a\v{r}, Daniel Barredo, Sylvain de~L\'es\'eleuc, Henning Labuhn, Thierry Lahaye, Antoine Browaeys, Emanuele Levi, and Igor Lesanovsky.
\newblock ``Facilitation dynamics and localization phenomena in {R}ydberg lattice gases with position disorder''.
\newblock \href{https://dx.doi.org/10.1103/PhysRevLett.118.063606}{Phys. Rev. Lett. {\bf 118}, 063606}~(2017).

\bibitem{Bachmann_2011}
Sven Bachmann, Spyridon Michalakis, Bruno Nachtergaele, and Robert Sims.
\newblock ``Automorphic equivalence within gapped phases of quantum lattice systems''.
\newblock \href{https://dx.doi.org/10.1007/s00220-011-1380-0}{Communications in Mathematical Physics {\bf 309}, 835--871}~(2011).

\bibitem{Kuwahara_2020}
Tomotaka Kuwahara and Keiji Saito.
\newblock ``Strictly linear light cones in long-range interacting systems of arbitrary dimensions''.
\newblock \href{https://dx.doi.org/10.1103/physrevx.10.031010}{Physical Review X\,{\bf 10}}~(2020).

\bibitem{E_2020}
Weinan E, Chao Ma, and Lei Wu.
\newblock ``Machine learning from a continuous viewpoint, {I}''.
\newblock \href{https://dx.doi.org/10.1007/s11425-020-1773-8}{Science China Mathematics {\bf 63}, 2233–2266}~(2020).

\bibitem{Schollw_ck_2011}
Ulrich Schollwöck.
\newblock ``The density-matrix renormalization group in the age of matrix product states''.
\newblock \href{https://dx.doi.org/10.1016/j.aop.2010.09.012}{Annals of Physics {\bf 326}, 96–192}~(2011).

\bibitem{TeNPy}
Johannes Hauschild and Frank Pollmann.
\newblock ``{Efficient numerical simulations with Tensor Networks: Tensor Network Python (TeNPy)}''.
\newblock \href{https://dx.doi.org/10.21468/SciPostPhysLectNotes.5}{SciPost Phys. Lect. NotesPage~5}~(2018).

\bibitem{rndFourierFeatures}
Ali Rahimi and Benjamin Recht.
\newblock ``Random features for large-scale kernel machines''.
\newblock In J.~Platt, D.~Koller, Y.~Singer, and S.~Roweis, editors, Advances in Neural Information Processing Systems.
\newblock Volume~20.
\newblock Curran Associates, Inc.~(2007).
\newblock  url:~\href{https://proceedings.neurips.cc/paper_files/paper/2007/file/013a006f03dbc5392effeb8f18fda755-Paper.pdf}{NeurIPS proceedings}.

\bibitem{scikit-learn}
Fabian Pedregosa, Ga\"{e}l Varoquaux, Alexandre Gramfort, Vincent Michel, Bertrand Thirion, Olivier Grisel, Mathieu Blondel, Peter Prettenhofer, Ron Weiss, Vincent Dubourg, Jake Vanderplas, Alexandre Passos, David Cournapeau, Matthieu Brucher, Matthieu Perrot, and \'{E}douard Duchesnay.
\newblock ``Scikit-learn: Machine learning in {P}ython''.
\newblock \href{https://dx.doi.org/10.48550/arXiv.1201.0490}{J. Mach. Learn. Res. {\bf 12}, 2825–2830}~(2011).

\bibitem{Juhasz_2014}
R.~Juhász, I.~A. Kovács, and F.~Iglói.
\newblock ``Random transverse-field {I}sing chain with long-range interactions''.
\newblock \href{https://dx.doi.org/10.1209/0295-5075/107/47008}{Europhysics Letters {\bf 107}, 47008}~(2014).

\bibitem{hauke2010complete}
Philipp Hauke, Fernando~M Cucchietti, Alexander M{\"u}ller-Hermes, Mari-Carmen Ba{\~n}uls, J~Ignacio Cirac, and Maciej Lewenstein.
\newblock ``Complete devil's staircase and crystal--superfluid transitions in a dipolar {XXZ} spin chain: a trapped ion quantum simulation''.
\newblock \href{https://dx.doi.org/10.1088/1367-2630/12/11/113037}{New Journal of Physics {\bf 12}, 113037}~(2010).

\bibitem{Hastings_2006}
Matthew~B. Hastings and Tohru Koma.
\newblock ``Spectral gap and exponential decay of correlations''.
\newblock \href{https://dx.doi.org/10.1007/s00220-006-0030-4}{Communications in Mathematical Physics {\bf 265}, 781–804}~(2006).

\bibitem{Newman1964}
D.~J. Newman and H.~S. Shapiro.
\newblock ``Jackson's theorem in higher dimensions''.
\newblock \href{https://dx.doi.org/10.1007/978-3-0348-4131-3_20}{Pages 208--219}.
\newblock Springer Basel. Basel~(1964).

\bibitem{gaetan2009spatial}
C.~Gaetan and X.~Guyon.
\newblock ``Spatial statistics and modeling''.
\newblock \href{https://dx.doi.org/10.1007/978-0-387-92257-7}{Springer Series in Statistics}. Springer New York. ~(2009).

\bibitem{janjic2013counting}
Milan Janjic and Boris Petkovic.
\newblock ``A counting function''~(2013).
\newblock  \href{http://arxiv.org/abs/1301.4550}{arXiv:1301.4550}.

\bibitem{pinelis2020exact}
Iosif Pinelis.
\newblock ``Exact lower and upper bounds on the incomplete gamma function''~(2020).
\newblock  \href{http://arxiv.org/abs/2005.06384}{arXiv:2005.06384}.

\bibitem{Tran_2020}
Minh~C. Tran, Chi-Fang Chen, Adam Ehrenberg, Andrew~Y. Guo, Abhinav Deshpande, Yifan Hong, Zhe-Xuan Gong, Alexey~V. Gorshkov, and Andrew Lucas.
\newblock ``Hierarchy of linear light cones with long-range interactions''.
\newblock \href{https://dx.doi.org/10.1103/physrevx.10.031009}{Physical Review X\,{\bf 10}}~(2020).

\bibitem{Pirvu_2010}
B~Pirvu, V~Murg, J~I Cirac, and F~Verstraete.
\newblock ``Matrix product operator representations''.
\newblock \href{https://dx.doi.org/10.1088/1367-2630/12/2/025012}{New Journal of Physics {\bf 12}, 025012}~(2010).

\bibitem{chen2019graph}
Chi Chen, Weike Ye, Yunxing Zuo, Chen Zheng, and Shyue~Ping Ong.
\newblock ``Graph networks as a universal machine learning framework for molecules and crystals''.
\newblock \href{https://dx.doi.org/10.1021/acs.chemmater.9b01294}{Chemistry of Materials {\bf 31}, 3564--3572}~(2019).

\end{thebibliography}

\end{document}